\documentclass[12pt]{article}
\usepackage{amsmath}
\usepackage{graphicx,psfrag,epsf}
\usepackage{enumerate}
\usepackage[round]{natbib}
\usepackage{url} 
\usepackage{hyperref}
\usepackage{amsmath}
\usepackage{amsthm}
\usepackage{amssymb}
\usepackage{mathrsfs}
\newtheorem{theorem}{Theorem}[section]
\newtheorem{lemma}{Lemma}[section]
\newtheorem{assumption}{Assumption}

\newtheorem{remark}{Remark}[section]

\usepackage{array}
\usepackage{booktabs} 
\usepackage{caption}
\usepackage{tabularx}
\usepackage{float}
\usepackage{graphicx}
\usepackage{subcaption}
\usepackage{caption}
\usepackage{appendix}
\usepackage{multirow}
\usepackage[linesnumbered,ruled]{algorithm2e}

\usepackage{bm}

\usepackage{xcolor}
\usepackage{makecell}

\newcommand{\blind}{1}

\hypersetup{
  colorlinks,       
  citecolor=blue,   
  linkcolor=blue,   
  pdfborder={0 0 0} 
}

\begin{document}

\def\spacingset#1{\renewcommand{\baselinestretch}%
{#1}\small\normalsize} \spacingset{1}


\if1\blind
{
  \title{\bf Theoretical Properties of Covariate-Adaptive Randomization with a Diverging Number of Covariates}
  \author{\small{Yuhang Tao\textsuperscript{1} and Li-Xin Zhang\textsuperscript{2,3,}}\thanks{Corresponding author. Email address: stazlx@mail.zjgsu.edu.cn} \hspace{.2cm}\\
    \small{\textsuperscript{1}School of Mathematical Sciences, Zhejiang University, Hangzhou, China} \\
    \small{\textsuperscript{2}School of Statistics and Data Science, Zhejiang Gongshang University, Hangzhou, China} \\
    \small{\textsuperscript{3}Center for Data Science, Zhejiang University, Hangzhou, China} }
  \date{}
  \maketitle
} \fi

\if0\blind
{
  \bigskip
  \bigskip
  \bigskip
  \begin{center}
    {\LARGE\bf Covariate-adaptive Randomization with Missing Covariates and Statistical Inference}
\end{center}
  \medskip
} \fi

\bigskip
\begin{abstract}
Covariate-adaptive randomization procedures are widely used in clinical trials to improve covariate balance. In modern applications, experimenters often have access to many covariates, motivating the need for a theory of covariate-adaptive randomization  procedures with a diverging number of covariates. This paper studies two classes of covariate-adaptive randomization procedures, referred to as imbalance-efficient covariate-adaptive randomization and imbalance-robust covariate-adaptive randomization, when the feature dimension diverges. We derive convergence rates for the imbalance of the covariates used in randomization. For both procedures, the imbalance is of a smaller order than that under complete randomization when the feature dimension is $o(n)$, whereas it is of the same order when the feature dimension is $\Omega(n)$. For imbalance-robust covariate-adaptive randomization, we further establish the asymptotic properties of the imbalance of additional covariates and use these results to derive the asymptotic distribution of the difference-in-means estimator for the average treatment effect and construct asymptotically valid confidence intervals. Furthermore, we provide extensive numerical and empirical studies to illustrate the practical relevance of our theoretical results. 
\end{abstract}

\noindent%
\vfill

\newpage
\spacingset{1.45} 
\section{Introduction}

Covariate-adaptive randomization (CAR) procedures have been powerful tools to balance covariates between the treatment group and the control group in randomized experiments. A commonly used class of CAR methods consists of stratified CAR procedures. These methods first stratify experimental units according to baseline covariates and then assign treatment statuses to achieve balance within each stratum, marginally across strata, or both. Representative examples include stratified block randomization (\cite{ZELEN1974365}), Pocock and Simon’s minimization procedure (\citealp{taves1974minimization}, \citealp{pocock1975sequential}) and Hu and Hu's procedure (\cite{HuHu}). The theoretical properties of this class of methods have been well studied in the existing literature (\citealp{baldi2011covariate}, \citealp{HuHu}, \citealp{hu2020theory}, \citealp{hu2023multi}). Building on the established theoretical properties, statistical inference under stratified CAR procedures has also been extensively studied (\citealp{shao2010theory}, \citealp{ma2015testing}, \citealp{bugni2018inference}, \citealp{bugni2019inference}, \citealp{zhang2020quantile}, \citealp{ye2020robust}, \citealp{liu2024testing}). In addition, regression adjustment has been widely adopted to further improve estimation efficiency (\citealp{ma2022regression}, \citealp{ye2022inference}, \citealp{liu2023lasso}, \citealp{gu2023regression}, \citealp{jiang2023regression}, \citealp{wang2023model}, \citealp{bannick2025general}, \citealp{jiang2025adjustments}). For continuous covariates, stratified CAR procedures can be implemented via discretization, but this may lead to a loss of efficiency. Therefore, several procedures have been proposed to directly incorporate continuous covariates (\citealp{atkinson1982optimum}, \citealp{frane1998method}, \citealp{endo2006minimization}, \cite{stigsby2010rank}, \citealp{lin2012balancing}, \citealp{ma2013balancing}, \citealp{qin2022adaptive}). Among the aforementioned works, the theoretical properties of many CAR procedures remain incomplete. Moreover, for those with well-established theoretical guarantees, subsequent statistical inference typically relies on linear model assumptions relating the response to the covariates.

Recently, \cite{ma2024new} proposed a unified family of CAR procedures. Their approach specifies a feature map of the covariates and sequentially minimizes the Euclidean norm of the imbalance vector defined by the mapped features. By choosing different feature maps, the procedure can target the balance of different covariate features and encompasses many commonly used CAR procedures as special cases. Under this framework, covariates included in the feature map can be well balanced. However, additional covariates that are not incorporated in the feature map may remain imbalanced, and their imbalance may even be worse than that under complete randomization (CR), which can lead to a loss of efficiency in estimating the treatment effect. Additionally, \cite{liu2025properties} extended this framework to accommodate unequal allocation ratios. Under unequal allocation, covariates included in the feature map can still be balanced, but additional covariates may even exhibit the so-called ``shift problem'', where the imbalance vector is no longer centered at zero. This phenomenon may introduce bias in treatment effect estimation. To address the issues described above, \cite{zhang2026covariate} proposed a new family of CAR procedures. Under this framework, although the covariates included in the feature map may not be balanced as well as under the procedure of \cite{ma2024new}, their imbalance is still smaller than that under CR. More importantly, the imbalance of additional covariates that are not incorporated in the feature map is guaranteed not to be worse than that under complete randomization. Furthermore, under unequal allocation ratios, the procedure does not lead to the "shift problem". 

The existing theoretical results on the CAR procedures and subsequent statistical inference mainly focus on the low-dimensional setting, that is, the dimension of covariates is fixed as the sample size goes to infinity. However, in modern experiments, researchers often collect a large number of covariates. This motivates the need to develop methodology and theory for high-dimensional settings, where the covariate dimension $p$ or the feature dimension $q$ diverges as the sample size $n$ grows. To the best of our knowledge, existing work is limited and has mainly developed along two directions. The first considers stratified CAR procedures, which partition units into a fixed number of strata based on essential baseline covariates and incorporate additional high-dimensional covariates via regression adjustment to improve efficiency (\citealp{liu2023lasso}, \citealp{jiang2025adjustments}). However, such approaches rely on a pre-specified low-dimensional stratification and may suffer from information loss when important covariates are excluded. The second focuses on selecting influential covariates from a diverging set and maintaining balance with respect to these selected covariates (\citealp{zhang2022covariate}, \citealp{guo2025covariate}). Theoretical guarantees for this approach typically rely on strong sparsity assumptions, often requiring the number of relevant covariates to be bounded by a constant, which may be unrealistic in practice.

The purpose of this paper is to establish theoretical properties of the unified family of CAR procedures proposed by \cite{ma2024new} and \cite{zhang2026covariate} in high-dimensional settings, which we refer to as imbalance-efficient CAR (IE-CAR) procedures and imbalance-robust CAR (IR-CAR) procedures in this paper, respectively, according to their theoretical behavior in low-dimensional settings. Our contributions are threefold. First, we establish the dimension-dependent convergence rates for the imbalance of the covariates used in randomization. For both classes of procedures, the imbalance is of a smaller order than that under CR when $q=o(n)$, whereas it is of the same order when $q=\Omega(n)$. These results quantify how the balancing performance changes as the feature dimension increases. Second, for IR-CAR, we establish a law of large numbers and a central limit theorem for the imbalance of additional covariates under suitable dimensionality conditions. Third, based on these theoretical properties, we establish the asymptotic behavior of the difference-in-means estimator for the treatment effect under IR-CAR, and further construct the corresponding asymptotic valid confidence intervals. These results show that the confidence intervals based on the difference-in-means estimator under IR-CAR are no longer than those under CR. Our results are obtained under moment and eigenvalue conditions that allow the number of covariate features to increase with the sample size. 

The remainder of this article is organized as follows. In Section \ref{Section_framework}, we describe the general framework. Theoretical properties of the imbalance-efficient CAR procedures and the imbalance-robust CAR procedures are presented in Section \ref{Section_IE-CAR} and Section \ref{Section_IR-CAR}, respectively. In Section \ref{Section_ex}, we provide some examples under our framework. Simulation studies and a clinical trial example are presented in Sections \ref{Section_simulation} and \ref{Section_realdata}. Section \ref{Section_conclusion} summarizes the study and provides directions for future work. Technical proofs and additional simulations are provided in the Supplementary Materials.

\section{Framework and Notation}\label{Section_framework}

\subsection{General Framework}
Consider an experiment with $n$ units that enter the experiment sequentially. The experimenter aims to assign a proportion $\pi\in(0,1)$ of the units to the treatment group. Let $T_i$ denote the assignment of $i$th unit, that is, $T_i=1$ indicates that unit $i$ is assigned to the treatment group, and $T_i=0$ indicates that unit $i$ is assigned to the control group. Let $X_i=(X_{i,1},\cdots,X_{i,p})^\top \in \mathbb{R}^p$ denote the covariates adopted in the randomization procedure. We consider balancing general covariate features $\phi(X_i)$, which is defined by a feature map $\phi:\mathbb{R}^p \mapsto \mathbb{R}^{q}$ that maps $X_i$ into a $q$-dimensional feature vector $\phi(X_i)$. Here, we consider a high-dimensional setting in which both \(p\) and \(q\) grow with \(n\), i.e., \(p,q\to\infty\) as \(n\to\infty\). Then, the corresponding imbalance vector can be defined as $\Lambda_n = \sum_{i=1}^n (T_i-\pi)\phi(X_i)$ and the imbalance measure can be defined as the squared Euclidean norm of $\Lambda_n$, that is 
\begin{align*}
\text{Imb}_n = \left\| \Lambda_n \right\|^2 = \left\| \sum_{i=1}^n (T_i-\pi)\phi(X_i) \right\|^2. 
\end{align*}
Let $\text{Imb}_i^{(1)}$ and $\text{Imb}_i^{(0)}$ denote the potential imbalance measure if the $i$th unit is assigned to the treatment group and the control group, respectively. It is easy to verify that 
\begin{align*}
\text{Imb}_i^{(1)}-\text{Imb}_i^{(0)} = 2\Lambda_{i-1}^\top\phi(X_i) + (1-2\pi)\|\phi(X_i) \|^2. 
\end{align*}
Hence, if $\Lambda_{i-1}^\top \phi(X_i) >0$, the $i$th unit is assigned to treatment with probability less than $\pi$; if $\Lambda_{i-1}^\top \phi(X_i) <0$ , the assignment probability is greater than $\pi$, thereby steering the allocation toward the direction that reduces the growth of the imbalance measure. 

Based on this perspective, CAR procedures can be broadly classified into two categories according to how the allocation probability depends on the imbalance. The first category consists of procedures in which the allocation probability depends on the value of $\Lambda_{i-1}^\top \phi(X_i)$ but does not vary with the allocation step $i$, so that the influence of imbalance remains at a constant level throughout the sequential assignment process. Representative examples include \cite{ma2024new}, \cite{zhang2023asymptotic} and \cite{liu2025properties}. The second category includes procedures in which the allocation probability depends on a scaled version of the imbalance, typically of the form $(i-1)^{-\gamma}$ with $0<\gamma<1$ for the  $i$th unit. Consequently, the allocation probability is adjusted in a direction that favors covariate balance, while the influence of the imbalance on the allocation probability gradually diminishes as the assignment proceeds. Representative examples include \cite{zhang2019theory} and \cite{zhang2026covariate}. We will study the theoretical properties of these two classes of CAR procedures separately under high-dimensional settings in the subsequent sections.

\subsection{Notation}
The following notation is adopted throughout this paper. For two positive sequences $a_n$ and $b_n$, we write $a_n = O(b_n)$ if $a_n/b_n \leq C$ for some constant $C > 0$, $a_n = \Omega(b_n)$ if $a_n/b_n \geq c$ for some constant $c > 0$, $a_n = o(b_n)$ if $a_n/b_n \to 0$ and $a_n \asymp b_n$ if $c \leq a_n/b_n \leq C$ for some constants $0 < c \leq C < \infty$, i.e., $a_n = O(b_n)$ and $a_n = \Omega(b_n)$ simultaneously. For a sequence of random variables $\xi_n$, we write $\xi_n\xrightarrow{d}\xi$, if $\xi_n$ converges weakly to $\xi$ and write $\xi_n \xrightarrow{P} \xi$, if $\xi_n$ converges in probability to $\xi$. For $a$, $b\in\mathbb{R}$, let $a\vee b$ and $a\wedge b$ denote the maximum and minimum of $a$ and $b$, respectively. For a random vector $\xi$ and a random variable $\eta$, let $\Pi_{\xi} \eta$ denote the $L_2$ projection of $\eta$ onto the linear span of $\xi$ without intercept, that is, $\Pi_{\xi} \eta = \xi^\top b^*$, where $b^* = \arg\min_{b} \mathbb{E}[(\eta - \xi^\top b)^2]$.

\section{Imbalance-Efficient CAR Procedures}\label{Section_IE-CAR}
\subsection{Framework}
\cite{ma2024new} proposed a unified CAR procedure with $\pi=1/2$, which was further extended to the setting with $\pi\ne 1/2$ by \cite{liu2025properties}. In the low-dimensional setting, the imbalance process $\Lambda_n$ is positive recurrent under mild regularity conditions (e.g., a density condition as specified in Assumption 4 of \cite{ma2024new}), which implies that $\Lambda_n = O_P(1)$. Owing to its strong ability to control imbalance for the targeted covariates, we refer to this class of procedures as imbalance-efficient CAR procedures (IE-CAR). The procedure is defined as follows: \\
(i) The first unit is assigned to the treatment with probability $\pi$ and the control with probability $1-\pi$. \\
(ii) Suppose that $i - 1$ units have been assigned ($1<i\leq n$), and the results of assignments of the $i-1$ units $\{T_j\}_{j=1}^{i-1}$ and all covariates up to $i$th unit $\{X_j\}_{j=1}^i$ are observed. Then, we can accordingly calculate $\Lambda_{i-1}^\top \phi(X_i)$ and the treatment assignment for the $i$th unit is generated according to
\begin{align*}
\mathbb{P}(T_i = 1 | \Lambda_{i-1}, X_i) = \begin{cases}
\rho & \text{if }  \Lambda_{i-1}^\top \phi(X_i) < 0,  \\
2\pi - \rho & \text{if } \Lambda_{i-1}^\top \phi(X_i) > 0, \\
\pi & \text{if } \Lambda_{i-1}^\top \phi(X_i) = 0,
\end{cases}
\end{align*}
where $\pi < \rho < \min\{1,2\pi \} $. \\
(iii) Repeat step (ii) until all units are assigned.
\begin{remark}
Here we make a modification to the procedure proposed in \cite{liu2025properties}. When $p$ and $q$ are fixed, the convergence rate of the imbalance vector is the same as that of the procedure in \cite{liu2025properties}, which follows as a corollary of Theorem \ref{Thm_traditional_CAR} established later in the paper. Furthermore, the proposed procedure facilitates theoretical analysis in high-dimensional settings. Specifically, the procedure is equivalent to the unified procedure proposed by \cite{ma2024new} when $\pi=1/2$. 
\end{remark}

\begin{remark}
The biased-coin probability $\rho$ is recommended to be chosen substantially larger than $\pi$. For example, one may take $0.1\pi+0.9\min\{1,2\pi\}$, $0.2\pi+0.8\min\{1,2\pi\}$ and $0.3\pi+0.7\min\{1,2\pi\}$. When $\pi=1/2$, these choices yield  $0.95$, $0.90$, and $0.85$, respectively, which are commonly recommended values under equal allocation (\cite{HuHu}).
\end{remark}

\subsection{Theoretical Properties of Imbalance Measure}
We proceed to investigate the convergence rate of $\mathbb{E}[\|\Lambda_n\|^2]$ under IE-CAR procedures. First, we introduce the following assumptions.
\begin{assumption}\label{assump_iid}
The covariates $\{X_i \}_{i=1}^n$ are independent and identically distributed as $X$. 
\end{assumption}

\begin{assumption}\label{assump_phi1}
There exist constants $\nu \geq \iota \geq 4$ such that the  feature map $\phi(X)$ satisfies $\mathbb{E}[\| \phi(X) \|^\nu] \asymp q^{\nu/2}$ and 
\begin{align}\label{eq_assump_phi}
	\sup_{\|u\|=1} \mathbb{E} \left[ \left|u^\top \phi(X)\right|^\iota\right] \leq M_{\iota} < \infty, 
\end{align}
where $M_{\iota}$ is a constant depending only on $\iota$. 
\end{assumption}

\begin{assumption}\label{assump_boundSigma}
The nonzero eigenvalues of $\boldsymbol{\Sigma}_{\phi\phi} = \mathbb{E}[\phi(X)\phi(X)^\top]$ are bounded below by a positive constant and $\operatorname{rank}(\boldsymbol{\Sigma}_{\phi\phi}) \asymp q$.
\end{assumption}

Assumption \ref{assump_phi1} requires that the $\nu$-th moment of the norm of the feature map $\phi(X)$ scales at the same rate as that of a $q$-dimensional vector with components of constant order. It also imposes a moment constraint on all one-dimensional projections of $\phi(X)$, which is a standard condition in high-dimensional statistics. In addition, the projection moment condition in Assumption \ref{assump_phi1} implies that the largest eigenvalue of $\boldsymbol{\Sigma}_{\phi\phi}$ is uniformly bounded. Together with Assumption \ref{assump_boundSigma}, this ensures that all nonzero eigenvalues of $\boldsymbol{\Sigma}_{\phi\phi}$ are bounded above and bounded away from zero. Moreover, the condition $\operatorname{rank}(\boldsymbol{\Sigma}_{\phi\phi}) \asymp q$ ensures that the number of nonzero eigenvalues of \(\boldsymbol{\Sigma}_{\phi\phi}\) is of order $q$. Together with Assumption \ref{assump_phi1}, this yields $\mathbb{E}\left[\|\phi(X)\|^r\right] \asymp q^{r/2}$ for $1\le r\le\nu$.

\begin{remark}
For the subsequent theoretical results, we require Assumption \ref{assump_phi1} to hold for all $\nu \geq \iota\geq 4$, which allows us to obtain sharper bounds. This condition is satisfied, for example, when \(\phi(X)\) has sub-exponential tails.
\end{remark}

The following theorem specifies the theoretical properties of IE-CAR procedure in high-dimensional settings. 

\begin{theorem}\label{Thm_traditional_CAR}
Suppose that Assumptions \ref{assump_iid}-\ref{assump_boundSigma} hold. Then, \\
(i) if $q=o(n)$, $\mathbb{E}[\|\Lambda_n\|^2] = O(n^{\frac{1}{\nu-1}} q^{\frac{2\nu-3}{\nu-1}}) = o(nq)$ and $\mathbb{E}[\|\Lambda_n\|^2] = \Omega(q^2)$; \\
(ii) if $q=\Omega(n)$, $\mathbb{E}[\|\Lambda_n\|^2] \asymp nq$. \\ Specifically, if Assumption \ref{assump_phi1} holds for all $\nu \geq 4$ and $q=o(n)$, then $\mathbb{E}[\|\Lambda_n\|^2] = O(n^{\varepsilon}q^{2-\varepsilon})$ for any $\varepsilon>0$. 
\end{theorem}
Under CR, which serves as a benchmark, $\mathbb{E}[\| \Lambda_n \|^2] \asymp nq$. Theorem \ref{Thm_traditional_CAR} therefore reveals a clear transition in the balancing performance of IE-CAR as the dimension \(q\) increases. When $q=o(n)$, IE-CAR achieves $\mathbb{E}[\|\Lambda_n  \|^2] = o(nq)$, and hence provides a strictly smaller order of covariate imbalance than CR. Moreover, when Assumption \ref{assump_phi1} holds for all \(\nu\ge4\), the upper bound can be sharpened to \(O(n^\varepsilon q^{2-\varepsilon})\) for any \(\varepsilon>0\), which approaches the lower bound of order \(q^2\) as higher-order moment conditions are imposed. In contrast, when $q=\Omega(n)$, $\mathbb{E}[\|\Lambda_n\|^2] \asymp nq$, which is of the same order as under CR. Thus, the relative improvement of IE-CAR over CR vanishes once the feature dimension becomes comparable to or larger than the sample size. 

Under low-dimensional settings, the asymptotic behavior of the imbalance of additional covariates can be derived using classical tools from Markov chain theory, including the establishment of invariant measures via drift conditions, ergodic laws of large numbers, and the analysis of the associated Poisson equation. However, in high-dimensional settings, these tools become less effective, making the analysis of the asymptotic behavior of the imbalance of additional covariates considerably more challenging.  We leave this problem for future research.

\section{Imbalance-Robust CAR Procedures}\label{Section_IR-CAR}

\subsection{Framework}
\cite{zhang2026covariate} proposed IR-CAR, which ensures that the additional covariates are no more imbalanced than under CR in low-dimensional settings. This improvement comes at the cost of a larger imbalance for the specified covariates than that under IE-CAR procedures, although it still remains of order $o_P(\sqrt{n})$. Motivated by this trade-off, we refer to the procedure as imbalance robust CAR procedure (IR-CAR). The procedure is defined as follows: \\
(i) The first unit is assigned to the treatment group with probability $\pi$ and the control group with probability $1-\pi$. \\
(ii) Suppose that $i - 1$ units have been assigned to treatments ($1<i\leq n$), and the results of assignments of the $i-1$ units $\{T_j\}_{j=1}^{i-1}$ and all covariates up to $i$th unit $\{X_j\}_{j=1}^i$ are observed. Then, we can accordingly calculate $\Lambda_{i-1}^\top \phi(X_i)$ and the treatment assignment for the $i$th unit is generated according to
\begin{align*}
\mathbb{P}(T_i = 1 | \Lambda_{i-1}, X_i) = \ell\left(\frac{\Lambda_{i-1}^\top \phi(X_i)}{(i-1)^\gamma} \right),
\end{align*}
where $0 < \gamma <1$, $\ell(x):\mathbb{R} \to (0,1)$ is a non-increasing function with $\ell(0) = \pi$, $\ell^{\prime}(0)<0$ and $\ell(x)$ is twice differentiable at $x=0$. \\
(iii) Repeat step (ii) until all units are assigned.
\begin{remark}
The allocation function can take various forms, for example, $\ell(x) = \Phi(-x+z_{\pi})$, where $\Phi$ and $z_{\alpha}$ denote the standard normal distribution and its $\alpha$th quantile, respectively. \cite{zhang2026covariate} further proposed a class of allocation functions to achieve improved theoretical properties, given by
\begin{align}\label{prop_ellx}
	\ell(x) = \underline{\rho} \vee (\pi-\lambda x) \wedge \overline{\rho},
\end{align}
where $\lambda>0$ is a tuning parameter, $0<\underline{\rho}<\pi$ and $\pi<\overline{\rho}<1$ are constants ensuring that the assignment probability lies in $(0,1)$.
\end{remark}

\subsection{Theoretical Properties of Imbalance Measure} 
We now study the convergence rate of $\mathbb{E}[\|\Lambda_n\|^2]$ under IR-CAR procedures, which is specified in the following theorem. 
\begin{theorem}\label{Thm_RCAR}
Suppose that Assumptions \ref{assump_iid}-\ref{assump_boundSigma} hold. Then, \\
(i) if $q=o(n^{\gamma})$, $\mathbb{E}[\|\Lambda_n\|^2] = O(n^{\gamma}q + n^{\frac{\gamma\nu -2\gamma + 2}{2(\nu-1)}} q^{\frac{3\nu-4}{2(\nu-1)}}) = o(nq)$ and $\mathbb{E}[\|\Lambda_n\|^2] = \Omega(n^{\gamma}q)$; \\
(ii) if $q=\Omega(n^{\gamma})$ and $q=o(n)$, $\mathbb{E}[\|\Lambda_n\|^2] =  O(n^{\frac{1}{\nu-1}} q^{\frac{2\nu-3}{\nu-1}}) = o(nq)$ and $\mathbb{E}[\|\Lambda_n\|^2] = \Omega(q^2)$; \\ 
(iii) if $q = \Omega(n)$, $\mathbb{E}[\|\Lambda_n\|^2] \asymp nq$.\\
Specifically, if Assumption \ref{assump_phi1} holds with all $\nu\geq 4$, then \\
(i) if $q=O(n^{\gamma-\varepsilon_0})$ for some $\varepsilon_0>0$, $\mathbb{E}[\|\Lambda_n\|^2] \asymp n^{\gamma}q$;\\
(ii) if $q=\Omega(n^{\gamma})$ and $q=o(n)$, $\mathbb{E}[\|\Lambda_n\|^2] = O(n^{\varepsilon}q^{2-\varepsilon})$ for any $\varepsilon>0$. 
\end{theorem}
Theorem \ref{Thm_RCAR} characterizes the growth of the covariate imbalance under IR-CAR across different regimes of the feature dimension $q$. When $q=o(n^\gamma)$, the imbalance satisfies $\mathbb{E}[\|\Lambda_n\|^2] = o(nq)$ with a lower bound of order $n^\gamma q$. In particular, under sufficiently strong moment conditions and $q=O(n^{\gamma-\varepsilon_0})$, $\mathbb{E}[\|\Lambda_n\|^2] \asymp n^{\gamma}q$. When \(q=\Omega(n^\gamma)\) but \(q=o(n)\), IR-CAR still achieves \(\mathbb{E}\|\Lambda_n\|^2=o(nq)\), although the imbalance grows at a faster rate as \(q\) increases. Finally, when \(q=\Omega(n)\), $\mathbb{E}[\|\Lambda_n\|^2] \asymp nq$, which is of the same order as under CR. Thus, \(q\asymp n^\gamma\) marks a change in the growth rate of the imbalance under IR-CAR, while \(q\asymp n\) marks the regime beyond which IR-CAR no longer improves the order of covariate imbalance relative to CR. Compared with IE-CAR, IR-CAR generally allows a larger imbalance of the specified covariates in the regime \(q=o(n^\gamma)\), whereas the two procedures exhibit the same order of upper and lower bounds when \(q=\Omega(n^\gamma)\) and \(q=o(n)\). 

\begin{remark}\label{remark_weaken_condition}
In Theorems \ref{Thm_traditional_CAR} and \ref{Thm_RCAR}, we assume that Assumption \ref{assump_iid}-\ref{assump_boundSigma} hold. If $\phi(X)$ satisfies the small-ball condition, that is, there exist constants $\kappa,\delta >0$ such that for all $u\in\mathbb{R}^q$, $\mathbb{P}(|u^\top \phi(X)|\geq \kappa\| u \| ) \geq \delta$, which is widely used in high-dimensional statistics and learning theory (\citealp{lecue2017regularization}; \citealp{10.1214/17-AOS1562}). Under this condition, we can relax Assumption \ref{assump_phi1} to $\nu>2$ and $\iota =2$, while the conclusions of Theorems \ref{Thm_traditional_CAR} and \ref{Thm_RCAR} remain valid. Further discussions are provided in Remark \ref{remark_further_detail} of the Supplementary Material. 
\end{remark}

\subsection{Theoretical Properties of Additional Covariates}
Let $Z_i$ and $W_i$ denote additional covariates for the $i$th unit, representing observed or unobserved covariates that are not used in the randomization procedure. We assume that $\{(X_i^\top,Z_i,W_i)^\top\}_{i=1}^n$ are i.i.d. copies of $(X^\top,Z,W)^\top$. The following theorem establishes a novel law of large numbers and a central limit theorem for IR-CAR under high-dimensional settings. 

\begin{theorem}\label{Thm_additional covariate}
Suppose that Assumptions \ref{assump_iid}-\ref{assump_boundSigma} hold. \\
(i) If $\gamma > \frac{2}{3\nu-2} $, $q=o( n^{\gamma - \frac{ 2(1-\gamma) }{3\nu-4}})$ and $\mathbb{E}[|Z|]<\infty$, then
\begin{align}\label{Thm_no_shift}
\frac{1}{n} \sum_{i=1}^n (T_i-\pi)Z_i \xrightarrow{P} 0. 
\end{align}
(ii) If $\gamma > 1/2$, $q = o(n^{\gamma-1/2})$, $\mathbb{E}[Z^2] < \infty$, $\mathbb{E}[W] = 0$, $0<\mathbb{E}[W^2] < \infty$ and there exists a constant $c_{\phi,Z}>0$ such that $\mathbb{E}[(Z-\Pi_{\phi(X)}Z)^2] \geq c_{\phi,Z}$, then 
\begin{align}\label{Thm_asymp_normality}
\left(\frac{1}{\sqrt{n}\widetilde{\sigma}_{Z}}\sum_{i=1}^n (T_i-\pi)Z_i, \frac{1}{\sqrt{n}\sigma_W} \sum_{i=1}^n W_i \right)^\top \xrightarrow{d} \mathcal{N}(\boldsymbol{0},\boldsymbol{I}_2),
\end{align}
where $\sigma_W^2 =\mathbb{E}[W^2]$, $\widetilde{\sigma}_Z^2 = \pi(1-\pi)(\mathbb{E}[Z^2] - \mathbb{E}[(\Pi_{\phi(X)}Z)^2]) =\pi(1-\pi)\mathbb{E}[(Z-\Pi_{\phi(X)}Z)^2]$ and $\boldsymbol{I}_2$ denotes the 2×2 identity matrix. \\
(iii) Specifically, if Assumption \ref{assump_phi1} holds for every $\nu\geq \iota\geq 4$, the allocation function is given by (\ref{prop_ellx}) and $q = o(n^{\gamma-\varepsilon})$ for some $\varepsilon>0$. Then, Equation (\ref{Thm_asymp_normality}) holds.
\end{theorem}
Theorem \ref{Thm_additional covariate} gives the asymptotic behavior of the imbalance of additional covariates under IR-CAR. Under the dimensionality conditions in (i), the imbalance divided by \(n\) converges to zero in probability, including under unequal allocation. Thus, the ``shift problem" discussed in \cite{liu2025properties} for IE-CAR with continuous covariates under unequal allocation is avoided under IR-CAR. Under the conditions in (ii) and (iii), the imbalance admits an asymptotic normal approximation, with approximate asymptotic variance $\pi(1-\pi)\mathbb{E}\big[(Z-\Pi_{\phi(X)}Z)^2\big]$, which is no larger than \(\pi(1-\pi)\mathbb{E}[Z^2]\), the corresponding variance under CR. Hence, IR-CAR does not increase the asymptotic variability of the additional-covariate imbalance relative to CR.

\begin{remark}
Theorem \ref{Thm_additional covariate} shows that a larger \(\gamma\) allows the feature dimension \(q\) to grow at a faster rate. On the other hand, Theorem \ref{Thm_RCAR} implies that increasing \(\gamma\) leads to a slower rate of covariate balance. Hence, the choice of \(\gamma\) involves a trade-off between accommodating a larger feature dimension and controlling the imbalance rate.
\end{remark}

\subsection{Estimation of Treatment Effect}
For unit $i$, let $Y_i(1)$ and $Y_i(0)$ denote the potential outcome under treatment and control, respectively. Then, the observed outcome $Y_i =T_iY_i(1) + (1-T_i)Y_i(0)$ and the treatment effect is defined as $\tau = \mathbb{E}[Y_i(1)] - \mathbb{E}[Y_i(0)]$. To estimate $\tau$, we consider the estimator based on the difference of observed average outcome between the treatment and control, that is 
\begin{align*}
	\widehat{\tau} = \Bar{Y}_1 - \Bar{Y}_0 = \frac{1}{n_1} \sum_{i=1}^n T_iY_i - \frac{1}{n_0} \sum_{i=1}^n (1-T_i)Y_i,
\end{align*}
where $n_1=\sum_{i=1}^n T_i$ and $n_0=\sum_{i=1}^n (1-T_i)$ denote the numbers of units in the treatment and control groups, respectively. 

\begin{assumption}\label{assump_Yi}
$\{(X_i^\top,Y_i(1),Y_i(0))^\top\}_{i=1}^n$ are independent and identically distributed as $(X^\top,Y(1),Y(0))^\top$ with $\mathbb{E}[Y^2(a)] < \infty$, $a=0,1$. 
\end{assumption}

The following theorem establishes the asymptotic normality of the difference-in-means estimator for the treatment effect.
\begin{theorem}\label{Thm_ATE}
Suppose that Assumptions \ref{assump_iid}-\ref{assump_Yi} and that one of the following conditions is satisfied:  
(i) $\gamma > 1/2$ and $q = o(n^{\gamma-1/2})$; \\
(ii) Assumption \ref{assump_phi1} holds for every $\nu\geq \iota\geq 4$, the allocation function is given by (\ref{prop_ellx}) and $q = o(n^{\gamma-\varepsilon})$ for some $\varepsilon>0$.\\
Then, 
\begin{align*}
\sqrt{n}\left(\widehat{\tau} - \tau\right) \Big/ \sqrt{\varsigma_{\widetilde{Y}}^2(\pi) + \varsigma_{H\widetilde{Y}}^2} \xrightarrow{d} \mathcal{N}(0,1),
\end{align*} 
	where
	\begin{align*}
		\varsigma_{\widetilde{Y}}^2(\pi) = \frac{\widetilde{\sigma}_{\widetilde{Y}(1)}^2}{\pi} + \frac{\widetilde{\sigma}_{\widetilde{Y}(0)}^2}{1-\pi},  \quad \varsigma_{H\widetilde{Y}}^2 = \mathbb{E}\left[\left(\Pi_{\phi(X)}\left\{\widetilde{Y}(1) - \widetilde{Y}(0)\right\}\right)^2 \right], 
	\end{align*}
	with $\widetilde{Y}(a) = Y(a) -\mathbb{E}[Y(a)]$ and $\widetilde{\sigma}_{\widetilde{Y}(a)}^2 = \mathbb{E}[\widetilde{Y}^2(a)] - \mathbb{E}[(\Pi_{\phi(X)}\{\widetilde{Y}(a)\})^2]$ for $a=0,1$. 
\end{theorem}
Next, we consider the estimation of $\varsigma_{\widetilde{Y}}^2(\pi)$ and $\varsigma_{H\widetilde{Y}}^2$. Let $\widehat{\beta}_{\phi}(1)$ and $\widehat{\beta}_{\phi}(0)$ denote the OLS estimators of the coefficients of $\phi(X_i)$ when regressing centered outcomes $Y_i-\Bar{Y}_1$ and $Y_i-\Bar{Y}_0$ on $\phi(X_i)$ without an additional intercept in the treatment and control groups, respectively, and let $\widehat{e}_{i,1} = Y_i(1) - \Bar{Y}_1 - \phi(X_i)^\top \widehat{\beta}_{\phi}(1)$ and $\widehat{e}_{i,0} = Y_i(0) - \Bar{Y}_0 - \phi(X_i)^\top \widehat{\beta}_{\phi}(0)$ denote the corresponding residuals. Then, the estimators for $\varsigma_{\widetilde{Y}}^2(\pi)$ and $\varsigma_{H\widetilde{Y}}^2$ are defined as 
\begin{align*}
	\widehat{\varsigma}_{\widetilde{Y}}^2(\pi) =&\ \frac{1}{n_1-q} \sum_{i=1}^n \frac{T_i}{\pi}\widehat{e}_{i,1}^2 + \frac{1}{n_0-q} \sum_{i=1}^n \frac{1-T_i}{1-\pi} \widehat{e}_{i,0}^2,\\
	\text{and }\,\, \widehat{\varsigma}_{H\widetilde{Y}}^2 =&\ \frac{1}{n} \sum_{i=1}^n \left[\phi(X_i)^\top \left\{\widehat{\beta}_{\phi}(1) - \widehat{\beta}_{\phi}(0) \right\}  \right]^2. 
\end{align*}

\begin{theorem}\label{Thm_ATE_test}
Suppose that the conditions of Theorem \ref{Thm_ATE} hold. In addition, suppose that $\mathbb{E}[Y^4(a)]<\infty$ for $a=0,1$. Then, $| \widehat{\varsigma}_{\widetilde{Y}}^2(\pi) - \varsigma_{\widetilde{Y}}^2(\pi) | = o_P(1)$ and $| \widehat{\varsigma}_{H\widetilde{Y}}^2 - \varsigma_{H\widetilde{Y}}^2| = o_P(1)$. Additionally, $\sqrt{n}\left(\widehat{\tau} - \tau\right)/ \sqrt{\widehat{\varsigma}_{\widetilde{Y}}^2(\pi) + \widehat{\varsigma}_{H\widetilde{Y}}^2} \xrightarrow{d} \mathcal{N}(0,1)$. 
\end{theorem}
Based on this theorem, we can construct the nominal $(1 - \alpha) \times 100\%$ large-sample Wald confidence interval as
\begin{align*}
	\text{CI}(\alpha) = \left[ \widehat{\tau} - z_{1-\alpha/2} \sqrt{\frac{\widehat{\varsigma}_{\widetilde{Y}}^2(\pi) + \widehat{\varsigma}_{H\widetilde{Y}}^2}{n}}, \widehat{\tau} + z_{1-\alpha/2} \sqrt{\frac{\widehat{\varsigma}_{\widetilde{Y}}^2(\pi) + \widehat{\varsigma}_{H\widetilde{Y}}^2}{n}}\ \right], 
\end{align*}
which has asymptotically correct coverage probability.

\section{Examples}\label{Section_ex}
\subsection{Discrete Covariates}
When the covariates are discrete, CAR procedures are typically designed to balance strata, marginals, or both. However, when the number of covariates $p$ diverges with $n$, the number of strata grows exponentially in $p$, making it infeasible for CAR to achieve balance across all strata. Therefore, we focus on balancing marginal imbalances instead. \\
Suppose that $X_{i,j}$ have $m_j$ levels, denoted by $x_{j}^{k}$, $k=1,\cdots,m_{j}$, $j=1,\cdots,p$, where $m_j$ is fixed and does not vary with the sample size. We let $\phi(X_i)$ be the following feature map, 
\begin{align*}
\phi_{\mathrm{dis}}(X_i) = \left(\sqrt{w_{m,j}}\ \mathbb{I}\{X_{i,j} = x_j^{k_j}\} \right)_{j=1,\cdots,p;\, k_j=1,\cdots,m_j}^\top \in\mathbb{R}^{\sum_{j=1}^p m_j},
\end{align*}
where $w_{m,j}\geq 0$ denotes the corresponding weights. Let $D_n(j;k_j) = \sum_{i=1}^n (T_i-\pi)\mathbb{I}\{X_{i,j} = x_j^{k_j}\}$ denote the marginal imbalance for the $j$th covariate at level $x_j^{k_j}$ for $j=1,\cdots,p$ and $k_j=1,\cdots,m_j$. Accordingly, the imbalance measure takes the form 
\begin{align*}
\text{Imb}_n = \sum_{j=1}^p w_{m,j} \sum_{k_j=1}^{m_j} D_n^2(j;k_j). 
\end{align*}

\subsection{Continuous Covariates}
For continuous covariates, we consider balancing both the overall imbalance and the covariate means, that is, 
\begin{align*}
\phi_{\mathrm{cont},1}(X_i) = (\sqrt{w_0},\sqrt{w_1}X_i^\top)^\top\in \mathbb{R}^{p+1},
\end{align*}
where $w_0$ and $w_1$ denote the corresponding non-negative weights. Furthermore, when the number of covariates is relatively small, we may additionally balance second-order moments (i.e., covariate matrices), by taking
\begin{align*}
\phi_{\mathrm{cont},2}(X_i) = (\sqrt{w_0},\sqrt{w_1}X_i^\top,\sqrt{w_2} \operatorname{vech}(X_iX_i^\top)^\top )^\top \in\mathbb{R}^{(p^2/2 + 3p/2 + 1)},
\end{align*}
where $w_0$, $w_1$ and $w_2$ denote the corresponding non-negative weights and $\operatorname{vech}(\cdot)$ denotes the half-vectorization operator that extracts the lower triangular part (including the diagonal) of a symmetric matrix and stacks it into a column vector. In this case, the resulting procedure is closely related to the COV procedure proposed in \cite{ma2024new}. 

\subsection{Both Discrete and Continuous Covariates}
When covariates include both discrete and continuous components, we write $X_i=(X_{D,i}^\top,X_{C,i}^\top)^\top$, where $X_{D,i}$ and $X_{C,i}$ denote the discrete and continuous parts, respectively. A natural feature map is
\begin{align}\label{Eq_phi_mix}
\phi_{\mathrm{mix}}(X_i) = (\sqrt{w_0},\phi_{\mathrm{dis}}(X_{D,i})^\top,\sqrt{w_1} X_{C,i}^\top)^\top, 
\end{align}
With this feature map, the resulting CAR procedure simultaneously controls the overall imbalance, the imbalance in the means of the continuous covariates, and the marginal imbalance of the discrete covariates.

\section{Simulation}\label{Section_simulation}
\subsection{Convergence Rate of Imbalance Measure}\label{subsec_Simu_1}
We first evaluate the convergence rates of imbalance measures under different randomization procedures, including CR, IE-CAR and IR-CAR with $\gamma\in\{0.3,0.4,0.5,0.6,0.7,0.8\}$. For each exponent $d\in\{0.3,0.4,\cdots,1.1\}$, we set the covariate dimension to $p = \lfloor n^{d}\rfloor$. For $d\in\{0.3,0.4,\cdots,0.8 \}$, the sample size ranges over $n\in\{100,200,400,800,1600,3200\}$ and for $d\in\{0.9,1,1.1\}$, we set $n\in\{100,200,300,400,500,600\}$. We consider both continuous and discrete covariate settings under equal allocation ($\pi =1/2$) and unequal allocation ($\pi=2/3$). The biased coin probability for IE-CAR is set to 0.9 and the allocation function for IR-CAR is set to $\ell(x) = (2\pi-0.9) \vee (\pi-0.5x)\wedge 0.9$. The corresponding data-generating processes are described below. \\
\textbf{Model 1: Continuous covariates. } The covariates $X_i \sim \mathcal{N}_p(\boldsymbol{0}, \boldsymbol{\Sigma}_{XX})$, where $\boldsymbol{\Sigma}_{XX}$ is a Toeplitz matrix with $[\boldsymbol{\Sigma}_{XX}]_{i,j} = 2^{-|i-j|}$. The feature map is $\phi(X_i) = (1,X_i^\top)^\top$. \\
\textbf{Model 2: Discrete covariates. } The covariates $X_{i,j} = \mathbb{I}\{\Xi_{i,j} \leq z_{0.6} \}$ with $\Xi_i=(\Xi_{i,1},\cdots,\Xi_{i,p})^\top \sim \mathcal{N}_p(\boldsymbol{0}, \boldsymbol{\Sigma}_{XX})$. The feature map is $\phi(X_i) = (\mathbb{I}\{X_{i,1}=0\},\mathbb{I}\{X_{i,1}=1\},\cdots,\mathbb{I}\{X_{i,p}=0\}, \mathbb{I}\{X_{i,p}=1\})^\top$. 

To evaluate the scaling behavior, we plot \(\log(\mathrm{Imb}_n)\) against \(\log(n)\) under Model 1 and Model 2 with \(\pi=1/2\), as shown in Figures \ref{fig1} and \ref{fig2}, respectively. All results are based on 5,000 Monte Carlo replications. Additional simulation results for \(\pi=2/3\) are provided in the Supplementary Material.

Several conclusions can be drawn from Figures \ref{fig1} and \ref{fig2}. First, when the covariate dimension is relatively low, the imbalance measure under CR is the largest, while that under IE-CAR is the smallest; moreover, under IR-CAR, the imbalance measure increases with $\gamma$. As the covariate dimension increases, the gap in imbalance measures between IR-CAR and IE-CAR gradually decreases. When the dimension is sufficiently large, the imbalance measures under IR-CAR and IE-CAR become nearly identical for all values of $\gamma$, while still remaining substantially smaller than that under CR. In addition, larger values of $\gamma$ require higher covariate dimensions for the imbalance measures of IE-CAR and IR-CAR to become close. Finally, the curves in Figures \ref{fig1} and \ref{fig2} are nearly linear on the log–log scale, suggesting that $\text{Imb}_n$ follows an approximate polynomial rate in $n$. The slope of each curve thus serves as an empirical estimate of the rate exponent. To make this connection explicit, we further plot the estimated slopes as a function of $d$. The horizontal axis represents $d$, and the vertical axis corresponds to the estimated slope from the log–log regression. Each curve in this plot represents a different randomization procedure, allowing for a direct comparison of their scaling behaviors across different dimensional regimes. The results are shown in Figure \ref{fig3}.

From Figure \ref{fig3}, we can clearly observe that when $d<1$, the slope in the log–log plot under IE-CAR is approximately $2d$, suggesting a polynomial growth rate of order $2d$. The results indicate that $\mathbb{E}[\|\Lambda_n\|^2]$ is approximately of order $n^{2d}$ when $p\asymp n^d$. For IR-CAR, the slope is approximately $d+\gamma$ when $d\leq \gamma$, and becomes approximately $2d$ when $\gamma<d<1$. Under CR, the slope remains approximately $d+1$ across all values of $d$. When $d \geq 1$, the growth rate of $\text{Imb}_n$ under all randomization procedures is approximately $d+1$, indicating that CAR procedures no longer offer an advantage over CR in terms of growth rate. Overall, the convergence rates observed in Figure \ref{fig3} are highly consistent with our theoretical predictions, providing strong empirical support for our theoretical results.

\begin{figure}[htbp]
    \centering
    \includegraphics[width=\textwidth]{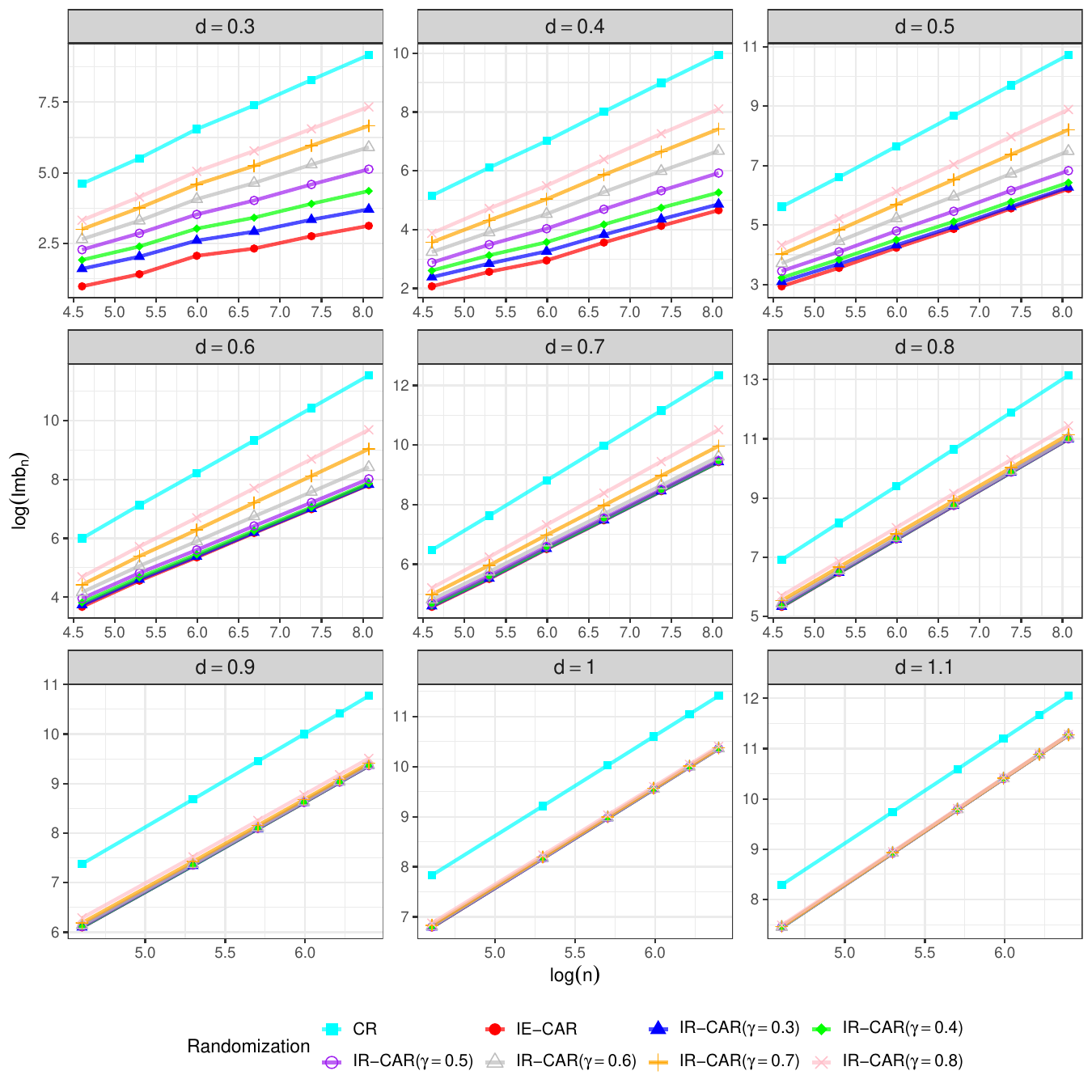}
    \caption{Log–log plot of $\text{Imb}_n$ versus $n$ with continuous covariates under $\pi=1/2$.}
    \label{fig1}
\end{figure}

\begin{figure}[htbp]
    \centering
    \includegraphics[width=\textwidth]{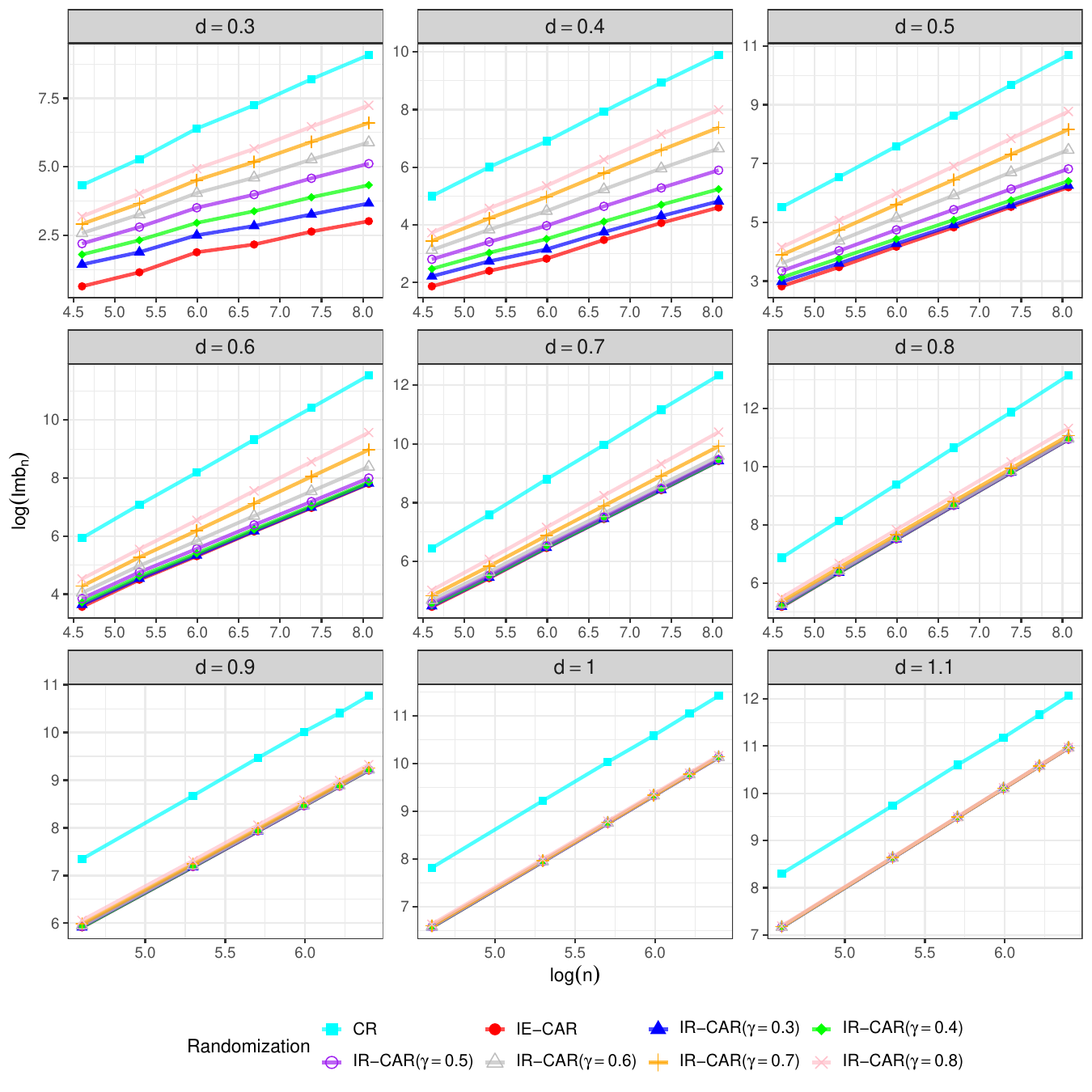}
    \caption{Log–log plot of $\text{Imb}_n$ versus $n$ with discrete covariates under $\pi=1/2$.}
    \label{fig2}
\end{figure}

\begin{figure}[htbp]
    \centering
    \includegraphics[width=\textwidth]{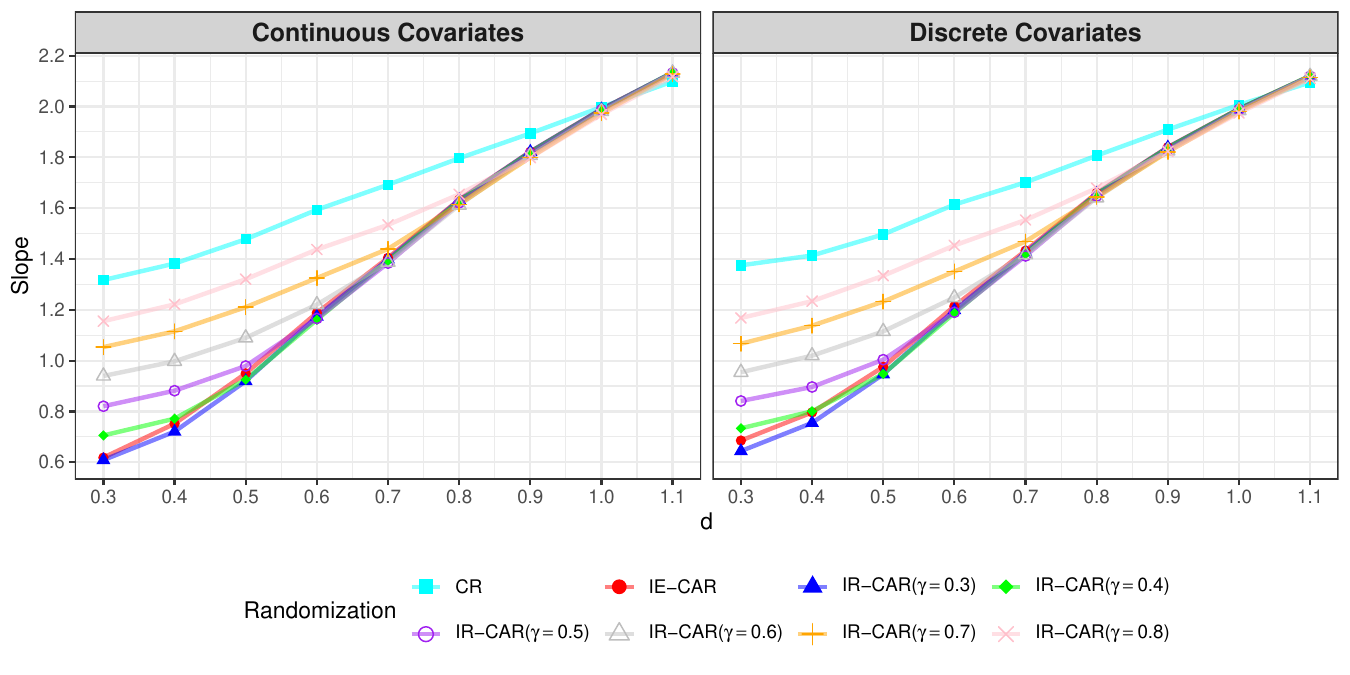}
    \caption{Estimated slopes from log–log plots under various randomization procedures with $\pi=1/2$.}
    \label{fig3}
\end{figure}

\subsection{Estimation of Treatment Effect}
To confirm and complement our theory, in this section, we examine the finite-sample performance of the difference-in-means estimator, the variance estimator, and the corresponding 95\% confidence intervals under IR-CAR. We consider both discrete and continuous covariate settings, with the data-generating processes specified as follows.\\
\textbf{Model 1: Continuous covariates. } The covariates are generated as in Model 1 of Section \ref{subsec_Simu_1} and the potential outcomes are generated by
\begin{align*}
Y_i(1) =&\ \frac{1}{p}\left[\left(\sum_{j=1}^p X_{i,j} \right)^2 - \boldsymbol{1}_p^\top \boldsymbol{\Sigma}_{XX} \boldsymbol{1}_p\right] + \left[1 + \frac{1}{p} \left(\sum_{j=1}^p X_{i,j}\right)^2 \right] \epsilon_i(1), \\
Y_i(0) =&\ \frac{1}{\sqrt{p}} \sum_{j=1}^p \left( e^{X_{i,j}} - \sqrt{e}\right) + \frac{1}{p} \sum_{j=1}^p X_{i,j}^2  \epsilon_i(0),
\end{align*}
where $\epsilon_i(1),\epsilon_i(0)\sim \mathcal{N}(0,4)$. \\
\textbf{Model 2: Discrete covariates. } The covariates are generated as in Model 2 of Section \ref{subsec_Simu_1} and the potential outcomes are generated by
\begin{align*}
Y_i(1) = Y_i(0) = \frac{1}{\sqrt{p}} \sum_{j=1}^p \left(X_{i,j} + 2 e^{X_{i,j}}\right) + \left(2+\frac{1}{p}\sum_{j=1}^p X_{i,j}\right) \epsilon_i,
\end{align*}
where $\epsilon_i \sim \mathcal{N}(0,4)$. 

We consider the sample size $n=500$ and let $p=\lfloor n^d\rfloor$ with $d\in\{0.25,0.35,0.45,0.55,0.65\}$ (i.e., $p \in\{ 4,8,16, 30, 56\}$). Due to theoretical constraints, for each choice of $d$,  we only consider IR-CAR procedures with $\gamma>d$, where $\gamma\in\{0.3,0.4,\cdots,0.8\}$. The allocation function is set to $\ell(x) = (2\pi-0.9) \vee (\pi-0.5x)\wedge 0.9$. 

Table \ref{ATE_est} presents the bias, standard deviation (SD), standard error (SE) estimators, and empirical coverage probabilities (CP) of the 95\% confidence intervals for $\widehat{\tau}$ for Model 1 and Model 2 with equal and unequal allocation under IR-CAR. Each scenario is evaluated using 5,000 Monte Carlo replications.  

Under all settings, $\widehat{\tau}$ exhibits negligible bias. The variance estimators, or equivalently the average estimated standard deviations, are generally accurate, leading to adequate coverage probabilities of the confidence intervals. When $p$ is small, the variance estimators are highly accurate; when $p$ is large, they tend to exhibit some upward bias, but the coverage probabilities of the confidence intervals are still well maintained at the nominal 95\% level.

\begin{table}[htbp]
	\centering
	\footnotesize 
	\setlength{\tabcolsep}{2pt} 
	\renewcommand{\arraystretch}{1}
	\caption{Simulated biases, SDs, SEs, and CPs for $\widehat{\tau}$ under IR-CAR procedures}
	\begin{tabular}{*{18}{c}}
		\toprule
		&  & \multicolumn{8}{c}{Continuous Covariates} & \multicolumn{8}{c}{Discrete Covariates}\\
		\cmidrule(lr{0pt}){3-10} \cmidrule(lr{0pt}){11-18}
		&  & \multicolumn{4}{c}{$\pi = 1/2$} & \multicolumn{4}{c}{$\pi = 2/3$} & \multicolumn{4}{c}{$\pi = 1/2$} & \multicolumn{4}{c}{$\pi = 2/3$}\\ 
		\cmidrule(lr{0pt}){3-6} \cmidrule(lr{0pt}){7-10} \cmidrule(lr{0pt}){11-14} \cmidrule(lr{0pt}){15-18}
		$p$ & $\gamma$ & Bias & SD & SE & CP & Bias & SD & SE & CP & Bias & SD & SE & CP & Bias & SD & SE & CP \\
		\midrule
		4  & 0.3 &-0.006&0.598&0.603&0.955&-0.001&0.546&0.552&0.956&-0.003&0.467&0.470&0.953&-0.006&0.489&0.499&0.953\\
		& 0.4 &-0.002&0.612&0.602&0.944&0.002&0.561&0.551&0.948&-0.003&0.474&0.470&0.947&0.004&0.502&0.499&0.949\\
		& 0.5 &-0.003&0.605&0.603&0.955&-0.015&0.552&0.552&0.953&-0.012&0.470&0.470&0.950&0.003&0.498&0.499&0.951\\
		& 0.6 &0.003&0.607&0.603&0.948&-0.007&0.554&0.551&0.950&-0.008&0.483&0.470&0.941&0.003&0.499&0.499&0.945\\
		& 0.7 &0.004&0.605&0.603&0.949&0.002&0.561&0.551&0.945&-0.010&0.482&0.470&0.944&0.001&0.513&0.499&0.942\\
		& 0.8 &0.000&0.604&0.602&0.952&-0.003&0.562&0.552&0.946&-0.009&0.480&0.470&0.944&-0.009&0.526&0.498&0.939\\
		\midrule
		8  & 0.4 &-0.006&0.699&0.697&0.951&-0.002&0.631&0.631&0.953&0.000&0.473&0.471&0.945&-0.002&0.486&0.499&0.952\\
		& 0.5 &-0.003&0.689&0.700&0.952&-0.009&0.628&0.630&0.951&-0.001&0.463&0.471&0.954&0.002&0.500&0.499&0.950\\
		& 0.6 &0.001&0.714&0.698&0.945&0.012&0.636&0.630&0.951&0.004&0.475&0.471&0.949&-0.002&0.503&0.499&0.947\\
		& 0.7 &-0.019&0.704&0.697&0.948&-0.007&0.635&0.629&0.950&0.015&0.480&0.471&0.943&-0.005&0.502&0.499&0.948\\
		& 0.8 &-0.004&0.706&0.699&0.947&-0.003&0.631&0.630&0.950&-0.001&0.479&0.471&0.945&-0.002&0.508&0.500&0.946\\
		\midrule
		16 & 0.5 &-0.007&0.763&0.759&0.946&-0.017&0.681&0.679&0.946&-0.006&0.478&0.473&0.946&0.002&0.497&0.502&0.950\\
		& 0.6 &-0.005&0.762&0.759&0.947&-0.010&0.683&0.680&0.948&-0.002&0.473&0.473&0.947&0.006&0.495&0.502&0.958\\
		& 0.7 &-0.005&0.758&0.756&0.951&-0.008&0.678&0.679&0.953&-0.009&0.482&0.473&0.947&0.009&0.510&0.502&0.948\\
		& 0.8 &-0.003&0.761&0.759&0.946&-0.014&0.677&0.679&0.951&0.001&0.487&0.473&0.942&-0.006&0.513&0.502&0.943\\
		\midrule
		30 & 0.6 &0.000&0.783&0.797&0.951&0.002&0.695&0.711&0.957&-0.003&0.470&0.480&0.954&-0.007&0.505&0.510&0.951\\
		& 0.7 &0.004&0.778&0.796&0.957&-0.002&0.697&0.712&0.955&0.003&0.478&0.480&0.951&-0.009&0.503&0.510&0.954\\
		& 0.8 &-0.004&0.789&0.795&0.951&-0.010&0.700&0.712&0.954&0.006&0.478&0.480&0.953&-0.002&0.513&0.510&0.948\\
		\midrule
		56 & 0.7 &-0.006&0.791&0.832&0.959&0.004&0.704&0.742&0.958&-0.010&0.475&0.494&0.958&-0.005&0.514&0.527&0.952\\
		& 0.8 &-0.006&0.793&0.832&0.958&-0.012&0.711&0.742&0.960&0.008&0.490&0.495&0.951&0.008&0.515&0.527&0.954\\
		\bottomrule
	\end{tabular}
	
	\label{ATE_est}
\end{table}

\section{Real Data Example}\label{Section_realdata}

In this section, we use the primary biliary cirrhosis (PBC) dataset from the Mayo Clinic as a real data example to compare the covariate balance and estimation precision under CR and various CAR procedures with both discrete and continuous covariates. This dataset includes 424 patients, encompassing both individuals who were actively enrolled in a randomized clinical trial and additional cases who consented to provide basic measurements. The original study was designed to investigate whether D-penicillamine is effective in improving clinical outcomes for patients with PBC, a chronic autoimmune liver disease characterized by progressive damage to the intrahepatic bile ducts (\cite{murtaugh1994primary}).

In our analysis, we restrict attention to the 312 patients who participated in the randomized trial, for whom treatment assignment, outcomes, and baseline covariates are systematically recorded. The outcome of interest is the square root of the survival time, where survival time is defined as the number of days from registration to death, transplantation, or the end of follow-up. The dataset contains 16 baseline covariates, including 7 discrete variables and 9 continuous variables. Missing values in these covariates are handled via median imputation, and continuous covariates are normalized accordingly. Detailed descriptions of all covariates are provided in the Supplementary Material. 

We redesign the treatment assignment under CR, IE-CAR, and IR-CAR with $\gamma\in\{0.4,0.5,0.6,0.7\}$ to obtain the treatment assignments for both the equal ($\pi=1/2$) and unequal ($\pi=2/3$) allocations. The feature map $\phi(X)$ is specified according to Equation (\ref{Eq_phi_mix}) with equal weights, and all other parameters are chosen in accordance with Section \ref{Section_simulation}. For each randomization scheme, we repeat the randomization procedure 5,000 times and report the average imbalance measure $\text{Imb}_n$. 

In addition, to evaluate treatment effect estimation after redesign, we construct a synthetic dataset by imputing the missing potential outcomes using a causal forest estimator trained on the observed data. Based on this synthetic dataset, we estimate the average treatment effect (ATE) under each assignment mechanism and construct corresponding 95\% confidence intervals. For IE-CAR and IR-CAR with $\gamma=0.4$, due to the lack of theoretical justification for valid variance estimation, we do not report confidence intervals. The results are presented in Table \ref{realdata}. 

\begin{table}[htbp]
\centering
\setlength{\tabcolsep}{3pt} 
\caption{Means of imbalance measure, ATE Estimates and 95\% confidence intervals (95\% CI) for $\widehat{\tau}$ under various randomization procedures for the synthetic PBC dataset. }
\begin{tabular}{*{7}{c}}
\toprule
& \multicolumn{3}{c}{$\pi = 1/2$} & \multicolumn{3}{c}{$\pi = 2/3$}\\
\cmidrule(lr{0pt}){2-4} \cmidrule(lr{0pt}){5-7}
Randomization & $\text{Imb}_n$ & Estimate & 95\% CI  & $\text{Imb}_n$ & Estimate & 95\% CI\\
\midrule 
CR  & 308.53 & -1.47 & (-5.76, 2.81) & 272.71 & -1.60 & (-5.97, 2.77)   \\
IE-CAR &  12.66 & 0.51 & \textemdash & 22.91 & -0.83 & \textemdash  \\
IR-CAR($\gamma=0.4$) & 29.53 & -0.84 & \textemdash & 30.98 & -0.50 & \textemdash \\
IR-CAR($\gamma=0.5$) & 41.15 & 0.11 & (-2.26, 2.49) & 39.02 & 0.58 & (-1.87, 3.03)  \\
IR-CAR($\gamma=0.6$) & 56.91 & 0.42 & (-1.95, 2.80) & 51.80 & -0.92 & (-3.39, 1.56)  \\
IR-CAR($\gamma=0.7$) & 75.86 & 0.76 & (-1.64, 3.15) & 68.35 & 1.30 & (-1.21, 3.81) \\
\bottomrule
\end{tabular}
\label{realdata}
\end{table}

As shown in Table \ref{realdata}, the imbalance measure under CR is the largest, while that under IE-CAR is the smallest. In addition, under IR-CAR, the imbalance measure increases with $\gamma$. Moreover, the confidence intervals constructed based on $\widehat{\tau}$ under IR-CAR are noticeably shorter than those under CR, demonstrating improved efficiency in ATE estimation. Across different randomization procedures, all confidence intervals cover zero, suggesting no statistically significant treatment effect. This finding is consistent with existing evidence based on the PBC data from the Mayo Clinic trial, where D-penicillamine was not found to provide a clear survival benefit.

\section{Conclusion}\label{Section_conclusion}

In this paper, we establish the theoretical properties of IE-CAR and IR-CAR, proposed by \cite{ma2024new} and \cite{zhang2026covariate}, respectively, under high-dimensional settings. The convergence rates of imbalance vectors are evaluated both theoretically and by numerical studies. In addition, we establish the asymptotic behavior of the imbalance of additional covariates and of the difference-in-means estimator for the ATE under IR-CAR. Table \ref{conclusion} provides a summary of the properties of the imbalance measure and additional covariates under CR, IE-CAR and IR-CAR. 

\begin{table}[h]
	\centering
	\footnotesize 
	\setlength{\tabcolsep}{3pt} 
	\renewcommand{\arraystretch}{1.5}
	\caption{Summary of theoretical properties for various randomization procedures. }
	\begin{tabular}{|c|c|c|c|c|}
	\hline
	\multirow{2}{*}{Procedure}	&  \multicolumn{3}{c|}{Imbalance Bounds} & Additional-Covariate  \\
		\cline{2-4} 
	 	& $q=O(n^{\gamma})$  & $q=\Omega(n^{\gamma})$ and $o(n)$ & $q=\Omega(n)$ &  Imbalance   \\
		\hline
	CR & \multicolumn{3}{c|}{$\asymp nq$} & Asymptotic normal \\
	\hline
	IE-CAR & \multicolumn{2}{c|}{$O\left(n^{\frac{1}{\nu-1}} q^{\frac{2\nu-3}{\nu-1}}\right)$ and $\Omega(q^2)$}  & $\asymp nq$ & Not established  \\
	\hline
	\multirow{2}{*}{IR-CAR} & $O\left(n^{\gamma}q + n^{\frac{\gamma\nu -2\gamma + 2}{2(\nu-1)}} q^{\frac{3\nu-4}{2(\nu-1)}}\right)$  & $O\left(n^{\frac{1}{\nu-1}} q^{\frac{2\nu-3}{\nu-1}}\right)$  & \multirow{2}{*}{$\asymp nq$} & Asymptotic normal  \\
	& $\Omega(n^{\gamma}q)$ & $\Omega(q^2)$ & & Smaller asymptotic variance \\
	\hline
	\end{tabular}
\label{conclusion}
\end{table}

Our results can be extended in several directions. First, multi-armed experiments are ubiquitous in clinical trials, which motivates the study of theoretical properties of multi-treatment CAR procedures in high-dimensional settings. Second, when $q=\Omega(n)$, the CAR procedures covered by our framework no longer improve the order of the imbalance measure relative to CR. Whether there exists a procedure that can achieve satisfactory performance in this regime remains an open problem. Third, \cite{fang2026general} proposed a general non-Markovian CAR framework to address the shift problem under unequal allocation, which was further extended to covariate-adjusted response-adaptive settings by \cite{fang2026cbara}. Establishing theoretical properties for these two procedures under high-dimensional settings remains an important research problem.

\section*{Acknowledgments}
This work was supported by grants from National Key R\&D Program of China (No. 2024YFA1013502), NSF of China (Grant No. U23A2064) and the Summit Advancement Disciplines of Zhejiang Province (Zhejiang Gongshang University - Statistics).

\nocite{*}
\bibliographystyle{chicago} 
\bibliography{ref.bib}

@article{ma2024new,
  title={A new and unified family of covariate adaptive randomization procedures and their properties},
  author={Ma, Wei and Li, Ping and Zhang, Li-Xin and Hu, Feifang},
  journal={Journal of the American Statistical Association},
  volume={119},
  number={545},
  pages={151--162},
  year={2024},
  publisher={Taylor \& Francis}
}

@article{liu2025properties,
  title={The properties of covariate-adaptive randomization procedures with possibly unequal allocation ratio},
  author={Liu, Xiao and Hu, Feifang and Ma, Wei},
  journal={The Annals of Applied Statistics},
  volume={19},
  number={2},
  pages={907--925},
  year={2025},
  publisher={Institute of Mathematical Statistics}
}

@article{zhang2023asymptotic,
  title={Asymptotic Properties of Multi-Treatment Covariate Adaptive Randomization Procedures for Balancing Observed and Unobserved Covariates},
  author={Zhang, Li-Xin},
  journal={arXiv preprint arXiv:2305.13842},
  year={2023}
}

@article{HuHu,
author = {Yanqing Hu and Feifang Hu},
title = {{Asymptotic properties of covariate-adaptive randomization}},
volume = {40},
journal = {The Annals of Statistics},
number = {3},
publisher = {Institute of Mathematical Statistics},
pages = {1794 -- 1815},
year = {2012},
}

@article{taves1974minimization,
  title={Minimization: a new method of assigning patients to treatment and control groups},
  author={Taves, Donald R},
  journal={Clinical Pharmacology \& Therapeutics},
  volume={15},
  number={5},
  pages={443--453},
  year={1974},
  publisher={Wiley Online Library}
}

@article{hu2020theory,
  title={On the theory of covariate-adaptive designs},
  author={Hu, Feifang and Zhang, Li-Xin},
  journal={arXiv preprint arXiv:2004.02994},
  year={2020}
}

@article{ZELEN1974365,
title = {The randomization and stratification of patients to clinical trials},
journal = {Journal of Chronic Diseases},
volume = {27},
number = {7},
pages = {365-375},
year = {1974},
author = {M. Zelen}
}

@article{pocock1975sequential,
  title={Sequential treatment assignment with balancing for prognostic factors in the controlled clinical trial},
  author={Pocock, Stuart J and Simon, Richard},
  journal={Biometrics},
  pages={103--115},
  year={1975},
  publisher={JSTOR}
}

@article{hu2023multi,
  title={Multi-arm covariate-adaptive randomization},
  author={Hu, Feifang and Ye, Xiaoqing and Zhang, Li-Xin},
  journal={Science China Mathematics},
  volume={66},
  number={1},
  pages={163--190},
  year={2023},
  publisher={Springer}
}

@article{baldi2011covariate,
  title={The covariate-adaptive biased coin design for balancing clinical trials in the presence of prognostic factors},
  author={Baldi Antognini, Alessandro and Zagoraiou, Maroussa},
  journal={Biometrika},
  volume={98},
  number={3},
  pages={519--535},
  year={2011},
  publisher={Oxford University Press}
}

@article{10.1214/17-AOS1562,
author = {Guillaume Lecu{\'e} and Shahar Mendelson},
title = {{Regularization and the small-ball method I: Sparse recovery}},
volume = {46},
journal = {The Annals of Statistics},
number = {2},
publisher = {Institute of Mathematical Statistics},
pages = {611 -- 641},
year = {2018},
doi = {10.1214/17-AOS1562},
URL = {https://doi.org/10.1214/17-AOS1562}
}

@article{lecue2017regularization,
  title={Regularization and the small-ball method II: complexity dependent error rates},
  author={Lecu{\'e}, Guillaume and Mendelson, Shahar},
  journal={Journal of Machine Learning Research},
  volume={18},
  number={146},
  pages={1--48},
  year={2017}
}

@article{zhang2026covariate,
  title={Covariate-Adaptive Randomization in Clinical Trials Without Inflated Variances},
  author={Zhang, Li-Xin},
  journal={arXiv preprint arXiv:2602.10760},
  year={2026}
}

@article{guo2025covariate,
  title={Covariate selection for optimizing balance with an innovative adaptive randomization approach},
  author={Guo, Ziqing and Liu, Yang and Xia, Lucy},
  journal={Statistical Methods in Medical Research},
  volume={34},
  number={9},
  pages={1751--1779},
  year={2025},
  publisher={SAGE Publications Sage UK: London, England}
}

@article{zhang2022covariate,
  title={Covariate-adaptive randomization with variable selection in clinical trials},
  author={Zhang, Hao and Hu, Feifang and Yin, Jianxin},
  journal={Stat},
  volume={11},
  number={1},
  pages={e461},
  year={2022},
  publisher={Wiley Online Library}
}

@article{jiang2025adjustments,
  title={Adjustments with many regressors under covariate-adaptive randomizations},
  author={Jiang, Liang and Li, Liyao and Miao, Ke and Zhang, Yichong},
  journal={Journal of Econometrics},
  volume={249},
  pages={105991},
  year={2025},
  publisher={Elsevier}
}

@article{liu2023lasso,
  title={Lasso-adjusted treatment effect estimation under covariate-adaptive randomization},
  author={Liu, Hanzhong and Tu, Fuyi and Ma, Wei},
  journal={Biometrika},
  volume={110},
  number={2},
  pages={431--447},
  year={2023},
  publisher={Oxford University Press}
}

@book{wainwright2019high,
  title={High-dimensional statistics: A non-asymptotic viewpoint},
  author={Wainwright, Martin J},
  volume={48},
  year={2019},
  publisher={Cambridge university press}
}

@article{lei2021regression,
  title={Regression adjustment in completely randomized experiments with a diverging number of covariates},
  author={Lei, Lihua and Ding, Peng},
  journal={Biometrika},
  volume={108},
  number={4},
  pages={815--828},
  year={2021},
  publisher={Oxford University Press}
}

@article{bugni2018inference,
  title={Inference under covariate-adaptive randomization},
  author={Bugni, Federico A and Canay, Ivan A and Shaikh, Azeem M},
  journal={Journal of the American Statistical Association},
  volume={113},
  number={524},
  pages={1784--1796},
  year={2018},
  publisher={Taylor \& Francis}
}

@article{ma2022regression,
  title={Regression analysis for covariate-adaptive randomization: a robust and efficient inference perspective},
  author={Ma, Wei and Tu, Fuyi and Liu, Hanzhong},
  journal={Statistics in Medicine},
  volume={41},
  number={29},
  pages={5645--5661},
  year={2022},
  publisher={Wiley Online Library}
}

@article{bugni2019inference,
  title={Inference under covariate-adaptive randomization with multiple treatments},
  author={Bugni, Federico A and Canay, Ivan A and Shaikh, Azeem M},
  journal={Quantitative Economics},
  volume={10},
  number={4},
  pages={1747--1785},
  year={2019},
  publisher={Wiley Online Library}
}

@article{ye2022inference,
  title={Inference on the average treatment effect under minimization and other covariate-adaptive randomization methods},
  author={Ye, Ting and Yi, Yanyao and Shao, Jun},
  journal={Biometrika},
  volume={109},
  number={1},
  pages={33--47},
  year={2022},
  publisher={Oxford University Press}
}

@article{gu2023regression,
  title={Regression-based multiple treatment effect estimation under covariate-adaptive randomization},
  author={Gu, Yujia and Liu, Hanzhong and Ma, Wei},
  journal={Biometrics},
  volume={79},
  number={4},
  pages={2869--2880},
  year={2023},
  publisher={Wiley Online Library}
}

@article{ma2015testing,
  title={Testing hypotheses of covariate-adaptive randomized clinical trials},
  author={Ma, Wei and Hu, Feifang and Zhang, Lixin},
  journal={Journal of the American Statistical Association},
  volume={110},
  number={510},
  pages={669--680},
  year={2015},
  publisher={Taylor \& Francis}
}

@article{shao2010theory,
  title={A theory for testing hypotheses under covariate-adaptive randomization},
  author={Shao, Jun and Yu, Xinxin and Zhong, Bob},
  journal={Biometrika},
  volume={97},
  number={2},
  pages={347--360},
  year={2010},
  publisher={Oxford University Press}
}

@article{zhang2020quantile,
  title={Quantile treatment effects and bootstrap inference under covariate-adaptive randomization},
  author={Zhang, Yichong and Zheng, Xin},
  journal={Quantitative Economics},
  volume={11},
  number={3},
  pages={957--982},
  year={2020},
  publisher={Wiley Online Library}
}

@article{liu2024testing,
  title={Testing heterogeneous treatment effect with quantile regression under covariate-adaptive randomization},
  author={Liu, Yang and Xia, Lucy and Hu, Feifang},
  journal={Journal of Econometrics},
  volume={249},
  pages={105808},
  year={2025},
  publisher={Elsevier}
}

@article{jiang2023regression,
  title={Regression-adjusted estimation of quantile treatment effects under covariate-adaptive randomizations},
  author={Jiang, Liang and Phillips, Peter CB and Tao, Yubo and Zhang, Yichong},
  journal={Journal of Econometrics},
  volume={234},
  number={2},
  pages={758--776},
  year={2023},
  publisher={Elsevier}
}

@article{ye2020robust,
  title={Robust tests for treatment effect in survival analysis under covariate-adaptive randomization},
  author={Ye, Ting and Shao, Jun},
  journal={Journal of the Royal Statistical Society Series B: Statistical Methodology},
  volume={82},
  number={5},
  pages={1301--1323},
  year={2020},
  publisher={Oxford University Press}
}

@article{bannick2025general,
  title={A general form of covariate adjustment in clinical trials under covariate-adaptive randomization},
  author={Bannick, Marlena S and Shao, Jun and Liu, Jingyi and Du, Yu and Yi, Yanyao and Ye, Ting},
  journal={Biometrika},
  volume={112},
  number={3},
  pages={asaf029},
  year={2025},
  publisher={Oxford University Press}
}

@article{wang2023model,
  title={Model-robust inference for clinical trials that improve precision by stratified randomization and covariate adjustment},
  author={Wang, Bingkai and Susukida, Ryoko and Mojtabai, Ramin and Amin-Esmaeili, Masoumeh and Rosenblum, Michael},
  journal={Journal of the American Statistical Association},
  volume={118},
  number={542},
  pages={1152--1163},
  year={2023},
  publisher={Taylor \& Francis}
}

@article{frane1998method,
  title={A method of biased coin randomization, its implementation, and its validation},
  author={Frane, James W},
  journal={Drug information journal: DIJ/Drug Information Association},
  volume={32},
  number={2},
  pages={423--432},
  year={1998},
  publisher={Springer}
}

@article{endo2006minimization,
  title={Minimization method for balancing continuous prognostic variables between treatment and control groups using Kullback-Leibler divergence},
  author={Endo, Akira and Nagatani, Fumio and Hamada, Chikuma and Yoshimura, Isao},
  journal={Contemporary Clinical Trials},
  volume={27},
  number={5},
  pages={420--431},
  year={2006},
  publisher={Elsevier}
}

@article{stigsby2010rank,
  title={Rank-minimization for balanced assignment of subjects in clinical trials},
  author={Stigsby, Bjarne and Taves, Donald R},
  journal={Contemporary clinical trials},
  volume={31},
  number={2},
  pages={147--150},
  year={2010},
  publisher={Elsevier}
}

@article{atkinson1982optimum,
  title={Optimum biased coin designs for sequential clinical trials with prognostic factors},
  author={Atkinson, Anthony C},
  journal={Biometrika},
  volume={69},
  number={1},
  pages={61--67},
  year={1982},
  publisher={Oxford University Press}
}

@article{lin2012balancing,
  title={Balancing continuous and categorical baseline covariates in sequential clinical trials using the area between empirical cumulative distribution functions},
  author={Lin, Yunzhi and Su, Zheng},
  journal={Statistics in medicine},
  volume={31},
  number={18},
  pages={1961--1971},
  year={2012},
  publisher={Wiley Online Library}
}

@article{ma2013balancing,
  title={Balancing continuous covariates based on kernel densities},
  author={Ma, Zhenjun and Hu, Feifang},
  journal={Contemporary clinical trials},
  volume={34},
  number={2},
  pages={262--269},
  year={2013},
  publisher={Elsevier}
}

@article{qin2022adaptive,
  title={Adaptive randomization via mahalanobis distance},
  author={Qin, Yichen and Li, Yang and Ma, Wei and Yang, Haoyu and Hu, Feifang},
  journal={Statistica Sinica},
  volume={1},
  pages={1--40},
  year={2022}
}

@article{ma2020statistical,
  title={Statistical inference for covariate-adaptive randomization procedures},
  author={Ma, Wei and Qin, Yichen and Li, Yang and Hu, Feifang},
  journal={Journal of the American Statistical Association},
  volume={115},
  number={531},
  pages={1488--1497},
  year={2020},
  publisher={Taylor \& Francis}
}

@article{murtaugh1994primary,
  title={Primary biliary cirrhosis: prediction of short-term survival based on repeated patient visits},
  author={Murtaugh, Paul A and Dickson, E Rolland and Van Dam, Gooitzen M and Malinchoc, Michael and Grambsch, Patricia M and Langworthy, Alice L and Gips, Chris H},
  journal={Hepatology},
  volume={20},
  number={1},
  pages={126--134},
  year={1994},
  publisher={Wiley Online Library}
}

@article{zhang2019theory,
  title={Theory on covariate-adaptive randomized clinical trials: efficiency, selection bias and randomization methods},
  author={Zhang, Li-Xin},
  journal={arXiv preprint arXiv:1912.03636},
  year={2019}
}

@article{fang2026general,
  title={A General (Non-Markovian) Framework for Covariate Adaptive Randomization: Achieving Balance While Eliminating the Shift},
  author={Fang, Hengjia and Ma, Wei},
  journal={arXiv preprint arXiv:2602.22648},
  year={2026}
}

@article{fang2026cbara,
  title={CBARA: Covariate-Balanced-and-Adjusted Response-Adaptive Randomization},
  author={Fang, Hengjia and Ma, Wei},
  journal={arXiv preprint arXiv:2604.25565},
  year={2026}
}

@article{junge2003noncommutative,
  title={Noncommutative burkholder/rosenthal inequalities},
  author={Junge, Marius and Xu, Quanhua},
  journal={The Annals of Probability},
  volume={31},
  number={2},
  pages={948--995},
  year={2003},
  publisher={Institute of Mathematical Statistics}
}

\appendix
\newpage
\bigskip
\begin{center}
{\large\bf SUPPLEMENTARY MATERIAL}
\end{center}

\setcounter{equation}{0}
\numberwithin{equation}{section}
\numberwithin{table}{section}
\numberwithin{figure}{section}



\section{Proofs of Theoretical Properties for IE-CAR Procedures}
First, we prove an essential lemma, which plays an essential role in the subsequent proofs of Theorems \ref{Thm_traditional_CAR} and \ref{Thm_RCAR}. 
\begin{lemma}\label{Lemma_qusi_smallball}
	Suppose that Assumptions \ref{assump_iid}-\ref{assump_boundSigma} hold. Then, there exists constants $\kappa$, $\delta>0$ such that for $m=0,\cdots,n-1$,
	\begin{align}\label{Eq_qusi_smallball}
		\mathbb{P}\left( \left|\Lambda_{m}^\top \phi(X_{m+1})\right| \geq \kappa \|\Lambda_m \| \middle| \Lambda_m  \right)
		\geq \delta.
	\end{align}
\end{lemma}

\begin{remark}[Further Details on Remark \ref{remark_weaken_condition}]\label{remark_further_detail}
	In Assumption \ref{assump_phi1}, the requirement $\nu\geq \iota\geq 4$ is imposed only to establish (\ref{Eq_qusi_smallball}). However, if $\phi(X)$ satisfies the small-ball condition, the conclusion follows directly. Therefore, in the subsequent proofs of Theorems \ref{Thm_traditional_CAR} and \ref{Thm_RCAR}, the moment conditions can be weakened to $\nu>2$ and $\iota =2$. 
\end{remark}

\begin{proof}[Proof of Lemma \ref{Lemma_qusi_smallball}: ]
	First, we prove that there exists a constant $\lambda_0>0$ such that for $m=0,1,\cdots,n-1$, 
	\begin{align*}
		\Lambda_m^\top \boldsymbol{\Sigma}_{\phi\phi} \Lambda_m \geq \lambda_0 \left\| \Lambda_m\right\|^2. 
	\end{align*}
	By Assumption \ref{assump_boundSigma}, there exists an orthogonal matrix $\boldsymbol{U}$ such that
	\begin{align*}
		\boldsymbol{U} \boldsymbol{\Sigma}_{\phi\phi} \boldsymbol{U}^\top = \operatorname{diag}\{\lambda_1,\cdots,\lambda_{q'},0,\cdots,0\},
	\end{align*}
	where $q'\leq q$ denotes the number of positive eigenvalues of $\boldsymbol{\Sigma}_{\phi\phi}$ and $\lambda_1,\cdots,\lambda_{q'}$ denote the corresponding positive eigenvalues. Let $\widetilde{\phi}(X) = (\widetilde{\phi}_{(1)}(X)^\top,\widetilde{\phi}_{(2)}(X)^\top)^\top = \boldsymbol{U}\phi(X)$, $\widetilde{\Lambda}_m = (\widetilde{\Lambda}_{m(1)}^\top,\widetilde{\Lambda}_{m(2)}^\top)^\top = \boldsymbol{U}\Lambda_m$, where $\widetilde{\phi}_{(1)}(X)$ and $\widetilde{\Lambda}_{m(1)}$ are $q'$-dimensional, and $\widetilde{\phi}_{(2)}(X)$ and $\widetilde{\Lambda}_{m(2)}$ are $(q-q')$-dimensional. Then, $\mathbb{E}[\widetilde{\phi}(X)\widetilde{\phi}(X)^\top] = \operatorname{diag}\{\lambda_1,\cdots,\lambda_{q'},0,\cdots,0\}$. Note that $\mathbb{E}[\widetilde{\phi}_{(2)}(X)\widetilde{\phi}_{(2)}(X)^\top] = \boldsymbol{0}$, which implies that $\widetilde{\phi}_{(2)}(X) =0$ almost surely and then $\widetilde{\Lambda}_{m(2)}=0$ almost surely. Therefore, let $\lambda_0 = \inf_{1\leq i \leq q'} \lambda_i > 0$,
	\begin{align*}
		\Lambda_m^\top \boldsymbol{\Sigma}_{\phi\phi} \Lambda_m = \widetilde{\Lambda}_{m(1)}^\top \operatorname{diag}\{\lambda_1,\cdots,\lambda_{q'}\} \widetilde{\Lambda}_{m(1)} \geq \lambda_0 \left\| \widetilde{\Lambda}_{m(1)} \right\|^2 = \lambda_0 \left\| \widetilde{\Lambda}_{m} \right\|^2 = \lambda_0 \left\| \Lambda_{m} \right\|^2. 
	\end{align*}
	By Paley-Zygmund inequality, for $\theta\in(0,1)$, 
	\begin{align*}
		\mathbb{P}\left( \left|\Lambda_{m}^\top \phi(X_{m+1})\right|^2 \geq \theta \Lambda_m^\top \boldsymbol{\Sigma}_{\phi\phi} \Lambda_m \middle| \Lambda_m  \right) \geq (1-\theta)^2 \frac{\left(\Lambda_m^\top \boldsymbol{\Sigma}_{\phi\phi} \Lambda_m\right)^2}{\mathbb{E}\left[\left(\Lambda_{m}^\top \phi(X_{m+1})\right)^4\middle| \Lambda_m \right] } \geq (1-\theta)^2 \frac{\lambda_0^2 }{M_4}.
	\end{align*}
	Besides,
	\begin{align*}
		\mathbb{P}\left( \left|\Lambda_{m}^\top \phi(X_{m+1})\right|^2 \geq \theta \Lambda_m^\top \boldsymbol{\Sigma}_{\phi\phi} \Lambda_m \middle| \Lambda_m  \right) \leq&\ \mathbb{P}\left( \left|\Lambda_{m}^\top \phi(X_{m+1})\right|^2 \geq \theta \lambda_0\|\Lambda_m\|^2 \middle| \Lambda_m  \right) \\
		=&\ \mathbb{P}\left( \left|\Lambda_{m}^\top \phi(X_{m+1})\right| \geq \theta \sqrt{\lambda_0}\|\Lambda_m\| \middle| \Lambda_m  \right).
	\end{align*}
	Then, let $\kappa = \theta\sqrt{\lambda_0}$ and $\delta = (1-\theta)^2 \lambda_0^2/M_4$, we prove the desired results. 
\end{proof}

\begin{proof}[Proof of Theorem \ref{Thm_traditional_CAR}:] 
	First, we consider an elementary inequality
	\begin{align}\label{proof_element_eq}
		\| u + v \|^r - \| u\|^r \leq r (u^\top v) \| u \|^{r-2} + c_r \left( \| u \|^{r-2}\| v\|^2 +  \| v\|^r \right),
	\end{align}
	where $c_{r}$ is a constant that only depends on $r$. Note that $\Lambda_{m+1} = \Lambda_{m} + (T_{m+1}-\pi)\phi(X_{m+1})$, which implies that
	\begin{align*}
		\| \Lambda_{m+1} \|^r - \| \Lambda_{m} \|^r  \leq &\ r (T_{m+1}-\pi) \Lambda_{m}^\top \phi(X_{m+1}) \| \Lambda_{m} \|^{r-2} \\
		&\ + c_{r} \left( \|\phi(X_{m+1}) \|^{r} + \|\phi(X_{m+1}) \|^2 \| \Lambda_{m} \|^{r-2} \right).
	\end{align*}
	Let \(\mathcal{F}_m\) be the history \(\sigma\)-field generated by the covariates \(\{X_i\}_{i=1}^m\) and the assignments \(\{T_i\}_{i=1}^m\), then
	\begin{align*}
		\mathbb{E}\left[\|\Lambda_{m+1}\|^{r} \middle| \mathcal{F}_{m}, X_{m+1}\right] - \|\Lambda_{m}\|^r 
		\leq&\  -r(\rho-\pi)\left|\Lambda_{m}^\top\phi(X_{m+1})\right| \|\Lambda_{m}\|^{r-2} \\
		&\ + c_r\left\{\|\phi(X_{m+1})\|^r + \|\phi(X_{m+1})\|^2\|\Lambda_{m}\|^{r-2} \right\},
	\end{align*}
	which implies that
	\begin{align*}
		\mathbb{E}\left[\|\Lambda_{m+1}\|^r\middle| \Lambda_{m} \right] - \| \Lambda_{m} \|^r
		&\leq -r(\rho-\pi)\| \Lambda_{m} \|^{r-2} \mathbb{E}\left[\left|  \Lambda_{m}^\top\phi(X_{m+1})\right| \middle|  \Lambda_{m}   \right] \\
		&\quad + c_r\left(\beta_r + \beta_2 \| \Lambda_{m}\|^{r-2}\right),
	\end{align*} 
	where $\beta_r = \mathbb{E}[\|\phi(X)\|^r]$ for $r\geq 2$. \\
	By Lemma \ref{Lemma_qusi_smallball}, we obtain that 
	\begin{align*}
		\mathbb{E}\left[\left|  \Lambda_{m}^\top\phi(X_{m+1})\right| \middle|  \Lambda_{m}   \right] 
		\geq&\ \mathbb{E}\left[\left|  \Lambda_{m}^\top\phi(X_{m+1})\right| \mathbb{I}\left\{ \left|\Lambda_{m}^\top \phi(X_{m+1})\right| \geq \kappa \|\Lambda_m \|\right\} \middle| \Lambda_m \right] \\
		\geq&\ \kappa\|\Lambda_m\| \mathbb{P}\left( \left|\Lambda_{m}^\top \phi(X_{m+1})\right| \geq \kappa \|\Lambda_m \| \middle| \Lambda_m  \right)
		\geq \kappa\delta \|  \Lambda_{m} \|,
	\end{align*}
	which implies that 
	\begin{align*}
		\mathbb{E}\left[\|\Lambda_{m+1}\|^r\middle|  \Lambda_{m}  \right] - \| \Lambda_{m} \|^r
		&\leq -c_* r(\rho-\pi)\| \Lambda_{m} \|^{r-1} + c_r\left(\beta_r + \beta_2 \| \Lambda_{m} \|^{r-2}\right),
	\end{align*}
	where $c_* = \kappa\delta$. Let $g(x) = -c_* r(\rho-\pi)x^{r-1} + c_r\left(\beta_r + \beta_2 x^{r-2}\right)$, we have
	\begin{align*}
		g^\prime (x) = -c_*r(r-1)(\rho-\pi)x^{r-2} + c_r\beta_2(r-2) x^{r-3} = 0,
	\end{align*}
	has only one solution $x^* =  \frac{c_r\beta_2(r-2)}{c_*r(r-1)(\rho-\pi)}$ when $x>0$. Hence, we have 
	\begin{align*}
		\mathbb{E}\left[\|\Lambda_{m+1}\|^r\middle|  \Lambda_{m}  \right] - \| \Lambda_{m} \|^r \leq g(x^*) = \frac{c_{r} \beta_2}{r-1} \left(\frac{c_r\beta_2(r-2)}{c_*r(r-1)(\rho-\pi)} \right)^{r-2} + c_{r}\beta_r := H_{r}.  
	\end{align*}
	Note that $\beta_r \asymp q^{r/2}$, then we have $H_r \asymp q^{r-1}$. Thus, we have
	\begin{align}
		\mathbb{E}\left[\|\Lambda_{m+1}\|^r\right] - \mathbb{E}\left[\| \Lambda_{m} \|^r\right] \leq H_{r},
	\end{align}
	which implies that
	\begin{align*}
		\mathbb{E}\left[\|\Lambda_{m}\|^r\right] \leq mH_{r}. 
	\end{align*}
	On the other hand, note that $\| \Lambda_{m+1} \|^2 = \| \Lambda_{m} \|^2 + 2(T_{m+1}-\pi) \Lambda_{m}^\top \phi(X_{m+1}) + (T_{m+1}-\pi)^2 \| \phi(X_{m+1}) \|^2$. Then, 
	\begin{align*}
		\mathbb{E}[\| \Lambda_{m+1} \|^2 | \Lambda_{m} ] - \| \Lambda_{m} \|^2  \leq& -2(\rho-\pi) \mathbb{E}\left[ \left|\Lambda_{m}^\top\phi(X_{m+1})\right| \middle| \Lambda_{m} \right] + \beta_2 \\
		\leq& -2(\rho-\pi)c_* \|\Lambda_{m} \|+ \beta_2,
	\end{align*} 
	which implies that
	\begin{align*}
		\mathbb{E}[\| \Lambda_{m+1} \|^2] - \mathbb{E}[\| \Lambda_{m} \|^2 ]  \leq -2(\rho-\pi)c_* \mathbb{E}[\|\Lambda_{m} \| ] + \beta_2,
	\end{align*}
	Let $m^* \in\{0,1,\cdots,m\}$ be the last one for which $-2(\rho-\pi)c_* \mathbb{E}[\|\Lambda_{m^*} \|]+\beta_2\geq 0$. Then, we have
	\begin{align}
		\mathbb{E}[\|\Lambda_{m^*} \|] \leq \frac{\beta_2}{2(\rho-\pi)c_*} \quad \text{and} \quad
		\mathbb{E}[\|\Lambda_{m} \|^2] \leq \mathbb{E}[\|\Lambda_{m^*} \|^2] + \beta_2. 
	\end{align}
	By H\"{o}lder inequality, we have
	\begin{align}
		\mathbb{E}[\|\Lambda_{m^*} \|^2] \leq \left(\mathbb{E}[\| \Lambda_{m^*} \|]\right)^{\frac{r-2}{r-1}} \left(\mathbb{E}[\| \Lambda_{m^*} \|^r]\right)^{\frac{1}{r-1}} \leq \left( \frac{\beta_2}{2(\rho-\pi)c_*} \right)^{\frac{r-2}{r-1}} \left(kH_{r}\right)^{\frac{1}{r-1}},
	\end{align}
	which implies that
	\begin{align}
		\mathbb{E}[\|\Lambda_{m} \|^2] \leq&\ \left( \frac{\beta_2}{2(\rho-\pi)c_*} \right)^{\frac{r-2}{r-1}} \left(mH_{r}\right)^{\frac{1}{r-1}} + \beta_2 = m^{\frac{1}{r-1}} \widetilde{H}_r + \beta_2,
	\end{align}
	where
	\begin{align*}
		\widetilde{H}_r = \left( \frac{\beta_2}{2(\rho-\pi)c_*} \right)^{\frac{r-2}{r-1}} H_{r}^{\frac{1}{r-1}} = O\left(q^{\frac{2r-3}{r-1}}\right).
	\end{align*}
	Let $r=\nu$, we obtain that for $m=1,\cdots,n$, 
	\begin{align*}
		\mathbb{E}[\|\Lambda_{m} \|^2] = O\left(m^{\frac{1}{\nu-1}} q^{\frac{2\nu-3}{\nu-1}} \right). 
	\end{align*}
	In addition, note that
	\begin{align*}
		\mathbb{E}[\| \Lambda_{m+1} \|^2 | \Lambda_{m} ] - \| \Lambda_{m} \|^2  \leq -2(\rho-\pi) \mathbb{E}\left[ \left|\Lambda_{m}^\top\phi(X_{m+1})\right| \middle| \Lambda_{m} \right] +  \beta_2 \leq \beta_2,
	\end{align*} 
	which implies that
	\begin{align*}
		\mathbb{E}[\| \Lambda_{m} \|^2] \leq m\beta_2.
	\end{align*}
	Therefore, we can conclude that
	\begin{align*}
		\mathbb{E}[\| \Lambda_{n} \|^2] \lesssim \min\left\{ n^{\frac{1}{\nu-1}} q^{\frac{2\nu-3}{\nu-1}}, n\beta_2\right\} = \begin{cases}
			O\left(n^{\frac{1}{\nu-1}} q^{\frac{2\nu-3}{\nu-1}}\right) &\text{if } q = o(n), \\
			O(nq) &\text{if } q = \Omega(n).
		\end{cases}
	\end{align*}
	Next, we turn to the lower bound for $\mathbb{E}[\| \Lambda_{n} \|^2]$. By Cauchy-Schwarz inequality, we have
	\begin{align}\label{eq_cs}
		\mathbb{E}\left[ \left|\Lambda_{m}^\top\phi(X_{m+1})\right| \middle| \Lambda_{m} \right] \leq \left(\mathbb{E}\left[ \left(\Lambda_{m}^\top\phi(X_{m+1})\right)^2  \middle| \Lambda_{m}  \right]\right)^{\frac{1}{2}}  \leq M_2^{1/2} \|\Lambda_{m} \|,
	\end{align}
	which implies that
	\begin{align}\label{eq_IECAR_lower} 
		\mathbb{E}[\| \Lambda_{m+1} \|^2 | \Lambda_{m} ] - \| \Lambda_{m} \|^2 &\geq   -2(\rho-\pi) \mathbb{E}\left[ \left|\Lambda_{m}^\top\phi(X_{m+1})\right| \middle| \Lambda_{m} \right] + (\pi \wedge (1-\pi))^2 \beta_2 \notag \\
		&\geq  -2(\rho-\pi)M_2^{1/2} \|\Lambda_{m}\| + (\pi \wedge (1-\pi))^2 \beta_2.
	\end{align}
	Let
	\begin{align*}
		A = \frac{(\pi \wedge (1-\pi))^4 \beta_2^2}{16(\rho-\pi)^2 M_2}.
	\end{align*}
	Then, if $\mathbb{E}[\| \Lambda_{m} \|^2] \leq A$, 
	\begin{align*}
		\mathbb{E}[\| \Lambda_{m+1} \|^2] - \mathbb{E}[\| \Lambda_{m} \|^2] &\geq -2(\rho-\pi)M_2^{1/2} \mathbb{E}[\|\Lambda_{m}\|] + (\pi \wedge (1-\pi))^2 \beta_2 \\
		&\geq -2(\rho-\pi)M_2^{1/2} \left(\mathbb{E}[\| \Lambda_{m} \|^2] \right)^{1/2} + (\pi \wedge (1-\pi))^2 \beta_2 \\
		&\geq \frac{1}{2} (\pi \wedge (1-\pi))^2 \beta_2.
	\end{align*}
	On the other hand, if $\mathbb{E}[\| \Lambda_{m} \|^2] \geq A$. Note that for sufficiently large $n$, we have $A \geq (\rho-\pi)^2 M_2$. Then,
	\begin{align*}
		\mathbb{E}[\| \Lambda_{m+1} \|^2] \geq A -2(\rho-\pi)M_2^{1/2}A^{1/2} + (\pi \wedge (1-\pi))^2 \beta_2  \geq A. 
	\end{align*}
	From above discussion, we can conclude that for sufficiently large $n$ and $m=0,\cdots,n-1$, 
	\begin{align*}
		\mathbb{E}[\| \Lambda_{m+1} \|^2] \geq \begin{cases}
			\mathbb{E}[\| \Lambda_{m} \|^2] + \frac{1}{2} (\pi \wedge (1-\pi))^2 \beta_2  & \text{if } \mathbb{E}[\|\Lambda_m \|^2] \leq A, \\
			A & \text{if } \mathbb{E}[\|\Lambda_m \|^2] \geq A.
		\end{cases}
	\end{align*}
	Therefore, if $\mathbb{E}[\|\Lambda_n\|^2] \leq A$, then for $1\leq m \leq n$, $\mathbb{E}[\|\Lambda_m \|^2] \leq A$. Thus, 
	\begin{align*}
		\mathbb{E}[\| \Lambda_n \|^2] \geq \mathbb{E}[\| \Lambda_{n-1} \|^2] + \frac{1}{2}(\pi \wedge (1-\pi))^2 \beta_2 \geq \cdots \geq \frac{1}{2}(\pi \wedge (1-\pi))^2 n\beta_2,
	\end{align*}
	which implies that
	\begin{align*}
		\mathbb{E}[\| \Lambda_{n} \|^2] \geq \min\left\{A, \frac{1}{2}(\pi \wedge (1-\pi))^2 n\beta_2 \right\} = \begin{cases}
			\Omega(q^2) &  \text{if } q = o(n), \\
			\Omega(nq) &  \text{if } q = \Omega(n). \\
		\end{cases}
	\end{align*}
\end{proof}
\section{Proofs of Theoretical Properties for IR-CAR Procedures}
\subsection{Theoretical Properties of Imbalance Measure}
\begin{proof}[Proof of Theorem \ref{Thm_RCAR}:]
	Note that for $1\leq m \leq n-1$, $\Lambda_{m+1} = \Lambda_m + (T_{m+1}-\pi)\phi(X_{m+1})$. Then, by (\ref{proof_element_eq}), 
	\begin{align*}
		\| \Lambda_{m+1} \|^r - \| \Lambda_{m} \|^r  \leq &\ r (T_{m+1}-\pi) \Lambda_{m}^\top \phi(X_{m+1}) \| \Lambda_{m} \|^{r-2} \\
		&\ + c_{r} \left( \|\phi(X_{m+1}) \|^{r} + \|\phi(X_{m+1}) \|^2 \| \Lambda_{m} \|^{r-2} \right),
	\end{align*} 
	which further implies that
	\begin{align*}
		&\ \mathbb{E}\left[\|\Lambda_{m+1}\|^r\middle| \mathcal{F}_{m},X_{m+1} \right] - \| \Lambda_{m} \|^r \\
		\leq &\  r \left( \ell_{m+1} - \pi  \right) \Lambda_{m}^\top\phi(X_{m+1}) \| \Lambda_{m} \|^{r-2}  + c_{r} \left( \|\phi(X_{m+1}) \|^{r} + \|\phi(X_{m+1}) \|^2 \| \Lambda_{m} \|^{r-2} \right)  \\ 
		= &\ -r \left| \ell_{m+1} - \pi \right|\cdot \left|  \Lambda_{m}^\top\phi(X_{m+1})\right| \cdot  \| \Lambda_{m} \|^{r-2}  + c_{r} \left( \|\phi(X_{m+1}) \|^{r} + \|\phi(X_{m+1}) \|^2 \| \Lambda_{m} \|^{r-2} \right),
	\end{align*}
	where
	\begin{align*}
		\ell_{m+1} = \ell\left( \frac{\Lambda_{m}^\top\phi(X_{m+1})}{m^{\gamma}}  \right). 
	\end{align*}
	First, we prove that there exists a constant $c_0>0$ such that
	\begin{align}\label{Proof1_eq1}
		\mathbb{E}\left[ \left| \ell_{m+1} - \pi \right|\cdot \left|  \Lambda_{m}^\top\phi(X_{m+1})\right| \middle| \mathcal{F}_{m}  \right] \geq c_0 \left(\frac{\|\Lambda_m \|^2 }{m^{\gamma}} \wedge \|\Lambda_m \| \right). 
	\end{align}
	Note that there exists a constant $\underline{\lambda}_{\ell} >0$ such that 
	\begin{align*}
		\left| \ell_{m+1} - \pi \right| \geq \underline{\lambda}_{\ell} \left(\frac{\left| \Lambda_{m}^\top \phi(X_{m+1}) \right| }{m^{\gamma}}\wedge 1 \right),
	\end{align*}
	which implies that 
	\begin{align*}
		\mathbb{E}\left[ \left| \ell_{m+1} - \pi \right|\cdot \left|  \Lambda_{m}^\top\phi(X_{m+1})\right| \middle| \mathcal{F}_{m}  \right] \geq \underline{\lambda}_{\ell} \mathbb{E}\left[\frac{\left|\Lambda_{m}^\top \phi(X_{m+1})\right|^2}{m^{\gamma}}  \wedge \left|\Lambda_{m}^\top \phi(X_{m+1})\right| \middle| \Lambda_m \right].
	\end{align*}
	By Lemma \ref{Lemma_qusi_smallball}, we have 
	\begin{align*}
		&\ \mathbb{E}\left[\frac{\left|\Lambda_{m}^\top \phi(X_{m+1})\right|^2}{m^{\gamma}}  \wedge \left|\Lambda_{m}^\top \phi(X_{m+1})\right| \middle| \Lambda_m \right] \\
		\geq&\  \mathbb{E}\left[\left(\frac{\left|\Lambda_{m}^\top \phi(X_{m+1})\right|^2}{m^{\gamma}}  \wedge \left|\Lambda_{m}^\top \phi(X_{m+1})\right|\right) \mathbb{I}\left\{ \left|\Lambda_{m}^\top \phi(X_{m+1})\right| \geq \kappa \|\Lambda_m \|\right\} \middle| \Lambda_m \right]  \\
		\geq&\ \left(\kappa^2 \frac{\|\Lambda_m \|^2}{m^{\gamma}} \wedge  \kappa \|\Lambda_m \|\right) \mathbb{P}\left(\left|\Lambda_{m}^\top \phi(X_{m+1})\right|
		\geq \kappa \|\Lambda_m \|\middle| \Lambda_m \right) \\
		\geq&\ \delta \left(\kappa^2 \frac{\|\Lambda_m \|^2}{m^{\gamma}} \wedge  \kappa \|\Lambda_m \|\right) \geq \delta \left(\kappa^2 \wedge \kappa \right) \left( \frac{\|\Lambda_m \|^2}{m^{\gamma}} \wedge \|\Lambda_m \|\right). 
	\end{align*}
	Then, let $c_0 = \delta \left(\kappa^2 \wedge \kappa \right)\underline{\lambda}_{\ell}$, we obtain (\ref{Proof1_eq1}).
	Therefore,
	\begin{align*}
		\mathbb{E}\left[\|\Lambda_{m+1}\|^r\middle| \Lambda_{m} \right] - \| \Lambda_{m} \|^r
		&\leq -rc_0  \| \Lambda_{m} \|^{r-1}\left(\frac{\|\Lambda_m \|}{m^{\gamma}} \wedge 1  \right)   + c_r\left(\beta_r + \beta_2 \| \Lambda_{m}\|^{r-2}\right). 
	\end{align*} 
	We define
	\begin{align*}
		f_m(x) = -rc_0 x^{r-1} \left(\frac{x}{m^{\gamma}} \wedge 1 \right) + c_r\left(\beta_r + \beta_2 x^{r-2}\right),\quad x\geq 0, r>2. 
	\end{align*}
	Then, we obtain that
	\begin{align*}
		f_m^\prime (x) = 
		\begin{cases}
			-m^{-\gamma}r^2 c_0 x^{r-1} + c_r\beta_2(r-2)x^{r-3} \quad &(0\leq x\leq m^{\gamma}), \\
			-r c_0(r-1) x^{r-2} + c_r\beta_2(r-2)x^{r-3} \quad &(x > m^{\gamma}),
		\end{cases}
	\end{align*}
	which implies that the maximizer of $f_m(x)$ can only occur at
	\begin{align*}
		x_1(m) = \sqrt{\frac{(r-2) c_r\beta_2 }{r^2c_0}} m^{\frac{\gamma}{2}} \asymp m^{\frac{\gamma}{2}}q^{\frac{1}{2}} \quad \text{or}\quad  x_2(m) = \frac{(r-2) c_r\beta_2 }{r(r-1)c_0} \asymp q.
	\end{align*}
	If $q=o(n^\gamma)$, there exists a sufficient large $n^*$ such that $x_1(m) \leq m^\gamma$ and $x_2(m) < m^\gamma$ for every $m\geq n^*$. Hence, $f_m(x)$ is increasing on $(0,x_1(m))$ and decreasing on $(x_1(m),+\infty)$, which implies that
	\begin{align*}
		\mathbb{E}\left[\|\Lambda_{m+1}\|^r\middle| \Lambda_{m} \right] - \| \Lambda_{m} \|^r \leq f_m(x_1(m)) = \frac{2c_r\beta_2}{r} \left(\frac{(r-2) c_r\beta_2 }{r^2c_0}\right)^{\frac{r-2}{2}} m^{\frac{\gamma r}{2}} + c_r\beta_r. 
	\end{align*}
	Therefore, for $m\geq n^*$,
	\begin{align*}
		\mathbb{E}\left[\|\Lambda_{m}\|^r \right] - \mathbb{E}\left[\|\Lambda_{n^*}\|^r \right] \leq \frac{2c_r\beta_2}{r} \left(\frac{(r-2) c_r\beta_2 }{r^2c_0}\right)^{\frac{r-2}{2}} \sum_{k=n^*}^{m-1} k^{\frac{\gamma r}{2}} + (m-n^*) c_r\beta_r \asymp m^{\frac{\gamma r}{2}-\gamma+1} q^{\frac{r}{2}}.
	\end{align*}
	Hence, for $r>2$ and every $m$, there exists a constant $C_r>0$ such that if $q=o(n^{\gamma})$,
	\begin{align}\label{proof_RCAR_Er1}
		\mathbb{E}\left[\|\Lambda_{m}\|^r \right] \leq C_r m^{\frac{\gamma r}{2}-\gamma+1} q^{\frac{r}{2}}, \quad 1\leq m \leq n. 
	\end{align}
	If $n^{\gamma}=o(q)$, there exists sufficient large $n'$ such that $x_1(m) > m^\gamma$ and $x_2(m) \geq m^\gamma$ for every $m\geq n'$. Hence, $f_m(x)$ is increasing on $(0,x_2(m))$ and decreasing on $(x_2(m),+\infty)$, which implies that
	\begin{align*}
		\mathbb{E}\left[\|\Lambda_{m+1}\|^r\middle| \Lambda_{m} \right] - \| \Lambda_{m} \|^r \leq f_m(x_2(m)) = \frac{c_r\beta_2}{r-1}\left(\frac{(r-2)c_r\beta_2}{r(r-1)c_0} \right)^{r-2} + c_r\beta_r \asymp q^{r-1}.
	\end{align*}
	Therefore, for $m\geq n'$,
	\begin{align*}
		\mathbb{E}\left[\|\Lambda_{m}\|^r \right] - \mathbb{E}\left[\|\Lambda_{n'}\|^r \right] \leq (m-n')\left[ \frac{c_r\beta_2}{r-1}\left(\frac{(r-2)c_r\beta_2}{r(r-1)c_0} \right)^{r-2} + c_r\beta_r \right] \asymp mq^{r-1}.
	\end{align*}
	Hence, for any $r>2$ there exists a constant $C_r>0$ such that if $n^{\gamma}=o(q)$,
	\begin{align}\label{proof_RCAR_Er2}
		\mathbb{E}\left[\|\Lambda_{m}\|^r \right] \leq C_r mq^{r-1}, \quad 1\leq m \leq n.
	\end{align}
	If $n^{\gamma}\asymp q$, note that $f_k(x_1(k))>0$ and $f_k(x_2(k)) >0$ hold for every $1\leq k \leq n$. Then, 
	\begin{align*}
		\mathbb{E}[\| \Lambda_{m} \|^r] \leq \sum_{k=1}^{m-1} \left(f_k(x_1(k)) + f_k(x_2(k)) \right) \asymp m q^{r-1}, 
	\end{align*}
	which implies that (\ref{proof_RCAR_Er1}) and (\ref{proof_RCAR_Er2}) also hold when $n^{\gamma}\asymp q$. \\
	Let $E_m = \{\|\Lambda_m \| \leq m^{\gamma} \}$ and $Q_{m,r} = \mathbb{E}\left[\| \Lambda_{m} \|^{r-1}\left(m^{-\gamma}\|\Lambda_m \| \wedge 1  \right)  \right]$. Then, we obtain that
	\begin{align}
		\mathbb{E}\left[\|\Lambda_m \|^r \mathbb{I}_{E_m}  \right] \leq m^{\gamma}Q_{m,r} \quad \text{and} \quad \mathbb{E}\left[\|\Lambda_m \|^{r-1} \mathbb{I}_{E_m^c}  \right] \leq Q_{m,r}. 
	\end{align}
	Let $m^*\in\{0,1,\cdots,m\}$ denotes the last one for which 
	\begin{align*}
		-rc_0Q_{m^*,r} + c_r\left(\beta_r + \beta_2 \mathbb{E}\left[\| \Lambda_{m^*}  \|^{r-2}   \right]   \right) \geq 0.
	\end{align*}
	Then, we conclude that 
	\begin{align}
		Q_{m^*,r} \leq&\ \frac{c_r}{rc_0} \left( \beta_r + \beta_2  \mathbb{E}\left[\| \Lambda_{m^*}  \|^{r-2}   \right] \right), \label{proof_RCAR_eq1} \\
		\mathbb{E}\left[\|\Lambda_m \|^r \right] \leq&\ \mathbb{E}\left[\|\Lambda_{m^*+1} \|^r \right] \leq \mathbb{E}\left[\|\Lambda_{m^*} \|^r \right] + c_r\left(\beta_r + \beta_2 \mathbb{E}\left[\| \Lambda_{m^*}  \|^{r-2}   \right]   \label{proof_RCAR_eq2}\right). 
	\end{align}
	By H\"{o}lder inequality, for any $s$, $s'>1$ such that $\frac{1}{s} + \frac{1}{s'} =1$, we have 
	\begin{align}\label{proof_RCAR_eq3}
		&\ \mathbb{E}\left[\|\Lambda_{m^*} \|^r \mathbb{I}_{E_{m^*}^c}  \right] = \mathbb{E}\left[\|\Lambda_{m^*} \|^\frac{r-1}{s} \|\Lambda_{m^*} \|^\frac{r-1+s'}{s'} \mathbb{I}_{E_{m^*}^c}  \right]\notag  \\
		\leq &\ \left(\mathbb{E}\left[\|\Lambda_m \|^{r-1} \mathbb{I}_{E_{m^*}^c}  \right] \right)^{1/s} \left(\mathbb{E}\left[\|\Lambda_{m^*} \|^{r-1+s'} \right]\right)^{1/s'} \leq Q_{m^*,r}^{1/s} \left(\mathbb{E}\left[\|\Lambda_{m^*} \|^{r-1+s'} \right]\right)^{1/s'}. 
	\end{align}
	Combining (\ref{proof_RCAR_eq1}), (\ref{proof_RCAR_eq2}), (\ref{proof_RCAR_eq3}) and letting $r-1+s' = \nu$, that is, $s' = \nu+1-r$ and $s=(\nu+1-r)/(\nu-r)$, we obtain that
	\begin{align}\label{proof_RCAR_E^r}		
		\mathbb{E}\left[\|\Lambda_m \|^r \right] \leq&\ c_r\left(\frac{m^{\gamma} }{rc_0} + 1\right)\left(\beta_r + \beta_2 \mathbb{E}\left[\| \Lambda_{m^*}  \|^{r-2}   \right]  \right) \notag \\
		&\ + \left(\frac{c_r}{rc_0}\right)^{\frac{\nu-r}{\nu+1-r}} \left(\beta_r + \beta_2 \mathbb{E}\left[\| \Lambda_{m^*}  \|^{r-2}   \right]  \right)^{\frac{\nu-r}{\nu+1-r}}\left(\mathbb{E}\left[\|\Lambda_{m^*} \|^{\nu} \right]\right)^{\frac{1}{\nu+1-r}}. 
	\end{align}
	Let $r=2$, we have 
	\begin{align*}
		\mathbb{E}\left[\|\Lambda_m \|^2 \right] \leq&\ 2c_2\beta_2 \left(\frac{m^{\gamma} }{2c_0} + 1\right) + \left(\frac{c_r\beta_2}{c_0}\right)^{\frac{\nu-2}{\nu-1}}\left(\mathbb{E}\left[\|\Lambda_{m^*} \|^{\nu} \right]\right)^{\frac{1}{\nu-1}}. 
	\end{align*}
	By (\ref{proof_RCAR_Er1}) and (\ref{proof_RCAR_Er2}), we obtain that 
	\begin{align}\label{RCAR_L2bound}
		\mathbb{E}\left[ \|\Lambda_n \|^2 \right] = 
		\begin{cases}
			O\left(n^{\gamma}q + n^{\frac{\gamma\nu -2\gamma + 2}{2(\nu-1)}} q^{\frac{3\nu-4}{2(\nu-1)}} \right) \quad &\text{if } q=o(n^\gamma), \\
			O\left(n^{\gamma}q + n^{\frac{1}{\nu-1}} q^{\frac{2\nu-3}{\nu-1}} \right) = O\left(n^{\frac{1}{\nu-1}} q^{\frac{2\nu-3}{\nu-1}} \right) \quad  &\text{if }  q = \Omega(n^\gamma).
		\end{cases}
	\end{align}
	In addition, note that
	\begin{align*}
		&\ \mathbb{E}\left[\|\Lambda_{m+1}\|^2\middle| \mathcal{F}_{m},X_{m+1} \right] - \| \Lambda_{m} \|^2 \\
		= &\  2 \left( \ell_{m+1} - \pi  \right) \Lambda_{m}^\top\phi(X_{m+1}) + \mathbb{E}\left[(T_{m+1}-\pi)^2\middle| \mathcal{F}_{m},X_{m+1}\right]  \|\phi(X_{m+1}) \|^2 \\
		\leq &\ - 2 \left|\ell_{m+1} - \pi \right| \left|\Lambda_{m}^\top\phi(X_{m+1}) \right| +  \|\phi(X_{m+1}) \|^2 \leq  \|\phi(X_{m+1}) \|^2, 
	\end{align*}
	which implies that 
	\begin{align}\label{Proof_RCAR_largedim_1}
		\mathbb{E}\left[\|\Lambda_{m}\|^2\right] \leq \mathbb{E}\left[\|\Lambda_{m-1}\|^2 \right] +  \beta_2 \leq \cdots \leq m  \beta_2. 
	\end{align}
	Therefore, we obtain that $\mathbb{E}[\|\Lambda_{n} \|^2] = O(nq)$ if $q=\Omega(n)$. \\
	Next, we turn to the lower bound for $\mathbb{E}[\|\Lambda_{n} \|^2]$. Note that
	\begin{align*}
		&\ \mathbb{E}\left[\|\Lambda_{m+1}\|^2\middle| \mathcal{F}_{m},X_{m+1} \right] - \| \Lambda_{m} \|^2 \\
		= &\  2 \left( \ell_{m+1} - \pi  \right) \Lambda_{m}^\top\phi(X_{m+1}) + \mathbb{E}\left[(T_{m+1}-\pi)^2\middle| \mathcal{F}_{m},X_{m+1}\right]  \|\phi(X_{m+1}) \|^2 \\
		\geq & - 2 \left|\ell_{m+1} - \pi \right| \left|\Lambda_{m}^\top\phi(X_{m+1}) \right| + \left(\pi\wedge (1-\pi)\right)^2 \|\phi(X_{m+1}) \|^2 \\
		\geq & -2\overline{\lambda}_{\ell,1} \left(m^{-\gamma}\left|\Lambda_{m}^\top\phi(X_{m+1}) \right|^2 \wedge \left|\Lambda_{m}^\top\phi(X_{m+1}) \right| \right)+ \left(\pi\wedge (1-\pi)\right)^2 \|\phi(X_{m+1}) \|^2,
	\end{align*}
	where the last inequality follows from the fact that $|\ell(x) -\pi | \leq \overline{\lambda}_{\ell,1}(|x|\wedge 1)$ for some constant $\overline{\lambda}_{\ell,1}>0$.\\
	Besides, by (\ref{eq_cs}), we have 
	\begin{align*}
		&\ \mathbb{E}\left[ \left(m^{-\gamma}\left|\Lambda_{m}^\top\phi(X_{m+1}) \right|^2\right) \wedge  \left|\Lambda_{m}^\top\phi(X_{m+1}) \right| \middle| \Lambda_{m} \right] \\
		\leq&\ \left(m^{-\gamma}\mathbb{E}\left[\left|\Lambda_{m}^\top\phi(X_{m+1}) \right|^2\middle| \Lambda_{m} \right]\right) \wedge \mathbb{E}\left[ \left|\Lambda_{m}^\top\phi(X_{m+1}) \right| \middle|\Lambda_{m} \right] \\
		\leq&\ \left(M_2m^{-\gamma} \left\|\Lambda_{m} \right\|^2\right) \wedge \left( M_2^{1/2} \left\|\Lambda_{m} \right\|\right).
	\end{align*}
	Therefore, 
	\begin{align*}
		&\ \mathbb{E}\left[\|\Lambda_{m+1}\|^2 \right] - \mathbb{E}\left[\| \Lambda_{m} \|^2\right] \\
		\geq&\ -2 \overline{\lambda}_{\ell,1}\left( \left\{M_2 m^{-\gamma} \mathbb{E}\left[\| \Lambda_{m} \|^2\right]\right\} \wedge \left\{ M_2^{1/2}\left(\mathbb{E}\left[\| \Lambda_{m} \|^2\right]\right)^{1/2} \right\} \right)  + \left(\pi\wedge (1-\pi)\right)^2\beta_2.
	\end{align*}
	On the one hand, 
	\begin{align*}
		\mathbb{E}\left[\|\Lambda_{m+1}\|^2 \right] - \mathbb{E}\left[\| \Lambda_{m} \|^2\right] \geq -2 \overline{\lambda}_{\ell,1}M_2^{1/2}\left(\mathbb{E}\left[\| \Lambda_{m} \|^2\right]\right)^{1/2} + \left(\pi\wedge (1-\pi)\right)^2\beta_2,
	\end{align*}
	which is of a similar form to the inequality in (\ref{eq_IECAR_lower}). Then, by applying the same technique used in the proof of Theorem \ref{Thm_traditional_CAR}, we obtain
	\begin{align*}
		\mathbb{E}[\| \Lambda_{n} \|^2] = \begin{cases}
			\Omega(q^2) &  \text{if } q = o(n), \\
			\Omega(nq) &  \text{if } q = \Omega(n). \\
		\end{cases}
	\end{align*}
	On the other hand, 
	\begin{align*}
		\mathbb{E}\left[\|\Lambda_{m+1}\|^2 \right]  \geq \left(1-2 \overline{\lambda}_{\ell,1}M_2 m^{-\gamma}\right) \mathbb{E}\left[\| \Lambda_{m} \|^2\right] + \left(\pi\wedge (1-\pi)\right)^2\beta_2. 
	\end{align*}
	Note that $\displaystyle \lim_{n\to \infty} [(n+1)^\gamma - n^{\gamma}] = 0$. Thus, there exists $n_0$ such that for $m\geq n_0$, $(m+1)^\gamma - m^{\gamma} \leq 2 \overline{\lambda}_{\ell,1}M_2$. Let $A^* = \min\{ (\pi\wedge (1-\pi))^2/(4\overline{\lambda}_{\ell,1}M_2), \mathbb{E}[\|\Lambda_{n_0}\|^2]/(n_0^{\gamma}\beta_2) \}$. Then, we claim that for $m\geq n_0$, $\mathbb{E}[\|\Lambda_{m} \|^2] \geq A^* m^{\gamma}\beta_2$. We prove this claim by induction. \\
	For \(m=n_0\), the result is immediate. Suppose that the claim holds for \(m\). Then, for \(m+1\), we have
	\begin{align*}
		\mathbb{E}\left[\|\Lambda_{m+1}\|^2 \right]  &\geq \left(1-2 \overline{\lambda}_{\ell,1}M_2 m^{-\gamma}\right) A^* m^{\gamma}\beta_2 + \left(\pi\wedge (1-\pi)\right)^2\beta_2 \\
		&\geq A^* m^{\gamma}\beta_2 + \left[\left(\pi\wedge (1-\pi)\right)^2 - 2 \overline{\lambda}_{\ell,1}M_2 A^* \right] \beta_2 \geq A^* (m+1)^{\gamma}\beta_2. 
	\end{align*}
	Thus, the claim holds for \(m+1\), completing the induction. Therefore, $\mathbb{E}[\|\Lambda_n \|^2] = \Omega(n^{\gamma}q)$. \\
	Based on the above discussion, we conclude that
	\begin{align*}
		\mathbb{E}[\|\Lambda_n  \|^2] = \begin{cases}
			\Omega(n^{\gamma}q) & \text{if } q = O(n^{\gamma}), \\
			\Omega(q^2) & \text{if } q = \Omega(n^{\gamma}) \text{ and } q =O(n),  \\
			\Omega(nq) & \text{if } q = \Omega(n).
		\end{cases}
	\end{align*}
\end{proof}

\subsection{Theoretical Properties of Additional Covariates}
\begin{lemma}\label{Lemma1}
	Suppose that Assumptions \ref{assump_iid}, \ref{assump_boundSigma} hold and Assumption \ref{assump_phi1} holds for all $\nu\geq 4$. Then, under IR-CAR, for $r \geq 1$, $\mathbb{E}[\|\Lambda_n \|^r]= O(n^{\frac{\gamma r}{2}}q^{\frac{r}{2}})$ if $q = O(n^{\gamma-\varepsilon_0})$ for some $\varepsilon_0>0$. 
\end{lemma}

\begin{proof}[Proof of Lemma \ref{Lemma1}: ]
It follows from Equation (\ref{RCAR_L2bound}) that $\mathbb{E}[\|\Lambda_n\|^2] = O(n^\gamma q)$. Then, for $1\leq r < 2$, 
\begin{align*}
\mathbb{E}\left[\|\Lambda_n \|^r \right] \leq \left(\mathbb{E}\left[\|\Lambda_n \|^2 \right]\right)^{\frac{r}{2}} = O(n^{\frac{\gamma r}{2}}q^{\frac{r}{2}}). 
\end{align*}
For $2\leq r \leq 3$, by (\ref{proof_RCAR_E^r}) we obtain that 
\begin{align*}
\mathbb{E}\left[\|\Lambda_n \|^r \right] \lesssim &\ n^{\gamma}\left(\beta_r + \beta_2 n^{\frac{\gamma (r-2)}{2}}q^{\frac{r-2}{2}} \right) + \left(\beta_r + \beta_2 n^{\frac{\gamma (r-2)}{2}}q^{\frac{r-2}{2}} \right)^{\frac{\nu-r}{\nu+1-r}}\left(n^{\frac{\gamma \nu}{2}-\gamma+1} q^{\frac{\nu}{2}} \right)^{\frac{1}{\nu+1-r}}\\
\lesssim &\ n^{\frac{\gamma r}{2}} q^{\frac{r}{2}} + \left(n^{\frac{\gamma (r-2)}{2}}q^{\frac{r}{2}} \right) \left(n^{\frac{\gamma(\nu-r)}{2}+1} q^{\frac{\nu-r}{2}}\right)^{\frac{1}{\nu+1-r}} \lesssim n^{\frac{\gamma r}{2}} q^{\frac{r}{2}},
\end{align*}
for sufficiently large $n$ and $\nu$. Hence, by induction, $\mathbb{E}[\|\Lambda_n \|^r]= O(n^{\frac{\gamma r}{2}}q^{\frac{r}{2}})$ for any $r\geq 1$, if $q=O(n^{\gamma-\varepsilon_0})$. 
\end{proof}

\begin{lemma}\label{Lemma_sequence}
	Let $\{x_m\}_{m\geq 1}$ be a non-negative sequence and let $h_j(q)$, $j=1,\cdots,J$ be non-negative functions of $q$. Suppose that there exist constants \(d>0\) and \(m_0\ge1\) such that $1-dm_0^{-\gamma} >0$ and, for all \(m\ge m_0\), 
	\begin{align*}
		x_{m+1} \leq (1-dm^{-\gamma})x_m + \sum_{j=1}^J b_{m,j} h_j(q),
	\end{align*}
	where $0<\gamma<1$, $b_{m,j} = O(m^{a_j})$, $j=1,\cdots, J$ for some constants \(a_j\) satisfying $a_j + \gamma\geq 0$, $j=1,\cdots, J$. Then, for $m\geq m_0$, 
	\begin{align*}
		x_m	=O\left(x_{m_0} + \sum_{j=1}^Jh_j(q)m^{a_j+\gamma} \right).
	\end{align*}
\end{lemma}

\begin{proof}[Proof of Lemma \ref{Lemma_sequence}: ]
	(i) Since \(b_{m,j}=O(m^{a_j})\), for each \(j=1,\ldots,J\), there exists a constant \(K_{b,j}>0\) such that $b_{m,j} \leq K_{b,j} m^{a_j}$. Choose constants $K_j>0$ such that $dK_j \geq K_{b,j}$, \(j=1,\ldots,J\). We claim that
	\begin{align}\label{Lemma_sequence_eq1}
		x_m\leq x_{m_0} + \sum_{j=1}^J K_jh_j(q)m^{a_j+\gamma},
		\qquad m\ge m_0.
	\end{align}
	We prove (\ref{Lemma_sequence_eq1}) by induction. The claim is immediate for \(m=m_0\), since \(x_{m_0}\ge0\) and \(h_j(q)\ge0\). Suppose that (\ref{Lemma_sequence_eq1}) holds for some \(m\ge m_0\). Then,
	\begin{align*}
		x_{m+1} &\leq (1-dm^{-\gamma})x_{m_0} + \sum_{j=1}^J (1-dm^{-\gamma})K_j h_j(q)m^{a_j+\gamma}+ \sum_{j=1}^J K_{b,j} h_j(q)m^{a_j}  \\
		&= (1-dm^{-\gamma})x_{m_0} + \sum_{j=1}^J K_j h_j(q)m^{a_j+\gamma} - \sum_{j=1}^J (dK_j - K_{b,j})h_j(q) m^{a_j} \\
		&\leq x_{m_0} + \sum_{j=1}^J K_j h_j(q) (m+1)^{a_j+\gamma}.
	\end{align*}
	Thus, (\ref{Lemma_sequence_eq1}) also holds for \(m+1\), and the induction is complete.
\end{proof}

\begin{lemma}\label{Lemma_Lambdaop}
Suppose that Assumptions \ref{assump_iid}-\ref{assump_boundSigma} hold. Under IR-CAR, \\
(i) If $\gamma > 1/2$ and $q = o(n^{\gamma-1/2})$, $\| \mathbb{E}[\Lambda_n\Lambda_n^\top]\|_{\mathrm{op}} = O(n^{1-\gamma}q^2) = O(n^\gamma)$. \\
(ii) Furthermore, if Assumption \ref{assump_phi1} holds for all $\iota \geq 4$ and the allocation function is given by (\ref{prop_ellx}). Then, if $q = o(n^{\gamma-\varepsilon})$ for some $\varepsilon>0$, $\| \mathbb{E}[\Lambda_n\Lambda_n^\top]\|_{\mathrm{op}} = O(n^{\gamma})$. 
\end{lemma}

\begin{proof}[Proof of Lemma \ref{Lemma_Lambdaop}: ]
Note that
\begin{align*}
\mathbb{E}\left[\Lambda_{m+1}\Lambda_{m+1}^\top\middle| \mathcal{F}_{m} \right] =&\ \Lambda_m\Lambda_m^\top + \mathbb{E}\left[(\ell_{m+1} - \pi) \left\{ \Lambda_{m} \phi(X_{m+1})^\top + \phi(X_{m+1})\Lambda_{m}^\top \right\}  \middle| \mathcal{F}_{m} \right] \\
&\ + \mathbb{E}\left[(T_{m+1} - \pi)^2 \phi(X_{m+1})\phi(X_{m+1})^\top \middle| \mathcal{F}_{m} \right]. 
\end{align*}
For the last term, we have
\begin{align*}
&\ u^\top \mathbb{E}\left[(T_{m+1} - \pi)^2 \phi(X_{m+1})\phi(X_{m+1})^\top \middle| \mathcal{F}_{m} \right]u \\ 
=&\ \mathbb{E}\left[(T_{m+1}-\pi)^2\left\{u^\top\phi(X_{m+1}) \right\}^2\middle| \mathcal{F}_m\right] \leq
\mathbb{E}\left[\left\{u^\top\phi(X_{m+1}) \right\}^2\right] 
	=	u^\top\boldsymbol{\Sigma}_{\phi\phi} u,
\end{align*}
which implies that
\begin{align*}
\left\|  \mathbb{E}\left[(T_{m+1} - \pi)^2 \phi(X_{m+1})\phi(X_{m+1})^\top \right] \right\|_{\mathrm{op}} \leq \left\| \boldsymbol{\Sigma}_{\phi\phi} \right\|_{\mathrm{op}} = \lambda_{\max}(\boldsymbol{\Sigma}_{\phi\phi}). 
\end{align*}
We write $h_m = \ell_{m} - \pi - \ell^\prime(0)(m-1)^{-\gamma}\Lambda_{m-1}^\top \phi(X_{m})$ for $m=2,\cdots,n$. By Taylor's expansion, we have
\begin{align}\label{ell_2}
	\left| \ell(x) - \pi -\ell^\prime(0)x \right|\leq \overline{\lambda}_{\ell,2}x^2. 
\end{align}
Therefore, we have $|h_m| \leq \overline{\lambda}_{\ell,2} (m-1)^{-2\gamma} \{\Lambda_{m-1}^\top \phi(X_{m})\}^2$. \\
We now consider the linear term. Note that
\begin{align*}
&\ \mathbb{E}\left[ \Lambda_{m}^\top \phi(X_{m+1}) \phi(X_{m+1}) \middle| \mathcal{F}_m\right] = \mathbb{E}\left[  \phi(X_{m+1}) \Lambda_{m}^\top \phi(X_{m+1}) \middle| \Lambda_{m}\right] \\
=&\ \mathbb{E}\left[  \phi(X_{m+1}) \phi(X_{m+1})^\top  \Lambda_{m} \middle| \Lambda_{m} \right] = \boldsymbol{\Sigma}_{\phi\phi} \Lambda_{m}.
\end{align*}
Then, let $\boldsymbol{B}_m = \boldsymbol{I} + \ell^\prime(0) m^{-\gamma} \boldsymbol{\Sigma}_{\phi\phi}$, we have 
\begin{align*}
\boldsymbol{B}_m \Lambda_m \Lambda_m^\top \boldsymbol{B}_m =&\  \Lambda_m \Lambda_m^\top + \frac{\ell^\prime(0)}{m^{\gamma}}  \boldsymbol{\Sigma}_{\phi\phi} \Lambda_m \Lambda_m^\top + \frac{\ell^\prime(0)}{m^{\gamma}} \Lambda_m \Lambda_m^\top\boldsymbol{\Sigma}_{\phi\phi} + \left(\frac{\ell^\prime(0)}{m^{\gamma}}\right)^2 \boldsymbol{\Sigma}_{\phi\phi} \Lambda_m \Lambda_m^\top\boldsymbol{\Sigma}_{\phi\phi} \\
=&\ \Lambda_m\Lambda_m^\top + \mathbb{E}\left[(\ell_{m+1} - \pi) \left\{ \Lambda_{m} \phi(X_{m+1})^\top + \phi(X_{m+1})\Lambda_{m}^\top \right\}  \middle| \mathcal{F}_{m} \right] \\
&\ + \left(\frac{\ell^\prime(0)}{m^{\gamma}}\right)^2 \boldsymbol{\Sigma}_{\phi\phi} \Lambda_m \Lambda_m^\top\boldsymbol{\Sigma}_{\phi\phi} - \boldsymbol{R}_m,
\end{align*}
where $\boldsymbol{R}_m = \mathbb{E}[h_{m+1} \{ \Lambda_{m} \phi(X_{m+1})^\top + \phi(X_{m+1})\Lambda_{m}^\top \} | \Lambda_{m}]$. 
For any unit vector $v\in\mathbb{R}^q$, 
\begin{align*}
\left| v^\top \boldsymbol{R}_m v\right|  =&\ 2\left| \mathbb{E}\left[ h_{m+1} \left(v^\top \Lambda_{m}\right) \left\{v^\top \phi(X_{m+1}) \right\}  \middle| \Lambda_m \right] \right| \\
\leq&\ \frac{2\overline{\lambda}_{\ell,2}}{m^{2\gamma}} \mathbb{E}\left[  \left\{\Lambda_{m}^\top \phi(X_{m+1})\right\}^2 \left| v^\top \Lambda_{m}\right| \left| v^\top \phi(X_{m+1}) \right|\, \middle| \Lambda_{m} \right] \\
\overset{\mathrm{(i)}}{\leq}&\ \frac{2\overline{\lambda}_{\ell,2}}{m^{2\gamma}} \left(\mathbb{E}\left[  \left\{\Lambda_{m}^\top \phi(X_{m+1})\right\}^4 \middle| \Lambda_{m} \right]\right)^{1/2} \left|v^\top \Lambda_{m} \right| \left(\mathbb{E}\left[\left|v^\top \phi(X_{m+1}) \right|^2 \middle| \Lambda_{m} \right]\right)^{1/2} \\
\leq&\ \frac{2\overline{\lambda}_{\ell,2}}{m^{2\gamma}} M_2^{1/2} M_4^{1/2}\left\| \Lambda_{m} \right\|^2 \left|v^\top \Lambda_{m} \right|,
\end{align*}
where (i) follows from Cauchy-Schwarz's inequality. It follows that
\begin{align*}
\left| v^\top \mathbb{E}[\boldsymbol{R}_m] v\right| \leq&\  \frac{2\overline{\lambda}_{\ell,2}}{m^{2\gamma}} M_2^{1/2} M_4^{1/2}\mathbb{E}\left[\left\| \Lambda_{m} \right\|^2 \left|v^\top \Lambda_{m} \right| \right] \\
\leq&\ \frac{2\overline{\lambda}_{\ell,2}}{m^{2\gamma}} M_2^{1/2} M_4^{1/2} \left( \mathbb{E}\left[\left\| \Lambda_{m} \right\|^4  \right]\right)^{1/2}  \left( v^\top \mathbb{E}\left[\Lambda_m\Lambda_m^\top \right] v  \right)^{1/2} \\
\leq&\ \frac{2\overline{\lambda}_{\ell,2}}{m^{2\gamma}} M_2^{1/2} M_4^{1/2} \left( \mathbb{E}\left[\left\| \Lambda_{m} \right\|^4  \right]\right)^{1/2} \left\|\mathbb{E}\left[\Lambda_m\Lambda_m^\top\right]\right\|_{\mathrm{op}}^{1/2} 
\end{align*}
that is, 
\begin{align*}
\left\| \mathbb{E}[\boldsymbol{R}_m]\right\|_{\mathrm{op}} \leq \frac{2\overline{\lambda}_{\ell,2}}{m^{2\gamma}} M_2^{1/2} M_4^{1/2} \left( \mathbb{E}\left[\left\| \Lambda_{m} \right\|^4  \right] \right)^{1/2} \left\|\mathbb{E}\left[\Lambda_m\Lambda_m^\top\right]\right\|_{\mathrm{op}}^{1/2}. 
\end{align*}
Recall the notation and conclusions in Lemma \ref{Lemma_qusi_smallball}, we have 
\begin{align*}
&\ \boldsymbol{U} \boldsymbol{B}_m \Lambda_m \Lambda_m^\top \boldsymbol{B}_m \boldsymbol{U}^\top = \left(\boldsymbol{U} \boldsymbol{B}_m \boldsymbol{U}^\top\right) \left(\boldsymbol{U}\Lambda_m\right) \left(\boldsymbol{U}\Lambda_m\right)^\top \left(\boldsymbol{U} \boldsymbol{B}_m \boldsymbol{U}^\top\right) \\
=&\ \left(\boldsymbol{I} + m^{-\gamma} \ell^\prime(0) \operatorname{diag}\{\lambda_1,\cdots,\lambda_{q'}\}\right) \widetilde{\Lambda}_{m(1)} \widetilde{\Lambda}_{m(1)}^\top \left(\boldsymbol{I} + m^{-\gamma} \ell^\prime(0) \operatorname{diag}\{\lambda_1,\cdots,\lambda_{q'}\}\right).
\end{align*}
It follows that if $m \geq \{-\ell^\prime(0)\lambda_{\max}(\boldsymbol{\Sigma}_{\phi\phi})\}^{1/\gamma}$, 
\begin{align*}
&\ \left\| \mathbb{E}\left[\boldsymbol{B}_m \Lambda_m \Lambda_m^\top \boldsymbol{B}_m \right]\right\|_{\mathrm{op}} = \left\| \mathbb{E}\left[\boldsymbol{U}\boldsymbol{B}_m \Lambda_m \Lambda_m^\top \boldsymbol{B}_m\boldsymbol{U}^\top \right]\right\|_{\mathrm{op}} \\
\leq&\ \left\|\boldsymbol{I} + m^{-\gamma} \ell^\prime(0) \operatorname{diag}\{\lambda_1,\cdots,\lambda_{q'}\}\right\|_{\mathrm{op}}^2  \left\| \mathbb{E}\left[ \widetilde{\Lambda}_{m(1)}\widetilde{\Lambda}_{m(1)}^\top \right] \right\|_{\mathrm{op}} \\
=&\ (1 + m^{-\gamma}\ell^\prime(0)\lambda_0 )^2 \left\| \mathbb{E}\left[ \Lambda_{m}\Lambda_{m}^\top \right] \right\|_{\mathrm{op}} \leq (1 + m^{-\gamma}\ell^\prime(0)\lambda_0) \left\| \mathbb{E}\left[ \Lambda_{m}\Lambda_{m}^\top \right] \right\|_{\mathrm{op}}. 
\end{align*}
Therefore, for $m \geq \{-\ell^\prime(0)\lambda_{\max}(\boldsymbol{\Sigma}_{\phi\phi})\}^{1/\gamma}$, 
\begin{align*}
 \left\| \mathbb{E}\left[ \Lambda_{m+1}\Lambda_{m+1}^\top \right] \right\|_{\mathrm{op}} 
\leq&\ (1 + m^{-\gamma}\ell^\prime(0)\lambda_0 ) \left\| \mathbb{E}\left[ \Lambda_{m}\Lambda_{m}^\top \right] \right\|_{\mathrm{op}} + \lambda_{\max}(\boldsymbol{\Sigma}_{\phi\phi}) \\
&\ +\frac{2\overline{\lambda}_{\ell,2}}{m^{2\gamma}} M_2^{1/2} M_4^{1/2} \left( \mathbb{E}\left[\left\| \Lambda_{m} \right\|^4  \right]\right)^{1/2} \left\|\mathbb{E}\left[\Lambda_m\Lambda_m^\top\right]\right\|_{\mathrm{op}}^{1/2} \\
\overset{(\mathrm{i})}{\leq}&\ (1 + m^{-\gamma}\ell^\prime(0)\lambda_0 ) \left\| \mathbb{E}\left[ \Lambda_{m}\Lambda_{m}^\top \right] \right\|_{\mathrm{op}} + \lambda_{\max}(\boldsymbol{\Sigma}_{\phi\phi}) \\
&\ - \frac{1}{2} \ell^\prime(0) \lambda_0 m^{-\gamma}\left\| \mathbb{E}\left[ \Lambda_{m}\Lambda_{m}^\top \right] \right\|_{\mathrm{op}} - \frac{2\overline{\lambda}_{\ell,2}M_2M_4}{\ell^\prime(0)\lambda_0} m^{-3\gamma}\mathbb{E}\left[\left\| \Lambda_{m} \right\|^4  \right]\\
\overset{(\mathrm{ii})}{\leq}&\ \left(1 + \frac{1}{2} m^{-\gamma}\ell^\prime(0)\lambda_0 \right) \left\| \mathbb{E}\left[ \Lambda_{m}\Lambda_{m}^\top \right] \right\|_{\mathrm{op}} + \lambda_{\max}(\boldsymbol{\Sigma}_{\phi\phi}) + C_{\Lambda} m^{1-2\gamma}q^2,
\end{align*}
for some constant $C_\Lambda >0$, where (i) follows from AM–GM inequality and (ii) follows from Equation (\ref{proof_RCAR_Er1}). In addition,
\begin{align*}
\max_{1\leq m \leq \{-\ell^\prime(0)\lambda_{\max}(\boldsymbol{\Sigma}_{\phi\phi})\}^{1/\gamma}} \left\| \mathbb{E}\left[ \Lambda_{m}\Lambda_{m}^\top \right] \right\|_{\mathrm{op}} = O(q).
\end{align*}
Then, by Lemma \ref{Lemma_sequence}, we have $\| \mathbb{E}[\Lambda_n\Lambda_n^\top] \|_{\mathrm{op}} = O(n^{1-\gamma}q^2+ n^{\gamma}) = O(n^{\gamma})$. \\
Furthermore, if Assumption \ref{assump_phi1} holds for all $\nu \geq 4$ and the allocation function is given by (\ref{prop_ellx}). Then, for any $K\geq 2$, there exists a constant $\overline{\lambda}_{\ell,K}$ such that
\begin{align}\label{ell_K}
	\left| \ell(x) - \pi -\ell^\prime(0)x \right|\leq \overline{\lambda}_{\ell,K}|x|^K. 
\end{align}
Therefore, we have $|h_m| \leq \overline{\lambda}_{\ell,K} (m-1)^{-K\gamma} \{\Lambda_{m-1}^\top \phi(X_{m})\}^K$. Then, for any unit vector $v\in\mathbb{R}^q$, 
\begin{align*}
	\left| v^\top \boldsymbol{R}_m v\right|  =&\ 2\left| \mathbb{E}\left[ h_{m+1} \left(v^\top \Lambda_{m}\right) \left\{v^\top \phi(X_{m+1}) \right\}  \middle| \Lambda_m \right] \right| \\
	\leq&\ \frac{2\overline{\lambda}_{\ell,K}}{m^{K\gamma}} \mathbb{E}\left[  \left\{\Lambda_{m}^\top \phi(X_{m+1})\right\}^K \left| v^\top \Lambda_{m}\right| \left| v^\top \phi(X_{m+1}) \right|\, \middle| \Lambda_{m} \right] \\
	\overset{\mathrm{(i)}}{\leq}&\ \frac{2\overline{\lambda}_{\ell,K}}{m^{K\gamma}} \left(\mathbb{E}\left[  \left\{\Lambda_{m}^\top \phi(X_{m+1})\right\}^{2K} \middle| \Lambda_{m} \right]\right)^{1/2} \left|v^\top \Lambda_{m} \right| \left(\mathbb{E}\left[\left|v^\top \phi(X_{m+1}) \right|^2 \middle| \Lambda_{m} \right]\right)^{1/2} \\
	\leq&\ \frac{2\overline{\lambda}_{\ell,K}}{m^{K\gamma}} M_2^{1/2} M_{2K}^{1/2}\left\| \Lambda_{m} \right\|^K \left|v^\top \Lambda_{m} \right|,
\end{align*}
which implies that 
\begin{align*}
\left\| \mathbb{E}\left[ \boldsymbol{R}_m \right] \right\|_{\mathrm{op}} \leq \frac{2\overline{\lambda}_{\ell,K}}{m^{K\gamma}} M_2^{1/2} M_{2K}^{1/2} \left( \mathbb{E}\left[\left\| \Lambda_{m} \right\|^{2K}  \right] \right)^{1/2} \left\|\mathbb{E}\left[\Lambda_m\Lambda_m^\top\right]\right\|_{\mathrm{op}}^{1/2}. 
\end{align*}
Therefore, for $m \geq \{-\ell^\prime(0)\lambda_{\max}(\boldsymbol{\Sigma}_{\phi\phi})\}^{1/\gamma}$, 
\begin{align*}
	\left\| \mathbb{E}\left[ \Lambda_{m+1}\Lambda_{m+1}^\top \right] \right\|_{\mathrm{op}} 
	\leq&\ (1 + m^{-\gamma}\ell^\prime(0)\lambda_0 ) \left\| \mathbb{E}\left[ \Lambda_{m}\Lambda_{m}^\top \right] \right\|_{\mathrm{op}} + \lambda_{\max}(\boldsymbol{\Sigma}_{\phi\phi}) \\
	&\ +\frac{\overline{\lambda}_{\ell,K}}{m^{K\gamma}} M_2^{1/2} M_{2K}^{1/2} \left( \mathbb{E}\left[\left\| \Lambda_{m} \right\|^{2K}  \right]\right)^{1/2} \left\|\mathbb{E}\left[\Lambda_m\Lambda_m^\top\right]\right\|_{\mathrm{op}}^{1/2} \\
	\leq &\ (1 + m^{-\gamma}\ell^\prime(0)\lambda_0 ) \left\| \mathbb{E}\left[ \Lambda_{m}\Lambda_{m}^\top \right] \right\|_{\mathrm{op}} + \lambda_{\max}(\boldsymbol{\Sigma}_{\phi\phi}) \\
	&\ - \frac{1}{2} \ell^\prime(0) \lambda_0 m^{-\gamma}\left\| \mathbb{E}\left[ \Lambda_{m}\Lambda_{m}^\top \right] \right\|_{\mathrm{op}} - \frac{2\overline{\lambda}_{\ell,K}M_2M_{2K}}{\ell^\prime(0)\lambda_0} m^{-(2K-1)\gamma}\mathbb{E}\left[\left\| \Lambda_{m} \right\|^{2K}  \right]\\
	\leq&\ \left(1 + \frac{1}{2} m^{-\gamma}\ell^\prime(0)\lambda_0 \right) \left\| \mathbb{E}\left[ \Lambda_{m}\Lambda_{m}^\top \right] \right\|_{\mathrm{op}} + \lambda_{\max}(\boldsymbol{\Sigma}_{\phi\phi}) + C_{\Lambda,K} m^{1-K\gamma}q^K,
\end{align*}
for some constant $ C_{\Lambda,K} >0$. \\
Then, by Lemma \ref{Lemma_sequence}, $\| \mathbb{E}[\Lambda_n\Lambda_n^\top]\|_{\mathrm{op}} = O(n^{1-(K-1)\gamma}q^{K} + n^{\gamma}) = O(n^{\gamma})$ for sufficiently large $K$. 
\end{proof}

\begin{proof}[Proof of Theorem \ref{Thm_additional covariate}:]
	\textbf{Proof of (i):} Let $\Delta M_i = (T_i-\pi) Z_i - \mathbb{E}[(\ell_i - \pi)Z_i | \Lambda_{i-1}]$ and $\mathcal{F}_m^{ZW}$ be the $\sigma$-field generated by $\{T_i,X_i,Z_i,W_i\}_{i=1}^m$. Then, $\{\Delta M_i,\mathcal{F}_i^{ZW} \}_{i=1}^n$ is a martingale difference sequence, which implies that
	\begin{align*}
		\frac{1}{n} \sum_{i=1}^n \Delta M_i \xrightarrow{P} 0.
	\end{align*}
	Note that $|\ell(x) -\pi | \leq \overline{\lambda}_{\ell,1}|x|$ and $|\ell(x)-\pi| \leq 1$. Then, for all $M>0$,
	\begin{align*}
		\left| \mathbb{E}\left[(\ell_i - \pi)Z_i | \Lambda_{i-1}\right]\right| \leq&\ M \mathbb{E}\left[ \left|\ell_i - \pi \right| \middle| \Lambda_{i-1} \right] + \mathbb{E}\left[|Z|\mathbb{I}\{|Z|\geq M\} \right] \\
		\leq&\ \frac{M\overline{\lambda}_{\ell,1}}{(i-1)^{\gamma}} \left(\mathbb{E}\left[\left(\Lambda_{i-1}^\top\phi(X_i) \right)^2 \middle| \Lambda_{i-1} \right]\right)^{1/2} + \mathbb{E}\left[|Z|\mathbb{I}\{|Z|\geq M\} \right] \\
		\leq&\ M\overline{\lambda}_{\ell,1} M_2^{1/2} \frac{\|\Lambda_{i-1} \|}{(i-1)^{\gamma}} + \mathbb{E}\left[|Z|\mathbb{I}\{|Z|\geq M\} \right].
	\end{align*}
	It implies that 
	\begin{align}\label{eq_Proof_RCAR_addition1}
		\frac{1}{n}\left|\sum_{i=2}^n \mathbb{E}\left[(\ell_i - \pi)Z_i | \Lambda_{i-1}\right]\right| \leq&\ \frac{1}{n}\sum_{i=2}^{n} M\overline{\lambda}_{\ell,1} M_2^{1/2} \frac{\|\Lambda_{i-1} \|}{(i-1)^{\gamma}} + \frac{n-1}{n} \mathbb{E}\left[|Z|\mathbb{I}\{|Z|\geq M\} \right] \notag \\
		=&\ \mathbb{E}\left[|Z|\mathbb{I}\{|Z|\geq M\} \right] + o_P(1), 
	\end{align}
	where the last equation follows from the fact that if $q=o( n^{\gamma - \frac{ 2(1-\gamma) }{3\nu-4}})$,
	\begin{align*}
		\mathbb{E}\left[\| \Lambda_n\| \right] \leq \left(\mathbb{E}\left[\| \Lambda_n\|^2 \right]\right)^{1/2} = o(n^\gamma). 
	\end{align*}
	Then, let $M\to\infty$, we prove (i). \\
	\textbf{Proof of (ii):} Let $\widetilde{Z}_i = Z_i - \phi(X_i)^\top b_{\phi,Z}$, where $b_{\phi,Z} = \arg\min_{b} \mathbb{E}[(Z-\phi(X)^\top b)^2]$ and $\Delta \widetilde{M}_i = (T_i-\pi) \widetilde{Z}_i - \mathbb{E}[(\ell_i - \pi)\widetilde{Z}_i | \Lambda_{i-1}]$. Then, $\{\Delta \widetilde{M}_i,\mathcal{F}_i^{ZW} \}_{i=1}^n$ is a martingale difference sequence and 
	\begin{align*}
		\sum_{i=1}^n (T_i - \pi)Z_i = \sum_{i=1}^n \Delta \widetilde{M}_i + \Lambda_n^\top b_{\phi,Z} + \sum_{i=2}^{n} \mathbb{E}\left[ \left(\ell_i - \pi\right)\widetilde{Z}_i \middle| \Lambda_{i-1} \right],
	\end{align*}
	By the definition of $b_{\phi,Z}$, we have $\boldsymbol{\Sigma}_{\phi\phi}b_{\phi,Z} =\mathbb{E}\left[\phi(X)Z\right]$. Recall the notation and conclusions in Lemma \ref{Lemma_qusi_smallball} and write that $\boldsymbol{U} b_{\phi,Z} = (\widetilde{b}_{\phi,Z(1)}^\top, \widetilde{b}_{\phi,Z(2)}^\top)^\top$, where $\widetilde{b}_{\phi,Z(1)} \in \mathbb{R}^{q'}$. Then, we have
	\begin{align*}
		\mathbb{E}\left[\widetilde{\phi}(X)Z\right] = \boldsymbol{U}\boldsymbol{\Sigma}_{\phi\phi}\boldsymbol{U}^\top \left(\boldsymbol{U} b_{\phi,Z} \right) = \operatorname{diag}\{\lambda_1,\cdots,\lambda_{q'},0,\cdots,0\} \left(\boldsymbol{U} b_{\phi,Z} \right), 
	\end{align*}
	which implies that
	\begin{align*}
		\mathbb{E}\left[\widetilde{\phi}_{(1)}(X)Z\right] = \operatorname{diag}\{\lambda_1,\cdots,\lambda_{q'}\} \, \widetilde{b}_{\phi,Z(1)},
	\end{align*}
     Note that $\| \widetilde{\Lambda}_{n(1)} \| = \| \widetilde{\Lambda}_{n} \|= \| \Lambda_{n} \|$ and 
	\begin{align*}
		\left\| \mathbb{E}\left[\widetilde{\phi}_{(1)}(X)Z\right] \right\| =&\ \sup_{\| u \|= 1} \left| \mathbb{E}\left[u^\top \widetilde{\phi}_{(1)}(X)Z\right]\right| \leq \sup_{\| u \|= 1}\left( \mathbb{E}\left[\left(u^\top \widetilde{\phi}_{(1)}(X)\right)^2\right]\right)^{1/2} \left(\mathbb{E}[Z^2]\right)^{1/2} \\
		=&\ \sup_{1\leq i \leq q'} \lambda_i \left(\mathbb{E}[Z^2]\right)^{1/2}. 
	\end{align*}
Then, we have $\| \widetilde{b}_{\phi,Z(1)} \| = O(1)$. Therefore, note that $q = o(n^{\gamma - 1/2})$, by Lemma \ref{Lemma_Lambdaop},
\begin{align*}
&\ \mathbb{E}\left[ \left(\Lambda_n^\top b_{\phi,Z}\right)^2 \right] = \mathbb{E}\left[\left\{ \left(\boldsymbol{U}\Lambda_n\right)^\top \left(\boldsymbol{U}b_{\phi,Z}\right)  \right\}^2 \right] = \mathbb{E}\left[ \left(\widetilde{\Lambda}_{n(1)}^\top \widetilde{b}_{\phi,Z(1)}\right)^2 \right] \\
\leq&\ \left\|\widetilde{b}_{\phi,Z(1)} \right\|^2 \left\| \mathbb{E}\left[ \widetilde{\Lambda}_{n(1)}\widetilde{\Lambda}_{n(1)}^\top \right] \right\|_{\mathrm{op}} = \left\|\widetilde{b}_{\phi,Z(1)} \right\|^2 \left\| \mathbb{E}\left[\Lambda_n\Lambda_n^\top  \right] \right\|_{\mathrm{op}} = O(n^{1-\gamma}q^2) = o(n), 
\end{align*}
which implies that $\Lambda_n^\top b_{\phi,Z}  = o_P(\sqrt{n})$. \\ 
In addition, by Equation (\ref{ell_2}) and Cauchy-Schwarz's inequality, we have
	\begin{align*}
		&\ \left| \mathbb{E}\left[  \left( \ell_{m+1} - \pi \right)\widetilde{Z}_{m+1}     \middle| \Lambda_m  \right] \right| = \left| \mathbb{E}\left[  \left( \ell_{m+1} - \pi - \ell^\prime(0) \frac{\Lambda_m^\top \phi(X_{m+1})}{m^{\gamma}} \right)\widetilde{Z}_{m+1}     \middle| \Lambda_m  \right] \right| \\
		\leq&\ \overline{\lambda}_{\ell,2} \mathbb{E}\left[ \frac{\left(\Lambda_m^\top \phi(X_{m+1})  \right)^2  }{m^{2\gamma} } \left|\widetilde{Z}_{m+1}\right|  \middle| \Lambda_m \right] \\
		\leq&\ \overline{\lambda}_{\ell,2} \frac{1}{m^{2\gamma}} \left(\mathbb{E}\left[\left(\Lambda_m^\top \phi(X_{m+1})  \right)^4\middle| \Lambda_m \right]\right)^{1/2} \left(\mathbb{E}[\widetilde{Z}^2] \right)^{1/2} \leq  \overline{\lambda}_{\ell,2} M_4^{1/2} \left(\mathbb{E}[\widetilde{Z}^2] \right)^{1/2} \frac{\|\Lambda_m\|^2}{m^{2\gamma}} , 
	\end{align*}
	which implies that if $q = o(n^{\gamma-1/2})$ with $\gamma>1/2$,  
	\begin{align*}
		\left|\sum_{i=2}^{n}\mathbb{E}\left[ \left(\ell_i - \pi\right)\widetilde{Z}_i \middle| \Lambda_{i-1} \right] \right| \leq&\ \overline{\lambda}_{\ell,2} M_4^{1/2} \left(\mathbb{E}[\widetilde{Z}^2]\right)^{1/2} \sum_{i=2}^{n} \frac{\| \Lambda_i \|^2}{i^{2\gamma}} \\
		=&\ O_P\left(n^{1-\gamma}q + n^{\frac{2\gamma + 2\nu - 3\gamma\nu}{2(\nu-1)}} q^{\frac{3\nu-4}{2(\nu-1)}}  \right) = O_P(n^{1-\gamma}q) = o_P(\sqrt{n}).
	\end{align*}
	A simple calculation yields that $\mathbb{E}[W_i^2|\mathcal{F}_{i-1}] = \sigma_W^2$ and by Equation (\ref{eq_Proof_RCAR_addition1}),
	\begin{align*}
		&\ \frac{1}{n\widetilde{\sigma}_Z^2} \sum_{i=1}^n \mathbb{E}\left[\left(\Delta \widetilde{M}_i\right)^2\middle| \mathcal{F}_{i-1}\right] \\
		=&\ \frac{1}{n\widetilde{\sigma}_Z^2} \sum_{i=1}^n (1-2\pi)\mathbb{E}\left[\left(\ell_i-\pi \right)\widetilde{Z}_i^2 \middle| \Lambda_{i-1} \right] 
		-\frac{1}{n\widetilde{\sigma}_Z^2} \sum_{i=1}^n \left(\mathbb{E}\left[(\ell_i-\pi)\widetilde{Z}_i\middle| \Lambda_{i-1} \right] \right)^2 + 1 \xrightarrow{P} 1, 
	\end{align*}
	and 
	\begin{align*}
		\frac{1}{n} \sum_{i=1}^n \mathbb{E}\left[ \Delta\widetilde{M}_i W_i\middle| \mathcal{F}_{i-1}\right] = \frac{1}{n}\sum_{i=1}^n \mathbb{E}\left[(\ell_i-\pi)\widetilde{Z}_i W_i\middle| \mathcal{F}_{i-1}\right] \xrightarrow{P} 0.
	\end{align*}
	Furthermore, Lindeberg's condition holds for two-dimensional martingale difference sequence $\{(\Delta \widetilde{M}_i/\widetilde{\sigma}_Z, W_i/\sigma_W )^\top,\mathcal{F}_i \}$, which follows from
	\begin{align*}
		&\ \frac{1}{n}\sum_{i=1}^n \mathbb{E}\left[\left(\Delta \widetilde{M}_i/\widetilde{\sigma}_Z\right)^2 \mathbb{I}\left\{\left(\Delta \widetilde{M}_i/\widetilde{\sigma}_Z\right)^2 \geq \epsilon n \right\} \right] \\
		\leq &\ \frac{1}{n\widetilde{\sigma}_Z^2} \sum_{i=1}^n \mathbb{E}\left[\left(|\widetilde{Z}_i|+\mathbb{E}[|\widetilde{Z}|] \right)^2\mathbb{I}\left\{\left(|\widetilde{Z}_i|+\mathbb{E}[|\widetilde{Z}|]\right) \geq \sqrt{\pi(1-\pi)c_{\phi,Z}\epsilon n} \right\}  \right]\\
		=&\ \frac{1}{\widetilde{\sigma}_Z^2}\mathbb{E}\left[\left(|\widetilde{Z}|+\mathbb{E}[|\widetilde{Z}|] \right)^2\mathbb{I}\left\{\left(|\widetilde{Z}|+\mathbb{E}[|\widetilde{Z}|]\right) \geq \sqrt{\pi(1-\pi)c_{\phi,Z}\epsilon n} \right\}  \right] \to 0, \\
		\text{and } &\ \frac{1}{n} \sum_{i=1}^n \mathbb{E}\left[W_i^2 \mathbb{I}\left\{W_i^2\geq \epsilon n \right\} \right] = \mathbb{E}\left[W^2 \mathbb{I}\left\{W^2\geq \epsilon n \right\} \right] \to 0. 
	\end{align*}
	Therefore, by central limit theorem for martingale difference sequence and Slutsky’s theorem, 
	\begin{align*}
		\frac{1}{\sqrt{n}} \sum_{i=1}^n
		\begin{pmatrix}
			(T_i-\pi)Z_i/\widetilde{\sigma}_{Z} \\
			W_i/\sigma_W
		\end{pmatrix} \xrightarrow{d}\mathcal{N}\left(\boldsymbol{0}_2,\boldsymbol{I}_2 \right). 
	\end{align*}
	\textbf{Proof of (iii):} We only need to prove that
	\begin{align*}
\Lambda_n^\top b_{\phi,Z}= o_P(\sqrt{n}) \quad \text{and} \quad		\sum_{i=1}^{n}\mathbb{E}\left[ \left(\ell_i - \pi\right)\widetilde{Z}_i \middle| \Lambda_{i-1} \right] = o_P(\sqrt{n}). 
	\end{align*}
The first equation follows from Lemma \ref{Lemma_Lambdaop} directly. Then, we only need to prove the second equation. For any $K\geq 2$, by Equation (\ref{ell_K}) and Cauchy-Schwarz's inequality, we have
	\begin{align*}
		&\ \left| \mathbb{E}\left[  \left( \ell_{m+1} - \pi \right)\widetilde{Z}_{m+1}     \middle| \Lambda_m  \right] \right| = \left| \mathbb{E}\left[  \left( \ell_{m+1} - \pi - \ell^\prime(0) \frac{\Lambda_m^\top \phi(X_{m+1})}{m^{\gamma}} \right)\widetilde{Z}_{m+1}     \middle| \Lambda_m  \right] \right| \\
		\leq&\ \overline{\lambda}_{\ell,K} \mathbb{E}\left[ \frac{\left|\Lambda_m^\top \phi(X_{m+1})  \right|^K  }{m^{K\gamma} } \left|\widetilde{Z}_{m+1}\right|  \middle| \Lambda_m \right] \\
		\leq&\ \overline{\lambda}_{\ell,K} \frac{1}{m^{K\gamma}} \left(\mathbb{E}\left[\left(\Lambda_m^\top \phi(X_{m+1})  \right)^{2K}\middle| \Lambda_m \right]\right)^{1/2} \left(\mathbb{E}[\widetilde{Z}^2] \right)^{1/2} \leq  \overline{\lambda}_{\ell,K} M_{2K}^{1/2} \left(\mathbb{E}[\widetilde{Z}^2] \right)^{1/2} \frac{\|\Lambda_m\|^K}{m^{K\gamma}}.
	\end{align*}
	By Lemma \ref{Lemma1}, let $K > \varepsilon_0^{-1}$, if $q=o(n^{\gamma-\varepsilon})$ for some $\varepsilon>0$, 
	\begin{align*}
		\left|\sum_{i=2}^{n}\mathbb{E}\left[ \left(\ell_i - \pi\right)\widetilde{Z}_i \middle| \Lambda_{i-1} \right] \right| \leq \overline{\lambda}_{\ell,K} M_{2K}^{1/2} \left(\mathbb{E}[\widetilde{Z}^2] \right)^{1/2} \sum_{i=2}^{n} \frac{\|\Lambda_i\|^K}{i^{K\gamma}} 
		= O_P(n^{1-\frac{K}{2}\gamma}q^{\frac{K}{2}}) = o_P(\sqrt{n}).  
	\end{align*}
\end{proof}
\section{Theoretical Properties of Treatment Effect under IR-CAR Procedures}
\subsection{Proof of Theorem \ref{Thm_ATE}}
\begin{proof}[Proof of Theorem \ref{Thm_ATE}: ]
	Let $D_n=\sum_{i=1}^n (T_i-\pi)$ denote the overall imbalance. To establish the stochastic order of $D_n$, we distinguish between two cases. First, suppose $1\in \mathrm{span}\{\phi(X)\}$, then by Lemma \ref{Lemma_Lambdaop}, $D_n = O_P(n^{\gamma/2}) = o_P(\sqrt{n})$. Next, suppose that $1\notin\operatorname{span}\{\phi(X)\}$, then by Theorem \ref{Thm_additional covariate}, $D_n = O_P(\sqrt{n})$. Combining the two cases yields $D_n = O_P(\sqrt{n})$. Let $\widetilde{Y}_{i}(a) = Y_i(a) - \mathbb{E}[Y(a)]$ and $\widetilde{Y}(a) = Y(a) - \mathbb{E}[Y(a)]$ for $a=0,1$, we have
	\begin{align*}
		\widehat{\tau} = \Bar{Y}_1 -\Bar{Y}_0 =&\ \tau + \left(\pi + \frac{D_n}{n}\right)^{-1} \left(\frac{1}{n}\sum_{i=1}^n \pi \widetilde{Y}_{i}(1) + \frac{1}{n} \sum_{i=1}^n (T_i-\pi)\widetilde{Y}_{i}(1) \right) \\
		&\ - \left(1 - \pi - \frac{D_n}{n}\right)^{-1} \left(\frac{1}{n}\sum_{i=1}^n \left(1-\pi\right) \widetilde{Y}_{i}(0) - \frac{1}{n} \sum_{i=1}^n (T_i-\pi)\widetilde{Y}_{i}(0) \right),
	\end{align*}
	which implies that
	\begin{align*}
		\sqrt{n}\left( \widehat{\tau} - \tau \right) = \frac{1}{\sqrt{n}} \sum_{i=1}^n \left(T_i-\pi\right)\left( \frac{\widetilde{Y}_{i}(1)}{\pi} + \frac{\widetilde{Y}_{i}(0)}{1-\pi} \right)  + \frac{1}{\sqrt{n}} \sum_{i=1}^n \left(\widetilde{Y}_{i}(1) - \widetilde{Y}_{i}(0)\right) + o_P(1).
	\end{align*}
	Let $\Breve{Y}(a) = \widetilde{Y}(a) - \Pi_{\phi(X)}\{\widetilde{Y}(a)\}$ for $a=0,1$. A simple calculation yields that
	\begin{align*}
		&\ \pi(1-\pi)\mathbb{E}\left[\left(\frac{\Breve{Y}(1)}{\pi} + \frac{\Breve{Y}(0)}{1-\pi}\right)^2\right] + \mathbb{E}\left[\left(\widetilde{Y}(1) - \widetilde{Y}(0)\right)^2 \right] \\
		=&\ \pi(1-\pi)\mathbb{E}\left[\left(\frac{\Breve{Y}(1)}{\pi} + \frac{\Breve{Y}(0)}{1-\pi}\right)^2\right] + \mathbb{E}\left[\left(\Breve{Y}(1) - \Breve{Y}(0)\right)^2 \right] + \mathbb{E}\left[\left(\Pi_{\phi(X)}\left\{\widetilde{Y}(1)-\widetilde{Y}(0)\right\}\right)^2 \right] \\
		=&\ \varsigma_{\widetilde{Y}}^2(\pi) + \varsigma_{H\widetilde{Y}}^2. 
	\end{align*}
	Therefore, by Theorem \ref{Thm_additional covariate} and Slutsky’s theorem,
	\begin{align*}
		\sqrt{n}\left( \widehat{\tau} - \tau \right)\Big/ \sqrt{\varsigma_{\widetilde{Y}}^2(\pi) + \varsigma_{H\widetilde{Y}}^2} \xrightarrow{d} \mathcal{N}(0,1). 
	\end{align*}
\end{proof}
\subsection{Proof of Theorem \ref{Thm_ATE_test}}
In this subsection, we assume that $\boldsymbol{\Sigma}_{\phi\phi}$ is positive definite without loss of generality. Otherwise, we can use $\widetilde{\phi}_{(1)}(X)$ instead of $\phi(X)$ (see its definition in the proof of Lemma  \ref{Lemma_qusi_smallball}). Therefore, combining Assumptions \ref{assump_phi1} and \ref{assump_boundSigma}, we assume throughout this subsection that the eigenvalues of $\boldsymbol{\Sigma}_{\phi\phi}$ are bounded above and below. We first establish several technical lemmas.

\begin{lemma}\label{Lemma_vector}
Under the conditions of Theorem \ref{Thm_ATE_test}, let $\{V_i\}_{i=1}^n$ be a sequence of i.i.d. copies of $V$, where $V\in\mathbb{R}^d$ and define $K_{\phi,V} = \sup_{\|u\|=1,\|v\|=1} \mathbb{E}[|u^\top V |\cdot | v^\top \phi(X) |]$. Then,
\begin{align*}
\left\| \frac{1}{n} \sum_{i=1}^n (T_i-\pi)V_i \right\| = O_P\left(\sqrt{\frac{\mathbb{E}[\|V\|^2]}{n}}\right) + o(K_{\phi,V}). 
\end{align*}
Furthermore, if $\mathbb{E}[V] = \boldsymbol{0}$. Then, for $a=0,1$,
\begin{align*}
\left\| \frac{1}{n_a} \sum_{i=1}^n \mathbb{I}\{T_i=a\} V_i \right\| = O_P\left(\sqrt{\frac{\mathbb{E}[\|V\|^2]}{n}}\right) + o(K_{\phi,V}).
\end{align*}
\end{lemma}
\begin{proof}
We consider the following decomposition:
\begin{align*}
\sum_{i=1}^n (T_i-\pi)V_i = \sum_{i=1}^n (T_i-\pi)V_i - \mathbb{E}\left[\left(\ell_i - \pi \right)V_i\middle|\Lambda_{i-1} \right] + \sum_{i=1}^{n-1} \mathbb{E}\left[\left(\ell_{i+1} - \pi \right)V_{i+1}\middle|\Lambda_{i} \right].
\end{align*}
First, note that $\{ (T_i-\pi)V_i - \mathbb{E}\left[\left(\ell_i - \pi \right)V_i\middle|\Lambda_{i-1} \right], \mathcal{F}_i \}$ is a martingale difference sequence, which implies that
\begin{align*}
&\ \mathbb{E}\left[\left\| \sum_{i=1}^n\left\{ (T_i-\pi)V_i - \mathbb{E}\left[\left(\ell_i - \pi \right)V_i\middle|\Lambda_{i-1} \right]\right\} \right\|^2 \right] \\
=&\ \sum_{i=1}^n \mathbb{E}\left[\left\|(T_i-\pi)V_i - \mathbb{E}\left[\left(\ell_i - \pi \right)V_i\middle|\Lambda_{i-1} \right] \right\|^2 \right] \leq 2\sum_{i=1}^n \left(\mathbb{E}\left[\|V_i\|^2 + \left\| \mathbb{E}\left[ V_i\middle| \Lambda_{i-1}\right]\right\|^2\right] \right)\\
\leq&\ 2\sum_{i=1}^n \left(\mathbb{E}\left[\|V_i\|^2 + \mathbb{E}\left[ \|V_i\|^2\middle| \Lambda_{i-1}\right]\right]\right) = 4n \mathbb{E}[\|V\|^2].
\end{align*}
Then, by Markov inequality, we obtain that
\begin{align*}
\left\| \sum_{i=1}^n\left\{ (T_i-\pi)V_i - \mathbb{E}\left[\left(\ell_i - \pi \right)V_i\middle|\Lambda_{i-1} \right]\right\} \right\|^2 = O_P(\sqrt{n\mathbb{E}[\|V\|^2]}). 
\end{align*}
In addition, for every unit vector $u$,
\begin{align*}
&\ \left| u^\top\mathbb{E}\left[\left(\ell_{i+1} - \pi \right)V_{i+1}\middle|\Lambda_{i} \right] \right| \leq 
\mathbb{E}\left[ \left|\ell_{i+1} - \pi\right|\cdot \left|u^\top V_{i+1} \right|\, \middle| \Lambda_{i}  \right] \\
\leq&\ \overline{\lambda}_{\ell,1} \mathbb{E}\left[\left| \frac{\Lambda_{i}^\top \phi(X_{i+1})}{i^{\gamma}} \right| \cdot \left|u^\top V_{i+1} \right|\, \middle| \Lambda_{i} \right] \leq \frac{\overline{\lambda}_{\ell,1}K_{\phi,V}}{i^{\gamma}}\left\|\Lambda_i \right\|,
\end{align*}
which implies that
\begin{align*}
\left\| \mathbb{E}\left[\left(\ell_{i+1} - \pi \right)V_{i+1}\middle|\Lambda_{i} \right] \right\| \leq \frac{\overline{\lambda}_{\ell,1}K_{\phi,V}}{i^{\gamma}}\left\|\Lambda_i \right\|.
\end{align*}
Consequently,
\begin{align*}
\left\| \sum_{i=2}^n \mathbb{E}\left[\left(\ell_{i+1} - \pi \right)V_{i+1}\middle|\Lambda_{i} \right] \right\| \leq \overline{\lambda}_{\ell,1} K_{\phi,V} \sum_{i=1}^{n-1} \frac{\|\Lambda_i\|}{i^{\gamma}} =  o_P\left(nK_{\phi,V}\right). 
\end{align*}
Therefore, 
\begin{align*}
\left\| \frac{1}{n} \sum_{i=1}^n (T_i-\pi)V_i \right\| = O_P\left(\sqrt{\frac{\mathbb{E}[\|V\|^2]}{n}}\right) + o(K_{\phi,V}).   
\end{align*}
Next, we consider the following decomposition
\begin{align*}
\frac{1}{n_1} \sum_{i=1}^n T_iV_i = \frac{n}{n_1} \left[\frac{1}{n} \sum_{i=1}^n (T_i-\pi)V_i + \frac{\pi}{n} \sum_{i=1}^n V_i  \right]. 
\end{align*}
Note that $n_1/n \xrightarrow{P} \pi$ and by Markov's inequality,
\begin{align*}
\left\|  \sum_{i=1}^n V_i \right\| = O_P(\sqrt{n\mathbb{E}[\|V\|^2]}). 
\end{align*}
Therefore, 
\begin{align*}
\left\| \frac{1}{n_1} \sum_{i=1}^n T_iV_i  \right\| = O_P\left(\sqrt{\frac{\mathbb{E}[\|V\|^2]}{n}}\right) + o(K_{\phi,V}). 
\end{align*}
Similarly, we can prove that
\begin{align*}
\left\| \frac{1}{n_0} \sum_{i=1}^n (1-T_i)V_i \right\| = O_P\left(\sqrt{\frac{\mathbb{E}[\|V\|^2]}{n}}\right) + o(K_{\phi,V}). 
\end{align*}
\end{proof}

\begin{lemma}[Burkholder–Rosenthal Inequalities for Self-Adjoint Random Matrices]\label{Lemma_Burkholder}
For a $q\times q$ matrix $B$ and $r\geq 2$, define $|B|=(B^*B)^{1/2}$ and its Schatten $r$-norm by $\|B\|_{S_r} =\{\operatorname{tr}(|B|^r)\}^{1/r}$. For a random $q\times q$ matrix $B$, define $\|B\|_{L_r(S_r)}=\{\mathbb{E}\|B\|_{S_r}^r\}^{1/r}
=\{\mathbb{E}\operatorname{tr}(|B|^r)\}^{1/r}$. Then there exists a constant $K_r>0$, depending only on $r$, such that the following statements hold. \\
(i) Suppose that \(\{D_i,\mathcal G_i\}_{i=1}^n\) is a sequence of \(q\times q\) self-adjoint matrix-valued martingale differences such that $\mathbb{E}\|D_i\|_{S_r}^r<\infty$ for $i=1,\ldots,n$. Then,
\begin{align*}
 \left\| \sum_{i=1}^nD_i \right\|_{L_r(S_r)}  \leq K_r
\left[ \left\| \left\{
	\sum_{i=1}^n \mathbb{E}[D_i^2| \mathcal G_{i-1}] \right\}^{1/2}
	\right\|_{L_r(S_r)} +
	\left\{
	\sum_{i=1}^n  \mathbb{E} \|D_i\|_{S_r}^r
	\right\}^{1/r}
	\right].
\end{align*}
(ii) Suppose that \(D_1,\ldots,D_n\) are independent, mean-zero, self-adjoint \(q\times q\) random matrices such that $\mathbb{E}\|D_i\|_{S_r}^r<\infty$ for $i=1,\ldots,n$. Then,
\begin{align*}
\left\|
\sum_{i=1}^nD_i
\right\|_{L_r(S_r)}
\le
K_r
\left[ \left\|
\left\{
\sum_{i=1}^n \mathbb{E}[D_i^2]
\right\}^{1/2}
\right\|_{S_r} + 
\left\{
\sum_{i=1}^n
\mathbb{E} \|D_i\|_{S_r}^r
\right\}^{1/r} \right].
\end{align*}
\end{lemma}

\begin{proof}
Part (i) is a finite-dimensional self-adjoint special case of the noncommutative Burkholder–Rosenthal inequality \cite[Theorem 5.1]{junge2003noncommutative}. Part (ii) follows from part (i) by taking the natural filtration generated by the independent random matrices.
\end{proof}

\begin{lemma}\label{Lemma_matrix}
Under the conditions of Theorem \ref{Thm_ATE_test}, let $\widehat{\boldsymbol{\Sigma}}_{\phi\phi,a} = n_a^{-1} \sum_{i=1}^n \mathbb{I}\{T_i=a\}\phi(X_i)\phi(X_i)^\top$ for $a=0,1$, 
	\begin{align*}
		\left\|\widehat{\boldsymbol{\Sigma}}_{\phi\phi,a} - \boldsymbol{\Sigma}_{\phi\phi} \right\|_{\mathrm{op}} = o_P(1), \quad  \left\|\widehat{\boldsymbol{\Sigma}}_{\phi\phi,a}^{-1} - \boldsymbol{\Sigma}_{\phi\phi}^{-1} \right\|_{\mathrm{op}} = o_P(1).
	\end{align*}
\end{lemma}

\begin{proof}
It suffices to prove the result for $a=1$, as the case $a=0$ can be established by the same argument. Consider the following decomposition: 
\begin{align*}
\widehat{\boldsymbol{\Sigma}}_{\phi\phi,1} - \boldsymbol{\Sigma}_{\phi\phi} =&\ \frac{n}{n_1}\left[\frac{1}{n}\sum_{i=1}^n (T_i-\pi)\left(\phi(X_i)\phi(X_i)^\top - \boldsymbol{\Sigma}_{\phi\phi}\right) + \frac{\pi}{n} \sum_{i=1}^n \left(\phi(X_i)\phi(X_i)^\top - \boldsymbol{\Sigma}_{\phi\phi}\right) \right] \\
=&\ \frac{n}{n_1} \left[ \frac{1}{n} \sum_{i=1}^{n} \Delta M_{i,\phi}  + \frac{1}{n} \sum_{i=1}^{n-1} \mathbb{E}\left[(\ell_{i+1} - \pi) \left(\phi(X_{i+1})\phi(X_{i+1})^\top - \boldsymbol{\Sigma}_{\phi\phi}\right) \middle| \Lambda_{i} \right] \right. \\
 &\ + \left. \frac{\pi}{n} \sum_{i=1}^n \left(\phi(X_i)\phi(X_i)^\top - \boldsymbol{\Sigma}_{\phi\phi}\right) \right],
\end{align*}
where $\Delta M_{i,\phi} = (T_i-\pi)(\phi(X_i)\phi(X_i)^\top - \boldsymbol{\Sigma}_{\phi\phi}) - \mathbb{E}[(\ell_i - \pi) (\phi(X_i)\phi(X_i)^\top - \boldsymbol{\Sigma}_{\phi\phi}) | \Lambda_{i-1}]$.\\
Note that $\{ \Delta M_{i,\phi} , \mathcal{F}_i \}$ is a martingale difference sequence. 
\begin{align*}
&\ \left| u^\top \sum_{i=1}^n \mathbb{E}\left[ \left(\Delta M_{i,\phi}\right)^2  \middle| \mathcal{F}_{i-1} \right] u \right| \leq \left| u^\top \sum_{i=1}^n \mathbb{E}\left[(T_i-\pi)^2 (\phi(X_i)\phi(X_i)^\top - \boldsymbol{\Sigma}_{\phi\phi})^2 \middle| \mathcal{F}_{i-1} \right]u  \right| \\
\leq&\ n u^\top \mathbb{E}\left[(\phi(X)\phi(X)^\top - \boldsymbol{\Sigma}_{\phi\phi})^2\right]u =  n u^\top \left\{ \mathbb{E} \left[ \left\|\phi(X)\right\|^2\phi(X)\phi(X)^\top \right] -\boldsymbol{\Sigma}_{\phi\phi}^2 \right\} u.
\end{align*}
For any unit vector $u$,
\begin{align*}
&\ u^\top\mathbb{E} \left[ \left\|\phi(X)\right\|^2\phi(X)\phi(X)^\top \right]u = \mathbb{E} \left[\left\|\phi(X)\right\|^2 \left\{u^\top\phi(X) \right\}^2 \right] \\
\leq&\ \left\{\mathbb{E} \left[ \left\|\phi(X)\right\|^4\right]\right\}^{1/2} \left\{\mathbb{E} \left[\left\{u^\top\phi(X) \right\}^4 \right] \right\}^{1/2} = O(q),
\end{align*}
which implies that
\begin{align}\label{Eq_phi1}
\left\| \mathbb{E}\left[(\phi(X_i)\phi(X_i)^\top - \boldsymbol{\Sigma}_{\phi\phi})^2\right]  \right\|_{\mathrm{op}} \leq \left\| \mathbb{E} \left[ \left\|\phi(X)\right\|^2\phi(X)\phi(X)^\top \right] \right\|_{\mathrm{op}} + \left\| \boldsymbol{\Sigma}_{\phi\phi}^2 \right\|_{\mathrm{op}} = O(q). 
\end{align}
Therefore,
\begin{align*}
\left\|\sum_{i=1}^n \mathbb{E}\left[ \left(\Delta M_{i,\phi}\right)^2  \middle| \mathcal{F}_{i-1} \right] \right\|_{\mathrm{op}} = O(nq),
\end{align*}
which implies that
\begin{align*}
\left\|\left\{\sum_{i=1}^n \mathbb{E}\left[ \left(\Delta M_{i,\phi}\right)^2  \middle| \mathcal{F}_{i-1}\right]\right\}^{1/2} \right\|_{S_r} \leq q^{1/r} \left\|\sum_{i=1}^n \mathbb{E}\left[ \left(\Delta M_{i,\phi}\right)^2  \middle| \mathcal{F}_{i-1} \right] \right\|_{\mathrm{op}}^{1/2} = O(n^{1/2}q^{1/2+1/r}).  
\end{align*}
In addition, by Jensen's inequality, for $2\leq r\leq \nu/2$,
\begin{align*}
&\ \mathbb{E}\left\| \Delta M_{i,\phi} \right\|_{S_r}^r \\
=&\
\mathbb{E}\left\| (T_i-\pi)(\phi(X_i)\phi(X_i)^\top - \boldsymbol{\Sigma}_{\phi\phi}) - \mathbb{E}[(\ell_i - \pi) (\phi(X_i)\phi(X_i)^\top - \boldsymbol{\Sigma}_{\phi\phi}) | \Lambda_{i-1}] \right\|_{S_r}^r \\
\leq&\  2^{r-1} \left[
\mathbb{E}\left\|(T_i-\pi)\left(\phi(X_i)\phi(X_i)^\top- \boldsymbol{\Sigma}_{\phi\phi}\right)\right\|_{S_r}^r
+ \mathbb{E} \left\|
	 \mathbb{E}[(\ell_i - \pi) (\phi(X_i)\phi(X_i)^\top - \boldsymbol{\Sigma}_{\phi\phi}) | \Lambda_{i-1}] 
	\right\|_{S_r}^r
	\right] \\
\leq&\ 2^r \, \mathbb{E}\left\|(T_i-\pi)\left(\phi(X_i)\phi(X_i)^\top- \boldsymbol{\Sigma}_{\phi\phi}\right) \right\|_{S_r}^r \leq 2^r \, \mathbb{E}\left\|\phi(X_i)\phi(X_i)^\top- \boldsymbol{\Sigma}_{\phi\phi} \right\|_{S_r}^r = O(q^r),
\end{align*}
where the last equation follows from the fact that $\|\phi(X_i)\phi(X_i)^\top \|_{S_r} = \| \phi(X_i)\|^2$ and Minkowski's inequality:
\begin{align}\label{Eq_phi2} 
\left\{\mathbb{E}\left\|\phi(X_i)\phi(X_i)^\top- \boldsymbol{\Sigma}_{\phi\phi} \right\|_{S_r}^r\right\}^{1/r} &\leq \left\{\mathbb{E}\left[\left\|\phi(X) \right\|^{2r} \right] \right\}^{1/r} + \left\| \boldsymbol{\Sigma}_{\phi\phi}  \right\|_{S_r} \notag \\
&\leq \left\{\mathbb{E}\left[\left\|\phi(X) \right\|^{2r} \right] \right\}^{1/r} + q^{1/r}\left\| \boldsymbol{\Sigma}_{\phi\phi}  \right\|_{\mathrm{op}} = O(q).
\end{align}
Therefore, by Lemma \ref{Lemma_Burkholder}, 
\begin{align*}
\left\| \sum_{i=1}^n \Delta M_{i,\phi} \right\|_{L_r(S_r)} = O(n^{1/2}q^{1/2+1/r}) + O(n^{1/r}q).
\end{align*}
If condition (i) of Theorem \ref{Thm_ATE} holds, we set $r=2$; if condition (ii) of holds, we set 
\begin{align*}
r = \max\left\{2, \frac{1}{1-\gamma + \varepsilon}, \frac{2(\gamma-\varepsilon)}{1-\gamma+\varepsilon} \right\},
\end{align*}
Then,
\begin{align*}
\left\| \sum_{i=1}^n \Delta M_{i,\phi} \right\|_{L_r(S_r)} = o(n).
\end{align*}
By Markov's inequality and the fact that $\|B\|_{\mathrm{op}} \leq \| B\|_{S_r}$,
\begin{align*}
\left\|\frac{1}{n} \sum_{i=1}^n \Delta M_{i,\phi} \right\|_{\mathrm{op}} = o_P(1).
\end{align*}
For the second term, note that for any unit vector $u$,
\begin{align*}
&\ \left| u^\top \mathbb{E}\left[(\ell_{i+1} - \pi) (\phi(X_{i+1})\phi(X_{i+1})^\top - \boldsymbol{\Sigma}_{\phi\phi}) \middle| \Lambda_{i}\right] u \right| \\
\leq&\ \mathbb{E}\left[\left|\ell_{i+1} - \pi\right| \cdot \left|u^\top (\phi(X_{i+1})\phi(X_{i+1})^\top - \boldsymbol{\Sigma}_{\phi\phi}) u\right|\, \middle| \Lambda_{i}\right] \\
\leq &\ \frac{\overline{\lambda}_{\ell,1}}{i^{\gamma}} \mathbb{E}\left[\left| \Lambda_{i}^\top \phi(X_{i+1}) \right| \cdot \left|u^\top (\phi(X_{i+1})\phi(X_{i+1})^\top - \boldsymbol{\Sigma}_{\phi\phi}) u\right|\, \middle| \Lambda_{i}\right] \\
\leq&\ \frac{\overline{\lambda}_{\ell,1}}{i^{\gamma}} \left\{\mathbb{E}\left[\left| \Lambda_{i}^\top \phi(X_{i+1}) \right|^2\middle| \Lambda_{i} \right] \right\}^{1/2} \left\{ \mathbb{E}\left[\left(u^\top \phi(X)\right)^4 \right]  + (u^\top \boldsymbol{\Sigma}_{\phi\phi} u)^2 \right\}^{1/2} \\
\leq&\ \frac{\overline{\lambda}_{\ell,1}}{i^{\gamma}} M_2^{1/2} \left(M_4+\lambda_{\max}^2(\boldsymbol{\Sigma}_{\phi\phi})\right)^{1/2} \left\| \Lambda_{i}\right\|,
\end{align*}
which implies that
\begin{align*}
\left\| \frac{1}{n}\sum_{i=1}^{n-1} \mathbb{E}\left[(\ell_{i+1} - \pi) (\phi(X_{i+1})\phi(X_{i+1})^\top - \boldsymbol{\Sigma}_{\phi\phi}) \middle| \Lambda_{i}\right] \right\|_{\mathrm{op}} = O_P\left(\frac{1}{n} \sum_{i=1}^{n-1} \frac{\left\| \Lambda_{i}\right\|}{i^\gamma} \right) = o_P(1). 
\end{align*}
For the third term, by Equation (\ref{Eq_phi1}), 
\begin{align*}
\left\|\left\{\mathbb{E}\left[\left(\phi(X_i)\phi(X_i)^\top- \boldsymbol{\Sigma}_{\phi\phi}\right)^2\right]\right\}^{1/2} \right\|_{S_r} \leq q^{1/r} \left\| \mathbb{E}\left[(\phi(X_i)\phi(X_i)^\top - \boldsymbol{\Sigma}_{\phi\phi})^2\right]  \right\|_{\mathrm{op}}^{1/2} = O(q^{1/r+1/2}),
\end{align*}
and by Equation (\ref{Eq_phi2}), 
\begin{align*}
\mathbb{E}\left\| \phi(X_i)\phi(X_i)^\top- \boldsymbol{\Sigma}_{\phi\phi}  \right\|_{S_r}^r = O(q).
\end{align*}
Then, by Lemma \ref{Lemma_Burkholder},
\begin{align*}
&\ \left\| \frac{1}{n} \sum_{i=1}^n \left(\phi(X_i)\phi(X_i)^\top - \boldsymbol{\Sigma}_{\phi\phi}\right)  \right\|_{L_r(S_r)} \\
\leq&\ K_r \left[ n^{-1/2} \left\|\left\{\mathbb{E}\left[\left(\phi(X_i)\phi(X_i)^\top- \boldsymbol{\Sigma}_{\phi\phi}\right)^2\right]\right\}^{1/2} \right\|_{S_r} + n^{1/r-1} \mathbb{E}\left\| \phi(X_i)\phi(X_i)^\top- \boldsymbol{\Sigma}_{\phi\phi}  \right\|_{S_r}^r\right] \\
=&\ O\left(q^{1/r}\sqrt{\frac{q}{n}} + qn^{1/r -1}\right) = o(1),
\end{align*}
where the selection of $r$ is chosen as in the earlier part of the proof. \\
By Markov's inequality and the fact that $\|B\|_{\mathrm{op}} \leq \| B\|_{S_r}$,
\begin{align*}
\left\| \frac{1}{n} \sum_{i=1}^n \left(\phi(X_i)\phi(X_i)^\top - \boldsymbol{\Sigma}_{\phi\phi}\right)  \right\|_{\mathrm{op}} = o_P(1).
\end{align*}
Therefore,
\begin{align*}
\left\|\widehat{\boldsymbol{\Sigma}}_{\phi\phi,1} - \boldsymbol{\Sigma}_{\phi\phi} \right\|_{\mathrm{op}} = o_P(1).
\end{align*}
Furthermore, 
\begin{align*}
\left\| \widehat{\boldsymbol{\Sigma}}_{\phi\phi,1}^{-1} - \boldsymbol{\Sigma}_{\phi\phi}^{-1} \right\|_{\mathrm{op}}
&= \left\| -\widehat{\boldsymbol{\Sigma}}_{\phi\phi,1}^{-1}
(\widehat{\boldsymbol{\Sigma}}_{\phi\phi,1}-\boldsymbol{\Sigma}_{\phi\phi})
\boldsymbol{\Sigma}_{\phi\phi}^{-1}\right\|_{\mathrm{op}} \\
&\leq \left\| \widehat{\boldsymbol{\Sigma}}_{\phi\phi,1}^{-1} \right\|_{\mathrm{op}}  \left\|\widehat{\boldsymbol{\Sigma}}_{\phi\phi,1} - \boldsymbol{\Sigma}_{\phi\phi} \right\|_{\mathrm{op}} \left\| \boldsymbol{\Sigma}_{\phi\phi}^{-1} \right\|_{\mathrm{op}} = o_P(1).
\end{align*}

\end{proof}

\begin{lemma}\label{lemma_beta}
 Under the conditions of Theorem \ref{Thm_ATE_test}, let $\beta_{\phi}(a) = \boldsymbol{\Sigma}_{\phi\phi}^{-1}g_a$, where $g_a = \mathbb{E}[\phi(X)\{Y(a) - \mathbb{E}[Y(a)] \}]$ for $a=0,1$. Then,
	\begin{align*}
		\left\| \widehat{\beta}_{\phi}(a) - \beta_{\phi}(a) \right\| = o_P(1), a=0,1.
	\end{align*}
\end{lemma}
\begin{proof}
Let $\widehat{g}_a = n_a^{-1} \sum_{i=1}^n \mathbb{I}\{T_i=a\} \phi(X_i)(Y_i - \bar{Y}_a) $ and $g_i(a) = \phi(X_i)(Y_i(a) - \mathbb{E}[Y(a)])$ for $a=0,1$. We first prove that $\| \widehat{g}_a - g_a \|= o_P(1)$ and it suffices to prove the result for $a=1$, as the case $a=0$ can be established by the same argument. Consider the following decomposition, 
\begin{align*}
\widehat{g}_1 -g_1 =  \frac{1}{n_1}\sum_{i=1}^n T_i\left(g_i(1)- g_1  \right)  - (\bar{Y}_1 - \mathbb{E}[Y(1)]) \left(\frac{1}{n_1} \sum_{i=1}^n T_i\phi(X_i) \right). 
\end{align*}
Note that
\begin{align*}
\mathbb{E}\left[\left\| g_i(1) - g_1 \right\|^2 \right] \leq&\  \mathbb{E}\left[\| \phi(X) \|^2 \left\{ Y(1) - \mathbb{E}[Y(1)] \right\}^2\right] \\
\leq&\ \left\{\mathbb{E}\left[\| \phi(X) \|^4\right]  \right\}^{1/2} \mathbb{E}\left[\left\{ Y(1) - \mathbb{E}[Y(1)] \right\}^4 \right]^{1/2} = O(q). 
\end{align*}
and for every pair of unit vectors $u$ and $v$, 
\begin{align*}
&\ \mathbb{E}\left[\left| u^\top (g_i(1)-g_1) \right| \cdot \left|v^\top \phi(X) \right|  \right] \\
\leq&\ \mathbb{E}\left[\left| Y(a) - \mathbb{E}[Y(a)] \right| \cdot \left| u^\top \phi(X) \right| \cdot \left| v^\top \phi(X) \right|   \right] + \left| u^\top g_1 \right| \mathbb{E}\left| v^\top \phi(X)\right| \\
\leq&\ \left(\operatorname{Var}\{Y(a)\}\right)^{1/2} \left(\mathbb{E}\left| u^\top \phi(X) \right|^4\right)^{1/2} \left(\mathbb{E}\left| v^\top \phi(X) \right|^4\right)^{1/2} =O(1).
\end{align*}
Then, by Lemma \ref{Lemma_vector},
\begin{align*}
\left\| \frac{1}{n}\sum_{i=1}^n T_i \left(g_i(1)- g_1  \right) \right\| = O_P\left(\sqrt{\frac{q}{n}}\right)+o_P(1) = o_P(1).  
\end{align*}
In addition, by the proof of Theorem \ref{Thm_ATE},
\begin{align*}
\left| \bar{Y}_1 - \mathbb{E}[Y(1)]\right| = O_P(n^{-1/2}),
\end{align*}
and by Lemma \ref{Lemma_vector},
\begin{align*}
\left\| \frac{1}{n_1} \sum_{i=1}^n T_i\phi(X_i) \right\| \leq \left\| \frac{1}{n_1} \sum_{i=1}^n T_i\phi(X_i) - \mathbb{E}\left[\phi(X) \right] \right\| + \mathbb{E}\left[\left\|\phi(X) \right\|\right] = O_P(\sqrt{q}).
\end{align*}
Consequently,
\begin{align*}
\left\| \widehat{g}_1 -g_1 \right\| \leq  \left\| \frac{1}{n_1}\sum_{i=1}^n T_i\left(g_i(1)- g_1  \right)\right\|  + \left| \bar{Y}_1 - \mathbb{E}[Y(1)]\right| \left\|\frac{1}{n_1} \sum_{i=1}^n T_i\phi(X_i) \right\| = o_P(1). 
\end{align*}
A direct calculation yields that
\begin{align*}
\widehat{\beta}_{\phi}(a) - \beta_{\phi}(a) = \widehat{\boldsymbol{\Sigma}}_{\phi\phi,a}^{-1} \left(\widehat{g}_a - g_a\right) + \left(\widehat{\boldsymbol{\Sigma}}_{\phi\phi,a}^{-1} - \boldsymbol{\Sigma}_{\phi\phi}^{-1}\right) g_a.
\end{align*}	
Therefore, by Lemma \ref{Lemma_matrix}, 
\begin{align*}
\left\| \widehat{\beta}_{\phi}(a) - \beta_{\phi}(a) \right\| \leq \left\| \widehat{\boldsymbol{\Sigma}}_{\phi\phi,a}^{-1} \right\|_{\mathrm{op}} \left\|\widehat{g}_a - g_a\right\| + \left\| \widehat{\boldsymbol{\Sigma}}_{\phi\phi,a}^{-1} - \boldsymbol{\Sigma}_{\phi\phi}^{-1} \right\|_{\mathrm{op}} \left\|g_a \right\| 
= o_P(1). 
\end{align*}
\end{proof}


\begin{proof}[Proof of Theorem \ref{Thm_ATE_test}:] 
	Let $\widehat{s}_1^2 = n_1^{-1} \sum_{i=1}^n T_i\{Y_i - \bar{Y}_1\}^2$ denote the sample variance in the treatment group. Then, a direct calculation gives 
	\begin{align*}
		&\ \frac{1}{n_1} \sum_{i=1}^n T_i \widehat{e}_{i,1}^2 = \frac{1}{n_1} \sum_{i=1}^n T_i \left( Y_i-\Bar{Y}_{1} - \phi(X_i)^\top \widehat{\beta}_{\phi}(1) \right)^2 \\
		=&\ \widehat{s}_1^2 - 2 \widehat{g}_1^\top \widehat{\beta}_{\phi}(1)  + \widehat{\beta}_{\phi}(1)^\top \widehat{\boldsymbol{\Sigma}}_{\phi\phi,1} \widehat{\beta}_{\phi}(1) = \widehat{s}_1^2 - \widehat{g}_1^\top \widehat{\beta}_{\phi}(1) ,
	\end{align*}
and $\widetilde{\sigma}_{\widetilde{Y}(1)}^2 = \operatorname{Var}\{Y(1)\} - g_1^\top \beta_{\phi}(1) $. \\
By Theorem \ref{Thm_additional covariate} (i) and the law of large numbers, 
\begin{align*}
\widehat{s}_1^2 =&\ \frac{1}{n_1} \sum_{i=1}^n T_i Y_i^2(1) - \left(\frac{1}{n_1} \sum_{i=1}^n T_i Y_i(1)\right)^2 \\
=&\ \frac{n}{n_1} \left(\frac{1}{n} \sum_{i=1}^n  (T_i-\pi) Y_i^2(1) + \frac{\pi}{n} \sum_{i=1}^n Y_i^2(1) \right) - \left\{\frac{n}{n_1} \left(\frac{1}{n} \sum_{i=1}^n  (T_i-\pi) Y_i(1) + \frac{\pi}{n} \sum_{i=1}^n Y_i(1) \right)  \right\}^2 \\
\xrightarrow{P}&\ \mathbb{E}\left[Y^2(1)\right] - \mathbb{E}\left[Y(1)\right]^2 = \operatorname{Var}\{Y(1)\}. 
\end{align*}
In addition, consider the following decomposition,
\begin{align*}
\left| \widehat{g}_1^\top \widehat{\beta}_{\phi}(1)  - g_1^\top \beta_{\phi}(1)  \right| \leq \left\| \widehat{g}_1 - g_1 \right\| \left\| \widehat{\beta}_{\phi}(1) \right\| + \left\| g_1 \right\| \left\| \widehat{\beta}_{\phi}(1) - \beta_{\phi}(1) \right\|.
\end{align*}
Note that 
\begin{align*}
\mathbb{E}[Y^2(1)] \geq \mathbb{E}[(\phi(X)^\top\beta_{\phi}(1))^2] = \beta_{\phi}(1)^\top \boldsymbol{\Sigma}_{\phi\phi} \beta_{\phi}(1) \geq \lambda_{\min}\left(\boldsymbol{\Sigma}_{\phi\phi}\right)\left\|\beta_{\phi}(1)\right\|^2,
\end{align*}
which implies that $\|\beta_{\phi}(1) \| = O(1)$ and
\begin{align*}
\left\|g_a \right\| \leq \left\| \boldsymbol{\Sigma}_{\phi\phi} \right\|_{\mathrm{op}} \left\|\beta_{\phi}(1) \right\| = O(1). 
\end{align*}
Then, by Lemma \ref{lemma_beta}, 
\begin{align*}
\left| \widehat{g}_1^\top \widehat{\beta}_{\phi}(1)  - g_1^\top \beta_{\phi}(1)  \right| = o_P(1). 
\end{align*}
Consequently, 
\begin{align*}
\left| \frac{1}{n_1-q} \sum_{i=1}^n T_i \widehat{e}_{i,1}^2 - \widetilde{\sigma}_{\widetilde{Y}(1)}^2 \right| \xrightarrow{P} 0.
\end{align*}
Similarly, we can prove that 
\begin{align*}
\left|\frac{1}{n_0-q} \sum_{i=1}^n (1-T_i) \widehat{e}_{i,0}^2 - \widetilde{\sigma}_{\widetilde{Y}(0)}^2\right| \xrightarrow{P} 0,
\end{align*}
which further implies that
\begin{align}\label{Proof_ATE_CI1}
\left| \widehat{\varsigma}_{\widetilde{Y}}^2(\pi) - \varsigma_{\widetilde{Y}}^2(\pi) \right| \xrightarrow{P} 0.
\end{align}
	Write $\boldsymbol{\widehat{\Sigma}}_{\phi\phi} = n^{-1}\sum_{i=1}^n \phi(X_i)\phi(X_i)^\top$. Then, $\|\boldsymbol{\widehat{\Sigma}}_{\phi\phi}\|_{\mathrm{op}} =O_P(1)$. In addition, for convenience, we write $\widehat{\beta}_{\phi}=\widehat{\beta}_{\phi}(1) - \widehat{\beta}_{\phi}(0)$ and $\beta_{\phi}=\beta_{\phi}(1) - \beta_{\phi}(0)$. Then, $\|\widehat{\beta}_{\phi} - \beta_{\phi} \| = o_P(1)$ and 
	\begin{align*}
		&\ \widehat{\varsigma}_{H\widetilde{Y}}^2 = \frac{1}{n} \sum_{i=1}^n \left(\phi(X_i)^\top \widehat{\beta}_{\phi}   \right)^2 = \widehat{\beta}_{\phi}^\top \boldsymbol{\widehat{\Sigma}}_{\phi\phi}\widehat{\beta}_{\phi} \\
		=&\ \beta_{\phi}^\top \boldsymbol{\widehat{\Sigma}}_{\phi\phi}\beta_{\phi} + 2\beta_{\phi}^\top \boldsymbol{\widehat{\Sigma}}_{\phi\phi}\left(\widehat{\beta}_{\phi}-\beta_{\phi} \right) + \left(\widehat{\beta}_{\phi}-\beta_{\phi} \right)^\top\boldsymbol{\widehat{\Sigma}}_{\phi\phi}\left(\widehat{\beta}_{\phi}-\beta_{\phi} \right)\\
		=&\ \frac{1}{n} \sum_{i=1}^n \left(\phi(X_i)^\top \beta_{\phi}   \right)^2 +  o_P(1),
	\end{align*}
	where the last equality follows from 
	\begin{align*}
		&\left|\beta_{\phi}^\top \boldsymbol{\widehat{\Sigma}}_{\phi\phi}\left(\widehat{\beta}_{\phi}-\beta_{\phi} \right) \right| \leq \left\|\beta_{\phi} \right\| \left\| \boldsymbol{\widehat{\Sigma}}_{\phi\phi} \right\|_{\mathrm{op}} \left\| \widehat{\beta}_{\phi}-\beta_{\phi}  \right\| = o_P(1), \\
		&\left(\widehat{\beta}_{\phi}-\beta_{\phi} \right)^\top\boldsymbol{\widehat{\Sigma}}_{\phi\phi}\left(\widehat{\beta}_{\phi}-\beta_{\phi} \right) \leq \left\| \boldsymbol{\widehat{\Sigma}}_{\phi\phi} \right\|_{\mathrm{op}} \left\| \widehat{\beta}_{\phi}-\beta_{\phi}  \right\|^2 = o_P(1).
	\end{align*}
Therefore, by the law of large numbers, 
\begin{align}\label{Proof_ATE_CI2}
\left| \widehat{\varsigma}_{H\widetilde{Y}}^2 - \varsigma_{H\widetilde{Y}}^2 \right| \xrightarrow{P} 0. 
\end{align}
Finally, by (\ref{Proof_ATE_CI1}), (\ref{Proof_ATE_CI2}), Theorem \ref{Thm_ATE} and Slutsky's theorem, 
\begin{align*}
\sqrt{n}\left(\widehat{\tau} - \tau\right)\Big/ \sqrt{\widehat{\varsigma}_{\widetilde{Y}}^2(\pi) + \widehat{\varsigma}_{H\widetilde{Y}}^2} \xrightarrow{d} \mathcal{N}(0,1). 
\end{align*}
\end{proof}

\section{Additional Simulation Results}
Following the simulation settings in the main paper, we further examine the scaling behavior under \(\pi=2/3\). The results are reported in Figures \ref{fig4} and \ref{fig5}. Similar to the results in the main text, the curves in Figures \ref{fig4} and \ref{fig5} are nearly linear on the log–log scale. Following the same presentation as in the main text, we further plot the estimated slopes as a function of \(d\). The results are shown in Figure \ref{fig6}. The results under \(\pi=2/3\) are consistent with those in the main text, and we therefore omit further discussion. 

\begin{figure}[htbp]
	\centering
	\includegraphics[width=\textwidth]{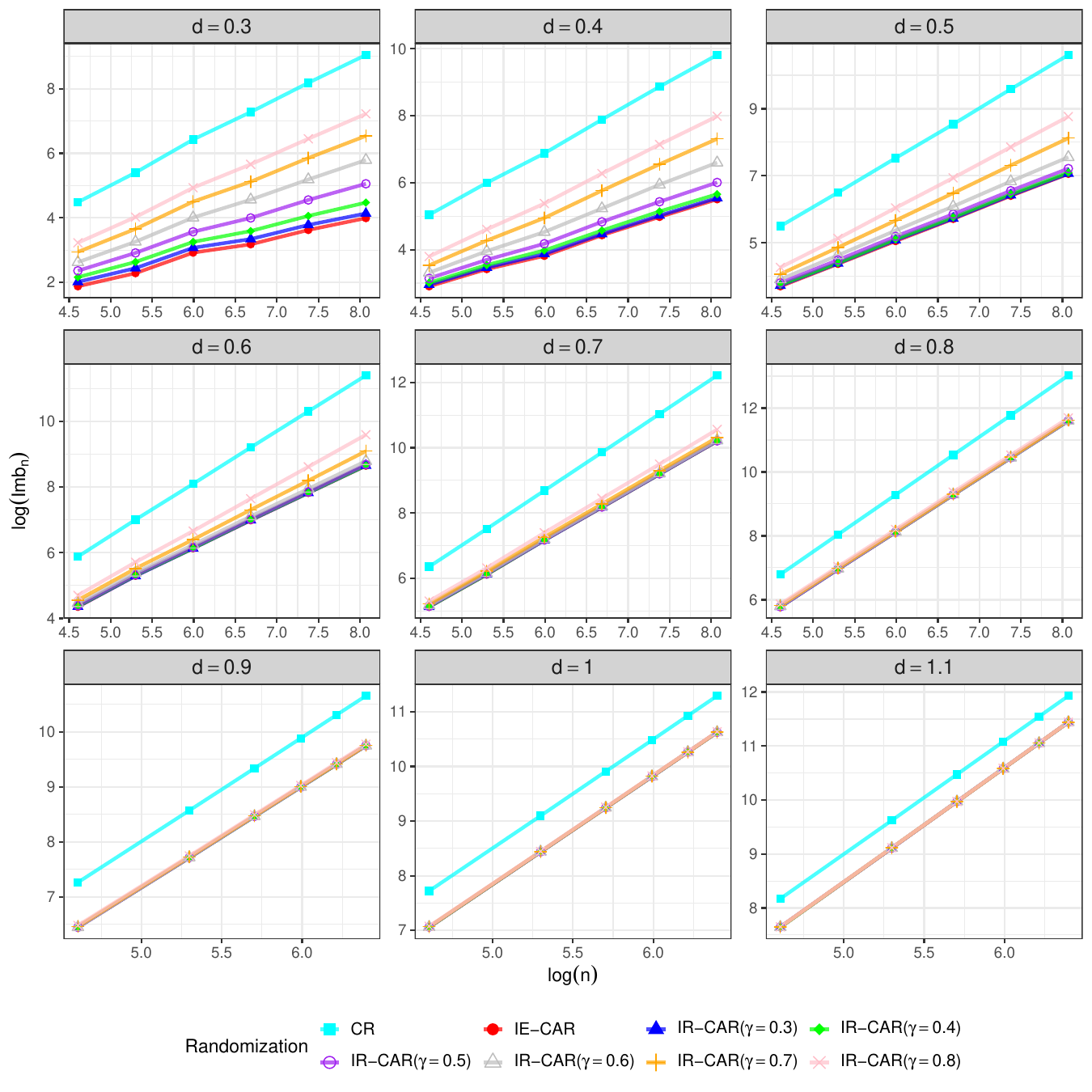}
	\caption{Log–log plot of $\text{Imb}_n$ versus $n$ with continuous covariates under $\pi=2/3$.}
	\label{fig4}
\end{figure}

\begin{figure}[htbp]
	\centering
	\includegraphics[width=\textwidth]{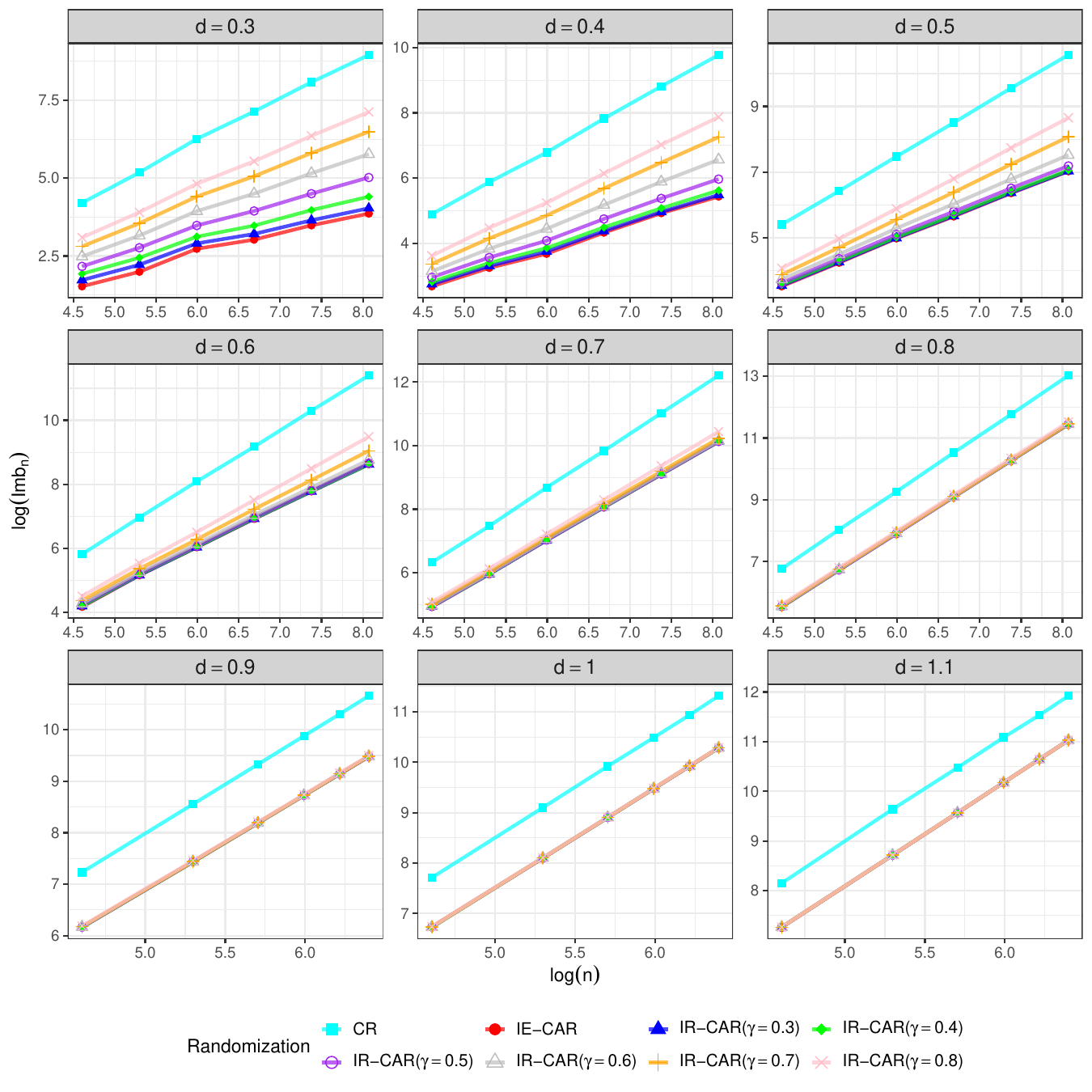}
	\caption{Log–log plot of $\text{Imb}_n$ versus $n$ with discrete  covariates under $\pi=2/3$.}
	\label{fig5}
\end{figure}

\begin{figure}[htbp]
	\centering
	\includegraphics[width=\textwidth]{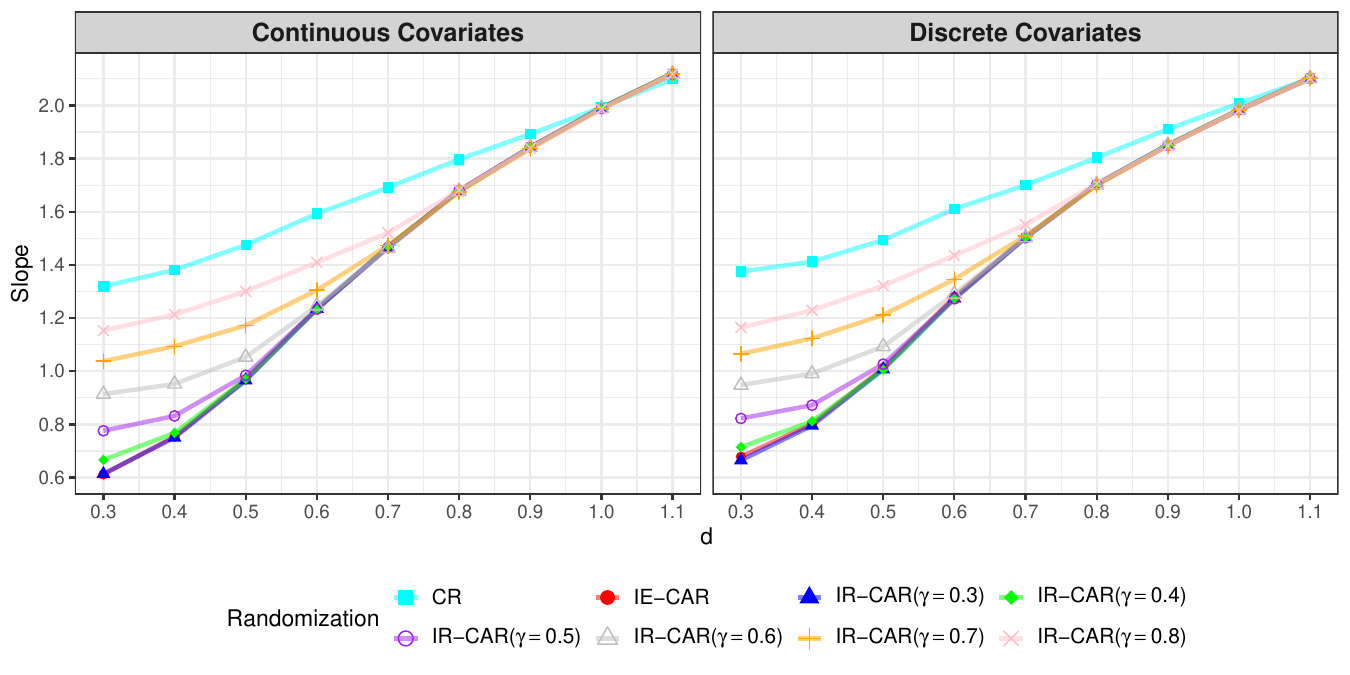}
	\caption{Estimated slopes from log–log plots under various randomization procedures with $\pi=2/3$.}
	\label{fig6}
\end{figure}

\section{Additional Details for Real Data Example}
Table \ref{covariates_realdata} provides a summary of the baseline covariates used in our analysis, including their types, definitions, and units.
\begin{table}[ht]
    \centering
    \setlength{\tabcolsep}{3pt} 
    \renewcommand{\arraystretch}{1}
    \caption{Description of covariates in PBC dataset.}
    \begin{tabular}{llll}
	\toprule
		Variable & Type & Description & Unit / Coding \\
	\midrule
		age & Continuous & Age at baseline & years \\
		sex & Discrete & Sex & 0 = female, 1 = male \\
		ascites & Discrete & Presence of ascites & 0 = no, 1 = yes \\
		hepatomegaly & Discrete & Presence of hepatomegaly & 0 = no, 1 = yes \\
		spiders & Discrete & Presence of spider angiomas & 0 = no, 1 = yes \\
		edema & Discrete & Edema status & 0 = no, 0.5 = mild, 1 = severe \\
		stage & Discrete & Histologic disease stage & 1, 2, 3, 4 \\
		bilirubin & Continuous & Serum bilirubin & mg/dL \\
		cholesterol & Continuous & Serum cholesterol & mg/dL \\
		albumin & Continuous & Serum albumin & g/dL \\
		copper & Continuous & Urine copper & $\mu$g/day \\
		alk.phos & Continuous & Alkaline phosphatase & U/L \\
		ast & Continuous & Aspartate aminotransferase & U/L \\
		triglycerides & Continuous & Triglycerides & mg/dL \\
		platelet & Continuous & Platelet count & $10^3$/mm$^3$ \\
		prothrombin & Continuous & Prothrombin time & seconds \\
	\bottomrule
	\end{tabular}

    \label{covariates_realdata}
\end{table}

\end{document}